\documentclass[a4paper,10pt,reqno]{amsart}
\usepackage{amsfonts}
\usepackage{amsmath, amsfonts, amsthm, amssymb, amscd, mathrsfs}
\usepackage{enumerate}
\usepackage{color}
\usepackage[normalem]{ulem}
\theoremstyle{plain}
\newtheorem{Theorem}{Main Theorem}

\newtheorem*{limitequation criterion}{Limit Equation Criterion}

\newtheorem{theorem}{Theorem}[section]
\newtheorem{proposition}[theorem]{Proposition}

\newtheorem{lemma}[theorem]{Lemma}

\theoremstyle{definition}
\newtheorem{definition}[theorem]{Definition}
\newtheorem{remark}[theorem]{Remark}

\numberwithin{equation}{section}

\newcommand{\rstar}{r_\star}
\newcommand{\rhot}{\widetilde \rho}
\newcommand \Hess {\mathrm{Hess}}
\newcommand \Scal {\mathrm{Scal}}
\newcommand \dM {\mathrm{diam}(M)}
\newcommand{\sik}{\mathrm{s}_\kappa}
\newcommand{\cok}{\mathrm{c}_\kappa}
\newcommand{\dr}{\partial r}
\newcommand{\Fcal}{\mathcal F}
\newcommand{\rhohat}{\widehat \rho}
\newcommand{\psihat}{\widehat \psi}
\newcommand{\Div}{\mathrm{div}}
\newcommand{\madm}{m_\mathrm{ADM}}
\newcommand{\Mgk}{\Mcal_{(g,k)}}
\newcommand{\Hgk}{H_{(g,k)}}
\newcommand{\eps}{\varepsilon}
\newcommand{\ppp}{{\prime\prime\prime}}
\newcommand{\Tt}{\widetilde T}
\newcommand{\rt}{\widetilde r}
\newcommand{\Ncal}{\mathcal N}
\newcommand{\Mcal}{\mathcal M}
\newcommand{\tr}{\mathrm{tr}}
\newcommand{\xit}{\widetilde \xi}

\newcommand{\varphit}{\widetilde \varphi}

\newcommand{\p}{\prime}
\newcommand{\pp}{{\prime\prime}}

\newcommand{\del}{\partial}
\newcommand{\ghat}{\widehat{g}}
\newcommand{\khat}{\widehat k}
\newcommand{\Vhat}{\widehat V}

\newcommand{\Ric}{\mathrm{Ric}}
\newcommand{\RR}{\mathbb R}
\newcommand{\HH}{\mathbb H}
\newcommand{\Sp}{\mathbb S}

\newcommand{\oo}[1]{(#1)}
\newcommand{\fo}[1]{\mathopen{[}#1)}

\newcommand{\ff}[1]{\mathopen{[}#1\mathclose{]}}
\begin{document}
\title[Spherically symmetric solutions to the constraint equations]{Spherically symmetric solutions to the Einstein-scalar field conformal constraint equations}
\author{Philippe Castillon and Cang Nguyen-The }
\date{}


\begin{abstract}
	Recent works \cite{DiltsHolstKozarevaMaxwell, NguyenNonexistence} have shown that the Einstein–scalar field conformal constraint equations are highly complex and generally intractable, even in the vacuum case. In this article, to gain a clearer understanding and offer a new perspective, we study these equations under special assumptions: the manifold $(M,g)$ is harmonic and all data are radial.
	
	\medskip
	
	In this setting, the system reduces to a single nonlinear equation and is completely solved in the standard cases. In particular, on the sphere, our results reveal phenomena that contrast with the well-known achievements on compact manifolds without conformal Killing vector fields, including nonexistence of solutions in the near-CMC regime and instability when the mean curvature is non-constant. By contrast, on Euclidean or hyperbolic manifolds, the equations are always solvable, with all expected properties of solutions satisfied. These findings support the view that, although the conformal method appears to present some drawbacks on compact manifolds, it remains a promising tool for parametrizing solutions to the constraint equations on asymptotically flat and hyperbolic manifolds in arbitrary mean curvature regimes.
	
	\medskip
	
	In this article, we also investigate the sign of mass, showing that the ADM and asymptotically hyperbolic mass of vacuume constraint solutions can take arbitrary sign when the decay rate of symmetric $(0,2)$-tensor $k$ at infinity is critical. Finally, most solution classes in our framework are explicit, providing a variety of models in general relativity and offering insights into the behavior of initial data, particularly in numerical applications.
\end{abstract}

\maketitle

\tableofcontents
%
%
%
\section{Introduction}
%
%
A scalar field initial data set for the Cauchy problem in general relativity is a manifold $(M,\ghat)$ of dimensions $n\ge 3$, equipped with a symmetric $(0,2)$-tensor $\khat$, a potential function $\Vhat$ and two scalar functions $(\psihat, \rhohat)$  such that  $(M, \ghat, \khat, \Vhat, \psihat, \rhohat)$ satisfies the system
\begin{subequations} \label{constraint}
	\begin{align}
		\Scal_{\ghat} - | \khat |_{\ghat}^2 +(\tr_{\ghat} \khat )^2 &=  \rhohat^2 + |d \psihat|_{\ghat}^2 + 2\Vhat(\psihat) &&\text{{\footnotesize [Hamiltonian constraint}]}
		\label{hamilton}
		\\
		\Div_{\ghat} \big( \khat - (\tr_{\ghat} \khat ) \ghat \big) & = \rhohat d\psihat, &&\text{[{\footnotesize Momentum constraint}]}
		\label{momentum}
	\end{align}
\end{subequations}
where $\Scal_{\ghat}$ is the scalar curvature of $\ghat$. One of the most efficient approaches to construct these initial data sets is the conformal method introduced by A. Lichnerowicz \cite{Lichnerowicz} and by Y. Choquet-Bruhat and J.W. York \cite{BruhatYork}.  The main idea of this method is to effectively parameterize the solutions to the constraint equations by some reasonable parts of the data set, and then to solve for them with the remaining of the data as unknown. More precisely, let $(M,g)$ be a Riemannian manifold of dimension $n$, $V$ be a potential function, $(\tau, \psi , \rho)$ be scalar functions, and $\sigma$ be a trace-free and divergence-free symmetric $(0,2)$-tensor on $(M,g)$.  One is required to find a positive function $\varphi$ and a $1-$form $W$ satisfying
\begin{subequations}\label{CE}
	\begin{align}
		- \tfrac{4(n-1)}{n-2} \Delta \varphi + \mathcal{R}_\psi \varphi  
		&= \mathcal{B}_{\tau, \psi} \varphi^{N-1}  + \frac{|\sigma + LW|^2 + \rho^2}{\varphi^{N + 1}}  &&[\text{{\footnotesize Lichnerowicz equation}}]
		\label{Liceq}
		\\
		\Div(L W) &= \tfrac{n-1}{n} \varphi^N d\tau - \rho d \psi,  &&[\text{{\footnotesize vector equations}}]
		\label{veceq}
	\end{align}
\end{subequations}
where 
$$
\mathcal{R}_\psi = \Scal_g - |d \psi|_g^2, \qquad \qquad \mathcal{B}_{\tau, \psi} = -\tfrac{n-1}{n} \tau^2 + 2 V (\psi),
$$
$\Scal_g$ is the scalar curvature of $g$, $N = \frac{2n}{n-2}$ and $L$ is the conformal killing operator defined by
\begin{equation} \label{killing}
	(LW)_{ij} = \nabla_i W_j + \nabla_j W_i - \tfrac{2 g_{ij}}{n} (\Div W).
\end{equation}
These equations are called the \textit{Einstein-scalar field conformal constraint equations} or simply the \textit{conformal equations}. Once $(\varphi, W)$ is a solution to \eqref{CE}, it follows that 
\begin{equation} \label{parametre}
	(\ghat, \, \khat, \, \Vhat, \, \psihat, \, \rhohat) = 
	\bigg(\varphi^{N - 2} g\, , \, \tfrac{\tau}{n} \varphi^{N - 2}g + \varphi^{-2} (\sigma + LW )\, , \, V\, , \, \psi\, , \, \varphi^{-N} \rho \bigg)
\end{equation}
will be a solution to the scalar field constraint  \eqref{constraint}. 

\medskip

Over the past few decades, the conformal method has been an important tool in studying the Einstein constraint equations, drawing considerable attention from many authors. While the method has proven effective under certain small-data assumptions-- such as when quantities like $\frac{d\tau}{\tau}$ or $(\sigma, \rho)$ are small-- its behavior under general data continues to be elusive, with fundamental questions about the existence of solutions and related problems remaining largely unanswered.

\medskip

A notable step in this direction was made in \cite{NguyenNonexistence, NguyenProgress}, where, using the limit equation technique developed by M. Dahl, R. Gicquaud and E. Humbert \cite{DahlGicquaudHumbert}, the second author showed that there exists a certain class of large data on compact manifolds for which the vacuum conformal equations \eqref{CE} admit no and multiple solutions. Complementing this, numerical studies by J. Dilts, M. Holst, T.~ Kozareva and D. Maxwell \cite{DiltsHolstKozarevaMaxwell} uncovered striking instabilities in solutions to \eqref{CE} in the vacuum case. Together, these results indicated the complexity of the conformal equations as well as raising unexpected concerns about the reliability and applicability of the conformal method as a parameterization tool for generating solutions to the constraint \eqref{constraint}. This explains why progress in this area has been limited in recent years and why the conformal method appears to have seen reduced attention.

\medskip

The purpose of this article is to revisit this narrative and offer a new perspective on the study of the conformal equations \eqref{CE}. Since the system is mainly investigated on compact, asymptotically flat (AF), or asymptotically hyperbolic (AH) manifolds, and since the currently available techniques are insufficient to resolve \eqref{CE} in full generality, we restrict our attention to simpler settings in which $(M,g)$ is spherical, hyperbolic, or Euclidean, and all data are assumed to be radial. As will become clear, working in these configurations offers two main advantages. First, the analytic framework is preserved under symmetry, so arguments developed for the general conformal equations can apply directly in the radial setting. These simplified models therefore provide a convenient testing ground for evaluating whether a given method or property is promising for the full system \eqref{CE}. Second, the radial assumption enables a more effective treatment of the vector equations. In fact, the difficulty of the conformal equations is that while the Lichnerowicz equation is relatively well understood, we have few tools available to analyze the vector equations, the second part of the system. In radial settings, however, building on the work of the second author in \cite{NguyenRadialAF}, these vector equations can be solved entirely, with solutions computed explicitly. This allows the system to be reduced to a single nonlinear equation, making the analysis much more manageable.

\medskip
  
With this framework in place, many interesting results can be established. Before presenting them, however, we first introduce the notation used in the article. From now on, $r$ is the distance function to a fixed point and we denote by $f^\p$ the derivative of $f$ with respect to $r$. We remark the the unknown of the conformal equations is a couple $(\varphi,W)$ of a function and a 1-form. However, in some situations, the 1-form $W$ can be recovered once the function $\varphi$ is known. For simplicity of expression, we will sometimes we will sometime just say that $\varphi$  is a solution of \eqref{CE}. Now, let $M_\kappa$ be the space form of curvature $\kappa\in\{1,0,-1\}$, that is $M_\kappa=\Sp^n$ if $\kappa=1$, $M_\kappa=\RR^n$ if $\kappa=0$ and $M_\kappa=\HH^n$ if $\kappa=-1$. The sine and cosine function $\sik$ and $\cok$ are
\begin{equation} \label{def sik and cok}
	\big(\sik(r), \cok(r)) =\left\{\begin{array}{ll}
		\big(\sin(r), \cos(r) \big) & \mbox{if } \kappa=1 \\
		(r, 1) & \mbox{if }  \kappa=0 \\
		\big(\sinh(r), \cosh(r) \big) & \mbox{if }  \kappa=-1. \\
	\end{array}\right.
\end{equation}
Since we work on radial settings, it is natural to assume throughout the article that $\sigma \equiv 0$ and that the remaining data set  $$
(V,\psi(r),\rho(r),\tau(r)) \in C^0(\RR)\times C^1(M_\kappa)\times C^0(M_\kappa)\times C^1(M_\kappa)
$$ 
is radial. For such a data set and for a positive radial function $\varphi(r)\in C^3(M_\kappa)$,  we define new functions $I_{V,\psi,\rho,\varphi}(r)$ and $S_{V,\psi,\rho,\varphi}(r)$ by
\[
I_{V,\psi,\rho,\varphi} := \tfrac{n}{n-1} \Big( 2 V(\psi) \varphi^{2N} + |\psi^\p|^2 \varphi^{N+2} + \rho^2 \Big),
\]
\[
S_{V,\psi,\rho,\varphi} := N^2 (\varphi^\p)^2 - n^2\kappa \varphi^2 + 2 n N \frac{\cok}{\sik}  \varphi \varphi^\p + \frac{1}{\sik^n } \int_0^r I_{V,\psi,\rho,\varphi} \varphi^{-2N} (\varphi^N \sik^n)^\p  \, dt.
\]
Having fixed these notations, we can now present our first main result in this article, which provides a necessary and sufficient condition for a positive radial $\varphi$ to solve the conformal equations. For clarity of exposition, we state here only typical results and refer the reader to Sections \ref{section reduction}--\ref{section Euclidean} for the precise statements.
\begin{Theorem} \label{main theorem} 
	Assume that $\sigma \equiv 0$. Let $(V,\psi(r),\rho(r),\tau(r))\in C^0(\RR)\times C^1(M_\kappa)\times C^0(M_\kappa)\times C^1(M_\kappa)$ be a radial data set on the space form $M_\kappa$, with $\rho \psi^\p=0$. For a positive radial function $\varphi(r)\in C^3(M_\kappa)$, we have
	\begin{enumerate}[(i)]
		\item If $\varphi(r)$ is a solution to the conformal equations \eqref{CE}, then $S_{V,\psi,\rho,\varphi}\ge0$.
		
		\item If $S_{V,\psi,\rho,\varphi}>0$ and $(\varphi^N\sik^n)^\p\neq0$ a.e., then $\varphi(r)$ is a solution to the conformal equations \eqref{CE} if and only if the mean curvature $\tau$ satisfies
		\begin{equation} \label{main identity}
			\tau= \pm \frac{S_{V,\psi,\rho,\varphi} + 2N  \big( \varphi^\pp + (n - 1)\frac{\cok}{\sik} \varphi^\p \big) \varphi - \kappa n^2 \varphi^2 + I_{V,\psi,\rho,\varphi} \varphi^{-N}}{2\sqrt{\varphi^N S_{V,\psi,\rho,\varphi}}}
		\end{equation}
	\end{enumerate}
\end{Theorem}

This result is the source of subsequent developments in the article. In particular, point (i) provides a necessary condition for the existence of solutions, while point (ii) shows that constructing radial solutions to the conformal equations \eqref{CE} becomes quite simple. Indeed, it suffices to choose a potential function $V$ together with radial functions $(\psi(r), \rho(r), \varphi(r))$ such that $S_{V,\psi,\rho,\varphi} > 0$. Once this condition is satisfied, the identity \eqref{main identity} uniquely determines the corresponding mean curvature function $\tau(r)$, thereby completing the construction of a full data set $(V,\psi(r),\rho(r),\tau(r))$ along with a solution $\varphi(r)$ to the conformal equations \eqref{CE}. As shown in Appendix~\ref{appendix construct explicit solutions}, seeking data that satisfies this positivity condition is not difficult. Thus, our theorem yields a broad class of explicit solutions to the constraint equations, which is valuable not only as theoretical examples but also as practical models in numerical relativity, where exact solutions are scarce and highly informative.

\medskip

For the further results in this article, apart from the instability phenomena on the sphere described in Main Theorem~\ref{main theorem instability}, we are interested in the existence of solutions to the conformal equations \eqref{CE} in the case where $V$, $\psi$, and $\sigma$ vanish identically, that is
\begin{subequations}\label{CE with null sigma in sphere intro}
	\begin{align}
		- \tfrac{4(n-1)}{n-2} \Delta \varphi + \Scal_g \varphi  
		& = - \tfrac{n-1}{n} \tau^2 \varphi^{N-1}  + \frac{|LW|^2+\rho^2}{\varphi^{N + 1}} 
		\\
		\Div(L W) &= \tfrac{n-1}{n} \varphi^N d\tau.
	\end{align}
\end{subequations}
From an analytical standpoint, this system shares many structural similarities with the conformal equations in the vacuum case, a subject extensively studied in \cite{IsenbergMoncrief, HolstNagyTsogtgerel, Maxwellcompact, DahlGicquaudHumbert, NguyenFPT, NguyenNonexistence}. The main distinction here is the replacement of the TT-tensor $\sigma$ by the time-derivative $\rho$. In certain contexts, this substitution is motivated by the identity
$$
|\sigma + LW|^2 = |\sigma|^2 + 2 \langle \sigma, LW \rangle + |LW|^2,
$$
which implies that, when $\sigma$ and $LW$ are point-wise orthogonal, $\rho$ can be interpreted as the norm of $\sigma$. For concrete model problems illustrating this construction, we refer the reader to \cite{BeigBizonSimon} and \cite{Maxwell: Modelproblem, DiltsHolstKozarevaMaxwell}.

\medskip

By studying \eqref{CE with null sigma in sphere intro} in this context, we provide a new perspective that, contrary to unexpected results in \cite{NguyenNonexistence, DiltsHolstKozarevaMaxwell}, the conformal equations \eqref{CE with null sigma in sphere intro} in radial settings are solvable in certain standard spaces, particularly in the hyperbolic and Euclidean cases, with all expected properties of solutions satisfied. This give evidence that the conformal method is still an effective parameterization of solutions to the constraint equations \eqref{constraint} for a broad class of data, not only when the mean curvature is constant or almost constant. Since the geometrical structure and the related problems for the conformal equations differ on the sphere, Euclidean space, and hyperbolic space, for ease
of the presentation, we present the remaining discussion and results according to the type of manifolds as follows.

\subsection{Standard sphere.}

Let $\Sp^n$ denote the standard round sphere of dimension $n \ge 3$. Although the spherical model is basic, the existence of solutions to the conformal equations on $\Sp^n$ has, to date, only been established when the mean curvature function $\tau$ is constant — i.e., $d\tau \equiv 0$, which decouples the system \eqref{CE}. The reason for this limitation is that, so far, available techniques for proving the existence of solutions to \eqref{CE} have been based on constructing appropriate elliptic operators, which in turn require that the manifold admit no conformal Killing vector fields (CKVF). However, on the round sphere, such vector fields exist, which prevent the method from working in this case, particularly when $\tau$ is non-constant.

\medskip

In the first part of this work, we are concerned with the question of the existence of solutions to the conformal equations for a given radial setting on the sphere in various mean curvature regimes. The outcomes we obtain are different from those previously known on compact manifolds without CKVF. The first distinction arises from existence/nonexistence of solutions in the near-constant mean curvature
regime as we will show below. For manifolds having no CKVF, previous works such as \cite{IsenbergMoncrief, Maxwellcompact, DahlGicquaudHumbert} have shown that the reduced system \eqref{CE with null sigma in sphere intro} admits a unique solution $\varphi$
whenever $\rho \not\equiv 0$  and the ratio $\frac{d\tau}{\tau}$ is sufficiently small. Solutions in this case are the so-called near-CMC solutions, and their existence may be obtained through several
approaches. However, in the theorem below, we will see that this conclusion no longer holds when the manifold has CKVF, with the following setting on the round sphere as a counterexample.

\begin{Theorem}[Non-existence of solution for near--CMC] \label{main theorem nonexistence near CMC}
	Let $(\rho(r), \tau(r))$ be a radial setting on the round sphere $\Sp^n$ such that $\tau^\p,\rho\not\equiv0$ and $\tau^\p \ge 0$.
	
	Then, for any sufficiently large constant $C > 0$ depending only on $\tau$, the conformal equations \eqref{CE with null sigma in sphere intro} associated with $(\rho, \tau + C)$ admit no solution.
\end{Theorem}

Another notable difference of \eqref{CE with null sigma in sphere intro} on the sphere we would like to present concerns the existence of solutions to \eqref{CE with null sigma in sphere intro} when $\rho \equiv 0$, that is
\begin{subequations}\label{vacuum CE null sigma}
	\begin{align}
		- \tfrac{4(n-1)}{n-2} \Delta \varphi + \Scal_g \varphi  
		& = - \tfrac{n-1}{n} \tau^2 \varphi^{N-1}  + \frac{|LW|^2}{\varphi^{N + 1}} \\
		\Div(L W) & = \tfrac{n-1}{n} \varphi^N d\tau,
	\end{align}
\end{subequations}
which coincides with the conformal equations \eqref{CE} of null TT-tensor $\sigma$ in the vacuum case. The question of existence of solutions to these equations was first posed by D.~Maxwell in \cite{Maxwell: Modelproblem} and it was shown by the second author in \cite{NguyenNonexistence, NguyenProgress} to be central to
the understanding of the full vacuum case of \eqref{CE}. In the absence of CKVF, the second author also proved in \cite{NguyenNonexistence} that if $\tau > 0$ satisfies
\begin{equation} \label{non-generic condition intro}
	\bigg | L\bigg(\frac{d\tau}{\tau}\bigg) \bigg| \le c \bigg|\frac{d\tau}{\tau}\bigg|^2,
\end{equation}
then for any large constant $a>0$, the equations \eqref{vacuum CE null sigma}  associated with $\tau^a$ admit a non-trivial solution. Motivated by this analytic insight, J. Dilts, M. Holst, T.~ Kozareva and D. Maxwell \cite{DiltsHolstKozarevaMaxwell} performed numerical investigations on the standard product manifold $\Sp^1 \times \Sp^2$, suggesting that solutions to \eqref{vacuum CE null sigma} may still exist in the far-from-CMC regime, without requiring the non-generic condition~\eqref{non-generic condition intro}. In the next result, we reveal a different behavior, that is, \eqref{vacuum CE null sigma} never admit a radial solution on the sphere. This result does not contradict the earlier ones in \cite{NguyenNonexistence, NguyenProgress} or \cite{DiltsHolstKozarevaMaxwell}, but rather complements them by giving a key insight: whether or not solutions to \eqref{vacuum CE null sigma} exist depends not only on how far the mean curvature is from being constant, but also on the geometrical structure of the manifold $(M, g)$ itself. The statement of our result is as follows.
\begin{Theorem}[Nonexistence of radial solution for null TT-tensor]
	Let $\tau(r) \in C^1(\Sp^n)$ be a radial function on the round sphere. Then the equations \eqref{vacuum CE null sigma} have no radial solution. In particular, if a solutions to \eqref{vacuum CE null sigma} exists, the set of solutions is infinite.
\end{Theorem}
We next turn our attention to the issue of the existence of solutions to the conformal equations \eqref{CE with null sigma in sphere intro} on the sphere in relation to the mean curvature $\tau$.

\medskip

The result in Main Theorem \ref{main theorem nonexistence near CMC} already makes clear that not every choice of $\tau$ leads to a solution, even when $\tau$ is close to constant (i.e., in the near-CMC setting). In fact, as we will see in Proposition \ref{proposition resolve vector equations} below, under radial symmetry, a necessary condition for the existence of a radial solution $\varphi$ to \eqref{CE with null sigma in sphere intro} is the vanishing condition
\begin{equation}\label{vanishing condition intro}
	\int_0^\pi \tau^\p \varphi^N \sin^n\,ds = 0.
\end{equation}
In particular, since $\varphi$  is positive, this requires at least that $\tau^\p$ must change sign in $(0,\pi)$. The arguments we will use to establish the existence of solutions to  \eqref{CE with null sigma in sphere intro} in the article are partially built on earlier works, such as those by   J.~Isenberg and V. Moncrief \cite{IsenbergMoncrief}, by D. Maxwell \cite{Maxwellcompact}, and by M. Dahl, R. Gicquaud and E.~Humbert \cite{DahlGicquaudHumbert}. However, the main challenge in this case is to incorporate the techniques from these works, while ensuring that the vanishing condition \eqref{vanishing condition intro} is satisfied in the process. To address this difficulty, we reformulate the problem as follows. We first observe that if we decompose the mean curvature as
\begin{equation} \label{decompose tau}
	\tau(r) = \tau_0(r) + \alpha \tau_1(r)
\end{equation}
for some radial functions $\tau_0(r), \tau_1(r)$ and $\alpha \in \RR$, then as long as \eqref{CE with null sigma in sphere intro} have a radial solution $\varphi$,  it follows from \eqref{vanishing condition intro} that the constant $\alpha$ can be computed by
$$
\alpha = \frac{\int_0^\pi \tau_0^\p \varphi^N \sin^n\,ds}{\int_0^\pi \tau_1^\p \varphi^N \sin^n\,ds}.
$$
Therefore,  since every radial function can be rewritten in the form of \eqref{decompose tau}, given an arbitrary pair of radial functions $(\tau_0, \tau_1)$,  we may restrict mean curvature $\tau$ to the class 
$$
\mathcal{T}_{\tau_0, \tau_1} := \big\{ \tau_0 + \alpha \tau_1 ~|~ \alpha \in \RR \big\},
$$ 
and then, instead of searching for mean curvatures that make \eqref{CE with null sigma in sphere intro} solvable, we will look for values of $\alpha$ such that the conformal equations associated with $\tau_0(r) +\alpha \tau_1(r)$ admit a radial solution. Of course, as we have noted above in connection with Main Theorem \ref{main theorem nonexistence near CMC} and the vanishing condition \eqref{vanishing condition intro}, not every value of $\alpha$ will be admissible in this class. However, the following theorem ensures that at least one such $\alpha$ always exists, guaranteeing the solvability of the system. For easy comparison with well-known results on compact manifolds without CKVF, we present here three types of existence results, each based on different assumptions. The techniques we use in the first two are inspired by the small TT-tensor argument from \cite{HolstNagyTsogtgerel, Maxwellcompact} and the limit equation criterion from \cite{DahlGicquaudHumbert, NguyenFPT}. Accordingly, we will refer to them as the small time-derivative existence and the limit equation criterion existence, respectively. What is particularly interesting is the third case, which we call the regular time-derivative existence, where it turns out that the $C^1$-regularity of the time--derivative $\rho$ alone is sufficient to guarantee the existence of solutions. 
\begin{Theorem}[Solvability of the conformal equations] \label{main theorem solvability on sphere}
	Let $(\rho(r), \tau_0(r), \tau_1(r))$ be a radial setting on the round sphere $\Sp^n$ with $\rho \not \equiv 0$. Assume that $|\tau_0^\p| \le c \tau_1^\p$, for some constant $c>0$, and that $(\rho, \tau_0, \tau_1)$ satisfies one of the following conditions
	\begin{itemize}
		\item small time-derivative: $\max|\rho|$ is sufficiently small;
		\item limit equation criterion: $\tau_0 - a \tau_1$ does not change sign for any $a \in [-c,c]$;
		\item regular time-derivative: $\rho(r) \in C^1(\Sp^n)$.
	\end{itemize}
	Then there exists $\alpha \in [-c, c]$ such that the conformal equations \eqref{CE with null sigma in sphere intro} associated with $(\rho, \tau_0 - \alpha \tau_1)$ admit at least one radial solution $\varphi(r)>0$. Moreover, in this case, $\alpha$ can be computed by
	$$
	\alpha = \frac{\int_0^\pi \tau_0^\p \varphi^N \sin^n\,ds}{\int_0^\pi \tau_1^\p \varphi^N \sin^n\,ds}.
	$$
\end{Theorem}

The last results derived from Main Theorem \ref{main theorem} on the sphere we would like to present concern stability and instability of the full system \eqref{CE}. The (in)stability behavior of this system is closely tied to the sign of the quantity $\mathcal{B}_{\tau, \psi}$, which plays a central role in the analysis. A brief overview of the known results is as follows:
\begin{itemize}
	\item When $\mathcal{B}_{\tau,\psi} > 0$, which we call the \textit{focusing case}, the problem is completely solved by E. Hebey, O. Druet and B. Premoselli in \cite{DruetHebey, DruetPremoselli, Premoselli_Stability}, in which the equations are shown to be stable in low dimensions as long as $\rho \not\equiv 0$. For manifolds of high dimensions, the result is also true under an additional assumption of zero-points of derivatives of $\tau$ and $\mathcal{B}_{\tau,\psi}$, which is proven to be sharp in \cite{PremoselliWei}.
	
	\item When $\mathcal{B}_{\tau, \psi} \le 0$, that is the \textit{defocusing case}, the system is known to be stable in the CMC and near-CMC settings thanks to well-known arguments in \cite{BartnikIsenberg, Maxwellcompact, DahlGicquaudHumbert}. However, when $\tau$ is far--from--CMC, the stability question still remains open.
	
	\item When $\mathcal{B}_{\tau, \psi}$ changes sign, the situation becomes significantly more subtle and challenging. To the best of the authors' knowledge, the only known result in this setting is due to O. Druet and E. Hebey \cite{DruetHebey}, who establish stability in the CMC case for dimensions $n \leq 5$.
\end{itemize}
In this article, we investigate this problem specifically on the sphere. Our goal is to provide answers, or at least compelling evidence, regarding the stability/instability of the system \eqref{CE} in the last two cases.

\medskip

For $\mathcal{B}_{\tau, \psi} \le 0$, we focus on a simplified setting in which $(V,\psi,\sigma)$ vanish identically, that is the equations \eqref{CE with null sigma in sphere intro}. Using the condition $S_{V,\psi,\rho,\varphi} \ge 0$, we show that, within the space of radial functions, the equations are not only stable, but that the set of solutions is also uniformly bounded, independently of the choice of $\tau$. This suggests that the stability property of the conformal equations may persist more generally in this case. When $\mathcal{B}_{\tau, \psi}$ changes sign, the situation is markedly different. In contrast to the stability result obtained by O. Druet and E. Hebey \cite{DruetHebey} for the CMC case in dimensions $n\le 5$ as mentioned above, thanks to Identity \eqref{main identity}, we will construct a sequence of blowing-up solutions to \eqref{CE}, thereby showing that the stability property of the conformal equations no longer holds in general in the non-CMC setting. This gives a negative answer to the stability question in the last case in all dimensions. More precisely, our associated results are that
\begin{Theorem}[Non-positive $\mathcal{B}_{\tau, \psi}$ and stability] \label{main theorem instability}
	Let $\rho(r)\in C^1(\Sp^n)$ be a radial function. For any radial function $\tau(r)\in C^1(\Sp^n)$, we denote by $\mathcal{S}_{\Sp^n}(\tau)$ the set of all radial solutions $\varphi(r)$ to \eqref{CE with null sigma in sphere intro} associated with $(\rho, \tau)$. 
	
	Then there exists a constant $C>0$ depending only on $\rho$ such that
	$$
	\sup\bigg\{\|\varphi(r)\|_{C^0(\Sp^n)} \ \bigg|\ \exists \tau(r) \in C^1(\Sp^n) \text{ such that } \varphi(r) \in \mathcal{S}_{\Sp^n}(\tau)  \bigg\} \le C.
	$$
\end{Theorem}
\begin{Theorem}[Instability under sign change of $\mathcal{B}_{\tau, \psi}$]
	For any $\Lambda > 0$, there exists a sequence of $\{ (\varphi_i(r), \tau_i(r)) \}$ satisfying the conformal equations \eqref{CE} associated with the data set $(\psi, \rho, \sigma, V, \tau) = (0, 0, 0, \Lambda, \tau_i(r))$ such that $\tau_i(r)$ converges in $C^1(\Sp^n)$ to some non-constant $\tau_\infty(r)$, but $\|\varphi_i\|_{L^\infty} \to + \infty$.
\end{Theorem}

\subsection{Hyperbolic manifold}

Let us now present the applications of Main Theorem \ref{main theorem} to the conformal equations \eqref{CE with null sigma in sphere intro} on the hyperbolic manifold $\HH^n$. In contrast to the spherical case discussed earlier, it is clear that $\HH^n$ admits no CKVF. As a result, solutions to \eqref{CE with null sigma in sphere intro} are not subject to additional constraints such as the vanishing condition \eqref{vanishing condition intro}. Moreover, the absence of CKVF also gives a significant advantage: we can directly apply well-known results established for AH manifolds without needing to adjust any argument to fit our situation. Among these benefits, a particularly important result on AH manifolds that we will make use of is the limit equation criterion proven by R.~Gicquaud and A. Sakovich \cite{GicquaudSakovich}, which states that if $\tau$ does not change sign and if the limit equation 
\begin{equation} \label{limit equation intro}
	\Div(L W) = \lambda \sqrt{\tfrac{n-1}{n}} |LW| \frac{d\tau}{\tau}
\end{equation}
has no non-trivial solution $W$ for all $\lambda \in (0,1]$, then the conformal equations \eqref{CE with null sigma in sphere intro} admit at least one solution. In the context of radial settings, as we will see in the process of the proof of Main Theorem \ref{main theorem}, this limit equation \eqref{limit equation intro} may be reduced to a first-order linear ODE, which can be easily shown to have no nontrivial solution, therefore, the existence of solutions to \eqref{CE with null sigma in sphere intro} is obtained. More precisely, we achieve the following result.
\begin{Theorem}[Solvability and stability]
	Let $(\rho(r), \tau(r)) \in C^0(\HH^n)\times C^1(\HH^n)$ be a radial setting in the hyperbolic space $\HH^n$, and assume that $\tau$ does not change sign.
	
	Then the conformal equations \eqref{CE with null sigma in sphere intro} associated with $(\rho, \tau)$ admit at least one radial solution. Moreover, the set of solutions is compact.
\end{Theorem}
Another application of Main Theorem \ref{main theorem} in the hyperbolic space we are interested in is existence of constraint solutions with negative AH mass. In this article, we will provide an example of such a construction.  It is worth noting that the finding we obtain does not contradict the result of Sakovich \cite{SakovichMassAH}, nor the commonly stated version of the Positive Mass Theorem for AH manifolds. On the contrary, because the symmetric $(0,2)-$tensor $k$ in our construction decays at a critical rate relative to the assumptions of the conjecture, our example is intended to show that the decay rate of $k$ plays a central role in determining the sign of the mass, and that the decay assumptions in the conjecture are sharp.
\begin{Theorem}[Constraint solutions with negative mass] \label{main theorem mass hyperbolic}
	There exists a vacuum constraint solution $(\HH^n,g,k)$ such that $k = O(e^{-nr/2})$ at infinity and the AH-mass of $(\HH^n,g)$ is negative.
\end{Theorem}

\subsection{Euclidean manifold}
Let $\RR^n$ denote the Euclidean space of dimension $n\ge 3$. Unlike the spherical and hyperbolic cases, existence of solutions to the conformal equations in the Euclidean setting is only known in the CMC case (see \cite{Diltsthesis}). This restriction is due to several technical difficulties. So far, we have had only three approaches to establish the existence of solutions to the conformal equations in the non-CMC regime: the near-CMC argument, the small TT-tensor approach, and the limit equation criterion. Among these, the small TT-tensor approach in \cite{HolstNagyTsogtgerel, Maxwellcompact, NguyenFPT, GicquaudNguyen}, designed for compact manifolds, cannot be applied in the Euclidean setting since solutions $\varphi$ must approach $1$ at infinity and therefore cannot be made small, as required by the method. For the near-CMC argument, although a similar existence result is already proven for \eqref{CE} in AF manifolds in \cite{DiltsIsenbergMazzeoMeier}, this non-CMC result is not really convincing since the assumption that $\frac{d\tau}{\tau}$ is sufficiently small may not hold due to the weighted Poincaré inequality. Likewise, the limit equation criterion, developed in works such as \cite{DahlGicquaudHumbert, NguyenFPT, GicquaudSakovich}, is also inapplicable because it requires $\tau$ to be strictly positive, a condition that breaks down since $\tau$ vanishes at infinity, except for $\tau$ identically zero.

\medskip

In this article, by restricting to radial settings, we do not encounter these limitations, as we can apply Main Theorem~\ref{main theorem} to solve \eqref{CE with null sigma in sphere intro} on the Euclidean space. Our main result is that, for any radial data set, the equations are always solvable, and moreover, the set of solutions is uniformly bounded independently of $\tau$.
\begin{Theorem}[Solvability and stability]
	Let $(\rho(r), \tau(r)) \in C^0(\RR^n)\times C^1(\RR^n)$ be a radial setting in the Euclidean space $\RR^n$ and assume that $\rho$ is $C^1$.
	
	Then the conformal equations \eqref{CE with null sigma in sphere intro} associated with $(\rho, \tau)$ admit at least one radial solution. Moreover, if for any radial $\tau$ we denote by $\mathcal{S}_{\RR^n}(\tau)$ the set of all radial solutions $\varphi(r)$ to \eqref{CE with null sigma in sphere intro} associated with $(\rho, \tau)$, then there exists a constant $C>0$ depending only on $\rho$ such that
	$$
	\sup\bigg\{\|\varphi(r)\|_{C^0(\RR^n)} \ \bigg|\ \exists \tau(r) \in C^1(\RR^n) \text{ such that } \varphi(r) \in \mathcal{S}_{\RR^n}(\tau)  \bigg\} \le C.
	$$
\end{Theorem}
Finally, in the same spirit as Main Theorem~\ref{main theorem mass hyperbolic} involving the sign of mass, we will construct an explicit example of a solution $(g,k)$ to the constraint equations \eqref{constraint} on an AF manifold with negative mass. As in the AH case discussed earlier, this does not contradict the Positive Mass Theorem since the decay rate of the symmetric $(0,2)$-tensor $k$ in our construction is critical, meaning it lies exactly at the threshold where the assumptions of the theorem no longer hold. In particular, the example illustrates that the decay condition on $k$ in the Positive Mass Theorem is sharp and cannot be weakened in general. It is worth noting that this result was already obtained in \cite{NguyenRadialAF}, where the second author unified the analysis of the Born–Infeld equation and the vacuum conformal constraint equations with null TT-tensor in the radial setting. So, we reprove it here for the sake of completeness of the article. 
\begin{Theorem}[Constraint solutions with negative mass]
	There exists a vacuum constraint solution $(\RR^n, g,k)$ satisfying $g - \delta_{\text{Euc}} = O\big(\frac{1}{r^{n-2}}\big)$ and $k = O\big(\frac{1}{r^{n/2}} \big)$ at infinity, and such that the ADM mass of $(\RR^n, g)$ is negative.
\end{Theorem}

We have presented all the main results. We now conclude this introduction with an outline of the article as follows. In Section \ref{sec-Harmonic_mfds}, we present some standard facts on harmonic manifolds which we will use throughout the article and form the basis for our solutions to the  conformal equations.

In Section 3, provided that the manifold is harmonic, building on the result in \cite{NguyenRadialAF}, we will reduce the conformal equations \eqref{CE} of radial settings to a single nonlinear equation involving functions of the distance, which is much easier to solve.

In Sections \ref{section Sphere}, \ref{section hyperbolic} and \ref{section Euclidean}, a complete description of this single equation will be obtained in the special cases where the manifold is the sphere, the hyperbolic space or the Euclidean one. Each of these sections begins with the proof of our Main Theorem \ref{main theorem} in the setting considered, and then derives applications by studying the properties of solutions.
%
%
\section{Analysis on harmonic manifolds} \label{sec-Harmonic_mfds}
%
%
In this section, we develop the analytical framework needed to study the conformal equations \eqref{CE}. Although our main results concern constant curvature spaces, they mainly rely on the properties of the mean curvature of geodesic spheres and most of the computation can be done in the more general setting of harmonic manifolds.

We begin in Subsection \ref{subsection Prelimi} with a review of key concepts related to harmonic manifolds. Then, in Subsection \ref{subsection primary functions}, we construct a class of primary functions on harmonic manifolds and establish their main properties, which play a central role in solving the conformal equations.

\subsection{Preliminaries} \label{subsection Prelimi}

Harmonic manifolds are those Riemannian manifolds whose harmonic functions have the mean value property. Equivalently, a Riemannian manifold $M$ is harmonic if and only if there exists a function $h:\RR_+^*\to\RR$ such that any sphere of radius $r$ has constant mean curvature $h(r)$. Another equivalent property is that the Laplacian of a radial function is still a radial function. In particular, the distance function $r$ from any point satisfies $\Delta r=h(r)$. It is not difficult to show that harmonic spaces are Einstein manifolds and that a Riemannian manifold is harmonic if and only if its universal cover is harmonic. Therefore, in what follows we will only consider simply connected harmonic manifolds.

There is a rich literature on the existence and classification of harmonic manifolds. The most basic examples are Euclidean spaces and the Rank One Symmetric Spaces (ROSS). In 1944, A. Lichnerowicz conjectured — proving it in dimension 4 — that these are the only harmonic manifolds. The conjecture was later confirmed by Z.I. Szabó for compact, simply connected manifolds of any dimension (see \cite{Szabo}). However, in a significant development, E. Damek and F. Ricci showed in \cite{Damek-Ricci} that the conjecture fails in the non-compact case by constructing harmonic homogeneous manifolds that are not ROSS. To date, the only known harmonic manifolds are the Euclidean spaces, ROSS, and Damek–Ricci spaces. For an introduction to the basic structure of ROSS, we refer the reader to \cite[Chapter 3]{Besse} or \cite{Kobayashi-Nomizu-II}. General properties of harmonic manifolds can be found in \cite[chapter 6]{Besse}, \cite[sections 1 and 2]{Szabo} and \cite[section 2]{Knieper-Peyerimhoff}.

Another characterization of harmonic manifolds is that there exists a function $\theta:(0,\dM)\to\RR_+$ such that at any point $x\in M$ the volume form in normal coordinates reads
\[
dv_M=\theta(r)drd\xi
\]
where $r$ is the distance to $x$ and $d\xi$ denotes the canonical volume form on the unit tangent sphere $U_xM$. In this context, the function $\theta$ is called the \textit{density function} of $(M,g)$.

From the classical asymptotic expansion of $\theta$, it follows that any harmonic manifold must be Einstein; that is, its Ricci curvature satisfies $\Ric=Cg$, for some constant $C$. For simplicity, we will often denote this by $\Ric=C$.

\medskip

Another important geometric aspect of harmonic manifolds is the shape operator of the sphere of radius $r$ defined by
\[
S(r):\left\{\begin{array}{rcl}
	\dr^\bot & \to & \dr^\bot \\
	u & \mapsto & D_u\dr
\end{array}\right.
\]
where $\dr$ denotes the radial field centered at $x$.  The mean curvature of the sphere of radius $r$, denoted $h(r)$, is the trace of $S(r)$:
$$
h(r)=\tr(S(r)).
$$ 
The relationship between the density function $\theta$, the shape operator $S(r)$, and the distance function $r$ is reflected in the following identities, which will be crucial in our analysis of the conformal equations:
$$
\Delta r = h(r) \qquad \text{and} \qquad h(r) = \frac{\theta^\p(r)}{\theta(r)}.
$$

In the next subsection, we will construct from $(\theta, h)$ basic functions for the study of the equations \eqref{CE}. In order to ensure essential properties of these functions, we need to recall some results on $(\theta(r), h(r), S(r))$ as follows.
\begin{lemma} \label{lemma shape operator S}
	Given a geodesic emanating from $x$, the shape operators along  this geodesic satisfy the Riccati equation :
	$$
	S^\p(r) + S(r)^2 + R(\dr,.)\dr = 0.
	$$
	In particular, taking the trace we get
	\begin{equation} \label{identity h}
		h^\p(r) + |S(r)|^2 + \Ric = 0.
	\end{equation}
\end{lemma} 
\begin{proof}
	These are classical results, see \cite{Jost}.
\end{proof}

\begin{lemma} \label{lemma inequality h} We have
	\begin{equation} \label{inequality h}
		h^\p + \frac{h^2}{n-1} + \Ric \le 0.
	\end{equation}
	Here, the equality holds if and only if $(M , g)$ is spherical, hyperbolic or Euclidean.
\end{lemma}
\begin{proof}
	This inequality is indeed a direct consequence of the previous lemma. In fact, since $h(r) = \tr(S(r))$, Cauchy-Schwartz inequality gives $h(r)^2\le(n-1)|S(r)|^2$. Therefore, taking into \eqref{identity h}, the inequality \eqref{inequality h} follows.
\end{proof}

\begin{lemma} \label{lemma h near 0}
	Let $(M,g)$ be a harmonic manifold, its density and mean curvature functions have the following properties:
	\begin{itemize}
		\item[(a)] as $r \to 0$ we have $\theta(r) \sim r^{n-1}$ and
		\begin{equation} \label{h near 0}
			h(r) \sim \frac{n-1}{r}.
		\end{equation} 
		\item[(b)] if $M$ is non-compact with negative Ricci curvature, then $h$ is positive and decreasing to a positive constant $E$.
		\item[(c)] If $(M,g)$ is compact, then 
		\begin{equation} \label{h at dM}
			\lim_{r\to\dM}\theta(r)=0 \quad \text{and} \lim_{r\to\dM}h(r)=-\infty.
		\end{equation}
	\end{itemize}
\end{lemma}
\begin{proof}
	Point (a) is a standard result of Riemannian geometry. For point (b), see \cite[Corollary 2.8 and Section 6]{Knieper-Peyerimhoff}. Point (c) follows from Szab\'o's theorem and the standard properties of compact ROSS in \cite{Szabo}.
\end{proof}

\begin{remark} \label{remark compact and Ric}
	According to the main result in Szab\'o \cite{Szabo}, a harmonic manifold is compact if and only if $\Ric > 0$. Therefore, in this paper, when we say the manifold is compact, this can sometimes be taken to mean $\Ric >0$. The same goes for the non-compact case, which may correspond to $\Ric \le 0$, depending on the context.
\end{remark}

%
\subsection{Primary functions for the conformal equations} \label{subsection primary functions}
Let $(M,g)$ be a harmonic manifold. In this subsection, based on the density and mean curvature functions $(\theta, h)$ defined above, we will introduce a set of primary radial functions on $(M,g)$, which will be employed for addressing the conformal equations.

For convenience, we will denote the space of radial functions by the letter $R$, and the subscript $+$ will denote the subset of positive functions (for example, $RC_+^1(M)$ denotes the spaces of radial positive functions of class $C^1$ on $M$). It is worth noting that the radial function must be regular at the origin, which implies some conditions on its derivatives. For example, if $f(r)\in RC^1(M)$, then we have $f\in C^1(\fo{0,\dM})$ and $f^\p(0)=0$ (and therefore, the gradient of $f(r)$ vanishes at the origin). If $M$ is compact we also have $f\in C^1(\ff{0,\dM})$ and $f^\p(\dM)=0$ and these conditions on $f$ are sufficient for $f(r)$ to be $C^1$ on $M$. In some computations, when all the considered functions are radial, we may omit the dependence on $r$ for convenience.

\medskip

We begin by defining $u_0 \in C^2((0, \text{diam}(M))) \cap C^1([0, \text{diam}(M)))$ to be the solution to the following linear equation
\begin{equation} \label{def u0}
	\begin{cases}
		u_0^\pp + h u_0^\p + \tfrac{n}{n-1} \Ric u_0 = 0, & \forall r \in (0,\dM)
		\\
		u_0 (0) = 1, \, u_0^\p (0) = 0. &
	\end{cases}
\end{equation}
Observing that since, by the definition of $h$ and Lemma \ref{lemma h near 0},
$$
h \in RC^1\big((0, \dM)\big) \quad \text{and} \quad \int_0^{2\eps} r h (r) \, dr < +\infty,
$$
existence and uniqueness of $u_0$ are guaranteed according to \cite{BilesRobinsonSpraker}. The following result ensures the no-sign change of $u_0$ when $(M,g)$ is non-compact. 
\begin{proposition} \label{proposition sign of u0} 
	If $\Ric \le 0$, then $u_0^\p \ge 0$. In particular, $u_0$ is increasing and so $u_0 \ge 1$.
\end{proposition}
\begin{proof} 
	Recall that $h(r) = \frac{\theta^\p(r)}{\theta(r)}$. Then, we may rewrite \eqref{def u0} as
	\begin{equation} \label{rewrite u_0 equation}
		(u_0^\p \theta)^\p = - \tfrac{n}{n-1} \Ric \theta u_0.
	\end{equation}
	Since $u_0(0) = 1$ and $\theta > 0$ near $0$, if $\Ric \le 0$, we deduce from \eqref{rewrite u_0 equation} that there exists $\eps > 0$ such that $(u_0^\p \theta)^\p \ge 0$ for all $r \in [0,\eps]$.
	Therefore, $u_0^\p \theta$ is increasing in $[0,\eps]$ and so
	$$
	u_0^\p \theta \ge (u_0^\p \theta)(0) = 0
	$$
	for all $r \in [0,\eps]$.  Since $\theta > 0$, it follows that
	\begin{equation}\label{> 1 near 0}
		u_0^\p \ge 0 \quad \text{in $[0,\eps]$}.
	\end{equation}
	Suppose that the proposition is false. Then, thanks to \eqref{> 1 near 0}, we can define
	$$
	r_0 := \sup\bigg\{r_1 \in (0, +\infty) ~~\text{such that}~~u_0^\p \ge 0 ~~\forall r \in [0, r_1]  \bigg\} < +\infty.
	$$
	It is clear that by the continuity of $u_0^\p$ and by the definition of $r_0$, we must have $u_0^\p \ge  0$ in $[0, r_0]$ and $u_0^\p(r_0) = 0$. In particular, this gives us that
	$$
	u_0(r_0) \ge u_0(0) = 1 \quad \text{and} \quad u_0^\p(r_0) = 0.
	$$
	However, repeating the argument of obtaining \eqref{> 1 near 0}, with $0$ replaced by $r_0$, shows that there exists $\eps_1 > 0$ such that 
	$$
	u_0^\p(r) \ge 0  \quad \text{in $[r_0, r_0 + \eps_1]$}.
	$$
	Of course, this is a contradiction due to the definition of $r_0$, therefore, $u_0^\p \ge 0$ for all $r \ge 0$. The proof is completed.
\end{proof}

We have shown in Proposition \ref{proposition sign of u0} that $u_0$ does not change sign when $\Ric \le 0$. However, when $\Ric > 0$, that is the compact case, the sign of $u_0$ may not be controlled in general, and so it causes difficulties in defining related functions as we will see below. Therefore, for technical reasons, from now on we will assume that $(M,g)$ satisfies one of the following two conditions:
\begin{subequations}\label{condition on u0}
	\begin{equation}
		\text{(i) $\Ric \le 0$. \hspace{8.6cm}}
	\end{equation}
	\begin{itemize}
		\item[(ii)] If $\Ric > 0$, then there exists a constant $l \ge 0$ such that
		\begin{equation}\label{h near dM}
			\lim_{r \to \dM^-} (\dM - r)h(r) = -l,	
		\end{equation}
		and $u_0 \in C^1([0,\dM])$ has a unique zero-point $\rstar$ in $[0, \dM]$ with $\rstar \ne \dM$.
	\end{itemize}
\end{subequations}

\medskip

We note that in the second case, the existence of a limit in \eqref{h near dM} holds for all compact ROSS manifolds. Therefore, it makes sense to include this as part of our assumptions. As for the condition that $u_0$
has at most one zero, we do not know whether this holds in general, but it is at least known to be true on the sphere.

\medskip

Now, based on $u_0$ and the density function $\theta$, let us define the following functions, which are fundamental to our analysis in this article. First, given $v \in RC^0\big(M\big)$, we write
\begin{equation} \label{def f}
	f_v(r) := \tfrac12 \int_0^r v \, ds.
\end{equation}
The behavior of $f_v$  in relation to $(u_0, \theta)$ at the origin and $\dM$ is as follows.
\begin{proposition} \label{proposition limit of f at boundary} 
	We have
	\begin{enumerate}[(a)]
		\item $\displaystyle \lim_{r \to 0^+} \bigg( \frac{h}{\theta} \int_0^r f_v u_0 \theta \, ds \bigg) = 0$.
		\item If $\Ric > 0$, as long as 
		\begin{equation}\label{condition on f}
			\int_0^{\dM} f_v u_0 \theta \, ds = 0,
		\end{equation}
		we have $\displaystyle \lim_{r \to \dM^-} \bigg( \frac{h}{\theta} \int_0^r f_v u_0 \theta \, ds \bigg) = \frac{l}{l+1} \big(f_v u_0 \big)(\dM)$.
	\end{enumerate}
\end{proposition}

\begin{proof} 
	We have by Lemma \ref{lemma h near 0}
	$$
	\lim_{r \to 0^+} \bigg( \frac{h}{\theta} \int_0^r f_v u_0 \theta \, ds \bigg) =  (n-1) \lim_{r \to 0^+} \bigg( \frac{1}{r \theta} \int_0^r f_v u_0 \theta \, ds \bigg).
	$$
	Applying l'H\^{o}pital's rule to the right-hand side, thanks to \eqref{h near 0} in Lemma \ref{lemma h near 0}, it follows that
	\begin{align*}
		\lim_{r \to 0^+} \bigg( \frac{h}{\theta} \int_0^r f_v u_0 \theta \, ds \bigg) &= (n-1) \lim_{r \to 0^+} \bigg( \frac{f_v u_0 \theta}{\theta + r \theta^\p}\bigg) \\
		&= (n-1)\lim_{r \to 0^+} \bigg( \frac{f_v u_0}{1+rh}\bigg)	\\
		&= 0,
	\end{align*}
	as $rh(r)\to n-1$ and $f_v(0)=0$, which claims (a).
	
	\medskip
	
	Now, for proving (b), suppose that $\Ric > 0$. In this case, we recall by Remark~\ref{remark compact and Ric} that $(M,g)$ is compact, and so $\dM < +\infty$. Similarly to what we have done for the convergence at $0$, as long as \eqref{condition on f} holds, we obtain from l'H\^{o}pital's rule and \eqref{h near dM} that
	\begin{align*}
		\lim_{r \to \dM^-} \bigg( \frac{h}{\theta} \int_0^r f_v u_0 \theta \, ds \bigg) &=  -l\lim_{r \to \dM^-} \bigg( \frac{1}{(\dM - r)\theta} \int_0^r f_v u_0 \theta \, ds \bigg) \\
		& = -l \lim_{r \to \dM^-} \bigg( \frac{f_v u_0 \theta}{-\theta + (\dM - r) \theta^\p}\bigg) \\
		& = -l \lim_{r \to \dM^-} \bigg( \frac{f_v u_0}{-1+(\dM - r)h}\bigg) \\
		& = \frac{l}{l+1} \big(f_v u_0\big) (\dM),
	\end{align*}
	which completes the proof.
\end{proof}

\begin{definition} \label{definition F}
	Assume that $(M,g)$ satisfies the one of two conditions \eqref{condition on u0}. For any $f_v$ given as in \eqref{def f}, we define $F_v$ as follows.
	\begin{itemize}
		\item If $\Ric \le 0$, we set 
		\begin{equation} \label{def F nonpositive Ric}
					F_v(r) = u_0(r) \int_0^r \bigg(\frac{1}{u_0^2(s)\theta(s)} \int_0^s f_v u_0 \theta \, dt \bigg)  \, ds.
		\end{equation}
		for all $r \in [0, \dM)$.
		
		\item If $\Ric>0$, as long as the vanishing condition \eqref{condition on f} holds, we write
		\begin{equation}\label{def F compact}
			F_v(r) = \begin{cases}
				 \displaystyle u_0(r) \int_{0}^r \bigg(\frac{1}{u_0^2(s)\theta(s)}\int_0^s f_v u_0 \theta \, dt \bigg)\, ds + c_v u_0(r) & \text{if $r \in [0 , \rstar)$,} \\
				 & \\
				 \displaystyle - \frac{1}{u_0^\p (\rstar)\theta (\rstar)} \int_0^{\rstar} f_v u_0 \theta \, ds & \text{if $r = \rstar$,} \\
				 & \\
				 \displaystyle - u_0(r) \int_r^{\dM} \bigg(\frac{1}{u_0^2(s)\theta(s)}\int_0^s f_v u_0 \theta \, dt \bigg)\, ds & \text{otherwise}
			\end{cases}
		\end{equation}

		where $c_v$ is a constant given in \eqref{def c} below.
	\end{itemize}
\end{definition}

\begin{proposition} \label{proposition properties of F}
	$F_v$ is in $C^2(M)$ and satisfies the equation 
	\begin{equation} \label{F identity}
		F^\pp_v + h F^\p_v + \tfrac{n}{n-1}\Ric  F_v =  f_v.
	\end{equation}
\end{proposition}
\begin{proof}
	We divide the proof into two cases.
	
	$\bullet$ \textbf{Case $\Ric \le 0$.} First, thanks to Proposition \ref{proposition sign of u0} and the definition of $\theta$, we keep in mind that both $u_0$ and $\theta$ are positive in $(0, +\infty)$. Since 
	$$
	u_0(0) = 1 \quad \text{and} \quad 
	\lim _{r \to 0^+} h =+\infty,
	$$
	Proposition \ref{proposition limit of f at boundary}(a) gives us
	\begin{equation} \label{limit of f at 0}
		\lim_{r \to 0^+} \bigg( \frac{1}{u_0^2\theta} \int_0^r f_v u_0 \theta \, ds \bigg) = 0.
	\end{equation}
	Taking this into account, it is clear by definition that $F_v$ is well defined in $[0, +\infty)$ and moreover $F_v \in C^2((0, +\infty))$. Therefore, to show that $F_v \in C^2(M)$, we only need to prove that $F_v$ is twice continuously differentiable at $0$ and that $h F^\p_v$ converges at $0$.
	
	In fact, with Lemma \ref{lemma h near 0} in mind, we obtain by l'H\^{o}pital's rule and \eqref{limit of f at 0} that
	\begin{align} \label{limit hF at 0}
		\lim_{r \to 0^+} hF_v &= \lim_{r \to 0^+} \bigg( h u_0 \int_0^r \Big(\frac{1}{u_0^2 \theta } \int_0^s f_v u_0 \theta \, ds_1 \Big) \bigg) \notag \\
		&= (n - 1) \lim_{r \to 0^+} \bigg(\frac{1}{r} \int_0^r \Big(\frac{1}{u_0^2 \theta } \int_0^s f_v u_0 \theta \, ds_1 \Big) \bigg) \notag \\
		&= (n - 1) \lim_{r \to 0^+} \bigg( \frac{1}{u_0^2 \theta } \int_0^r f_v u_0 \theta \, ds \bigg) \notag \\
		&= 0.
	\end{align}
	Next, by straightforward calculations, we have
	\begin{subequations}\label{identity of F1 and its diff}
		\begin{align}
			F^\p_v &= \frac{u_0^\p}{u_0}F_v +  \frac{1}{u_0 \theta } \int_0^r f_v u_0 \theta \, ds, \\
			F^\pp_v &= \frac{u_0^\pp}{u_0}F_v - \frac{h}{u_0 \theta} \int_0^r f_v u_0 \theta \, ds + f_v
		\end{align}
	\end{subequations}
	for all $r \in (0 , +\infty)$. Combining  Proposition \ref{proposition limit of f at boundary}(a), \eqref{limit hF at 0} and \eqref{identity of F1 and its diff}, it follows that
	\begin{align*}
		\lim_{r \to 0^+} F_v & = 0, \\
		\lim_{r \to 0^+} hF^\p_v & = \lim_{r \to 0^+} \bigg( \frac{u_0^\p}{u_0} h F_v + \frac{h}{u_0 \theta } \int_0^r f_v u_0 \theta \, ds \bigg) = 0, \\
		\lim_{r \to 0^+} F^\pp_v & = \lim_{r \to 0^+} \bigg( \frac{u_0^\pp}{u_0}F_v - \frac{h}{u_0 \theta} \int_0^r f_v u_0 \theta \, ds + f_v \bigg) = 0,
	\end{align*}
	and so, $F_v \in C^2(M)$ as desired. Now, multiplying the first equation of \eqref{identity of F1 and its diff} by $h$ and taking the sum, we  have
	$$
	F^\pp_v + hF^\p_v - \frac{u_0^\pp + h u_0^\p}{u_0} F_v = f_v.
	$$
	Note that, by definition, $u_0^\pp + h u_0^\p = - \tfrac{n}{n -1} \Ric u_0$, then the identity \eqref{F identity} follows, which completes the proof for the case $\Ric\le0$.
	
	\medskip
	
	$\bullet$ \textbf{Case $\Ric > 0$.} Recall that in this case $\theta$ vanishes at $\{0, \dM\}$ and $u_0$ is assumed to have a unique zero-point $\rstar \in [0, \dM]$ with $\rstar \ne \dM$.  Therefore, since $u_0(0) = 1$, it follows that $u_0$ is positive in $[0,\rstar)$ and negative in $(\rstar, \dM]$.
	
	For abbreviation, we decompose $F_v = G_v + H_v$ where
	\begin{equation} \label{def F1 compact}
		G_v(r) = \begin{cases}
			\displaystyle u_0 \int_{0}^r \bigg(\frac{1}{u_0^2(s)\theta(s)}\int_0^s f_v u_0 \theta \, dt \bigg)\, ds & \text{if $r \in [0 , \rstar)$,} \\
			& \\
			\displaystyle - \frac{1}{u_0^\p (\rstar)\theta (\rstar)} \int_0^{\rstar} f_v u_0 \theta \, ds & \text{if $r = \rstar$,} \\
			& \\
			\displaystyle - u_0 \int_r^{\dM} \bigg(\frac{1}{u_0^2(s)\theta(s)}\int_0^s f_v u_0 \theta \, dt \bigg)\, ds & \text{otherwise}
		\end{cases} 
	\end{equation}
	and 
	\begin{equation}
		H_v(r) = \begin{cases}
			c_v u_0(r)  & \text{if $r \in [0 , \rstar)$,} \\
			0 \quad & \text{otherwise.}
		\end{cases}
	\end{equation}
	By the definition of $H_v$, it is easy to check that
	\begin{enumerate}[(i)]
		\item $H_v \in C^0\big([0,\dM]\big) \cap C^2 \big([0, \dM] \setminus \{\rstar\}\big)$, and
		\begin{equation} \label{F2 identity}
			H^\pp_v + h H^\p_v + \tfrac{n}{n-1}\Ric H_v =  0,
		\end{equation}
		on $[0, \dM] \setminus \{\rstar\}$.
		
		\item $H_v$ is semi-differentiable at $\rstar$ with
		$$
		\lim_{r \to \rstar^+} H^\p_v(r) = 0 \quad \text{and} \quad \lim_{r \to \rstar^-} H^\p_v(r) = c_v u_0^\p (\rstar),
		$$
		
		\item $hH^\p_v$ converges at $0$ and $\dM$. 
	\end{enumerate}
	Now, the following observations are the key in our arguments:
	\begin{itemize}
		\item If $G_v$ also satisfies three properties (i-iii) of $H_v$ above except that \eqref{F2 identity} is replaced by
		\begin{equation} \label{F1 identity}
			G^\pp_v + h G^\p_v + \tfrac{n}{n-1}\Ric G_v = f_v,
		\end{equation}
		then, by choosing 
		\begin{equation} \label{choose c}
			c_v = \frac{1}{u_0^\p(\rstar)} \bigg( \lim_{r \to \rstar^+} G^\p_v - \lim_{r \to \rstar^-} G^\p_v \bigg),
		\end{equation}
		it follows that
		$$
		\lim_{r\to \rstar^-}  \big(G^\p_v + H^\p_v\big)(r) = \lim_{r\to \rstar^+} \big(G^\p_v + H^\p_v \big)(r).
		$$ 
		This means that $F_v = G_v + H_v \in C^1([0, \dM])$, and so thanks to (\ref{F2 identity}--\ref{F1 identity}) we deduce that $F_v \in C^2([0, \dM])$ and 
		$$
		F^\pp_v + h F^\p_v + \tfrac{n}{n-1}\Ric F_v = f_v
		$$
		for all $r \in [0,\dM]$.
		
		\item Continuing with the previous observation, we note that as long as $G_v$ and $H_v$ satisfy the assertion (iii) above, that is $hG_v$ and $hH_v$  converge at $0$ and $\dM$, then so does $hF_v$. Therefore, the identity \eqref{F identity} can be rewritten as
		$$
		-\Delta F_v = f_v - \tfrac{n}{n-1}\Ric F_v,
		$$
		and hence, since $f_v - \tfrac{n}{n-1}\Ric  F_v \in C^1(M)$, we obtain from the Sobolev embedding theorem that $F_v \in C^3(M)$. 
	\end{itemize}
	From these two observations, to complete the proof, we only need to prove that $G_v$ satisfies the conditions in the first observation above. In fact, first, similarly to the case where $\Ric \le 0$, we have 
	$$
	G_v \in C^2([0,\rstar)), \quad hG^\p_v (0) = 0
	$$ 
	and
	\begin{equation}\label{F1 identity -}
		G^\pp_v + h G^\p_v + \tfrac{n}{n-1}\Ric  G_v =  f_v \quad \text{for all $r \in [0,\rstar)$}.
	\end{equation}
	For $r \in (\rstar, \dM]$, in the same manner as proving the regularity at $0$, thanks to Proposition \ref{proposition limit of f at boundary}(b), we also get that
	$$
	G_v \in C^2((\rstar, \dM]) \quad \text{and} \quad \lim_{r \to \dM^-} h G^\p_v  = \frac{l}{l+1} f_v(\dM),
	$$
	where $l$ is the limit defined in \eqref{h near dM}. Moreover, by straightforward calculations, we have that for all $r \in (\rstar,dM)$
	\begin{align*}
		G^\p_v &= \frac{u_0^\p}{u_0}G_v + \frac{1}{u_0 \theta } \int_0^r f_v u_0 \theta \, ds, \\	
		G^\pp_v &= \frac{u_0^\pp}{u_0}G_v - \frac{h}{u_0 \theta} \int_0^r f_v u_0 \theta \, ds + f_v.
	\end{align*}
	Therefore, multiplying the first equation by $h$ and taking the sum, we obtain
	\begin{equation} \label{F1 identity +}
		G^\pp_v + h G^\p_v + \tfrac{n}{n-1}\Ric  G_v =  f_v, \quad \text{$\forall r \in (\rstar, \dM]$}.
	\end{equation}
	Now, to address $G_v$ at $r = \rstar$, we first claim that since $u_0^\p(\rstar) \ne 0$, $G_v$ is well-defined at $\rstar$ by definition. To establish the continuity of $G_v$ at $\rstar$, applying l'H\^{o}pital's rule, we have
	\begin{align*}
		\lim_{r \to \rstar^-} G_v(r) & = \lim_{r \to \rstar^-} \frac{\bigg( \int_0^r \Big( \int_0^s f_v u_0 \theta \, ds_1 \Big) \frac{1}{u_0^2 \theta } \, ds \bigg)^\p}{\bigg( \frac{1}{u_0}\bigg)^\p} \\
		& = - \frac{1}{u_0^\p (\rstar)\theta (\rstar)} \int_0^{\rstar} f_v u_0 \theta \, dr,
	\end{align*}
	and
	\begin{align*}
		\lim_{r \to \rstar^+} G_v(r) & = - \lim_{r \to \rstar^+} \frac{\bigg( \int_r^{\dM} \Big( \int_0^s f_v u_0 \theta \, ds_1 \Big) \frac{1}{u_0^2 \theta } \, ds  \bigg)^\p}{\bigg( \frac{1}{u_0}\bigg)^\p} \\
		& = - \frac{1}{u_0^\p (\rstar)\theta (\rstar)} \int_0^{\rstar} f_v u_0 \theta \, dr.
	\end{align*}
	This means by definition that $\lim_{r \to \rstar^-} G_v(r) = \lim_{r \to \rstar^+} G_v(r) = G_v(\rstar)$, and so $G_v$ is continuous at $\rstar$ as desired. To show the semi-differentiability of $G_v$ at $\rstar$, thanks to \eqref{F1 identity -}, \eqref{F1 identity +} and the definition of $h$, it follows that 
	$$
	\big(\theta G^\p_v \big)^\p  = \theta f_v - \tfrac{n}{n-1} \Ric \theta G_v \quad \text{in $[0,\dM] \setminus \{\rstar\}$}.
	$$
	Integrating this equation, we obtain that
	\begin{itemize}
		\item if $r \in (0, \rstar)$, then
		$$
		\theta(r) G^\p_v(r) = \int_0^r \big(\theta G^\p_v \big)^\p \, ds = \int_0^r \bigg(\theta f_v - \tfrac{n}{n-1} \Ric \theta G_v \bigg) \, ds ,
		$$
		\item if $r \in ( \rstar, \dM)$, then
		$$ 
		\theta(r) G^\p_v(r) = - \int_r^{\dM} \big(\theta G^\p_v \big)^\p \, ds = - \int_r^{\dM} \bigg(\theta f_v - \tfrac{n}{n-1} \Ric \theta G_v \bigg) \, ds.
		$$
	\end{itemize}
	Since $\theta(\rstar) > 0$, this gives us that
	\[
	\lim_{r \to \rstar^-}G^\p_v(r) = \frac{1}{\theta(\rstar)}\int_0^{\rstar} \bigg(\theta f_v - \tfrac{n}{n-1} \Ric \theta G_v \bigg) \, ds, 
	\]
	and
	\[
	\lim_{r \to \rstar^+} G^\p_v(r) = - \tfrac{1}{\theta(\rstar)}\int_{\rstar}^{\dM} \bigg(\theta f_v - \tfrac{n}{n-1} \Ric \theta G_v \bigg) \, ds.
	\]
	Therefore, $G_v$ is semi-differentiable at $\rstar$ as claimed. 
	
	Now, in view of \eqref{choose c}, by setting
	\begin{align} \label{def c}
		c_v &:= \tfrac{1}{u_0^\p(\rstar)} \bigg(\lim_{r \to \rstar^+} G^\p_v(r) - \lim_{r \to \rstar^-} G^\p_v(r) \bigg) \notag \\
		&= - \tfrac{1}{\theta(\rstar) u_0^\p(\rstar)} \int_0^{\dM} \bigg(\theta f_v - \tfrac{n}{n-1} \Ric \theta G_v \bigg) \, ds,
	\end{align}
	it follows from the observation above that $F_v = G_v + H_v$ satisfies the properties claimed in the proposition. The proof is completed.
\end{proof}
%
%
\section{Conformal equations in harmonic manifolds} \label{section reduction}
%
%
We now study the conformal equations on a harmonic manifold $(M,g)$. For our purpose in the article, we are only interested in radial solutions $\varphi(r)$ to \eqref{CE}. Therefore, we will consider \eqref{CE} in a simple setting with $\sigma$ zero everywhere and $(\tau(r), \psi(r), \rho(r))$ radial functions. In this case, the conformal equations \eqref{CE} are rewritten as
\begin{subequations}\label{CE with null sigma}
	\begin{align}
		- \tfrac{4(n-1)}{n-2} \Delta \varphi + \mathcal{R}_\psi \varphi  
		& = \mathcal{B}_{\tau, \psi} \varphi^{N-1}  + \tfrac{|LW|^2 + \rho^2}{\varphi^{N + 1} } \\
		\Div(L W) & = \tfrac{n-1}{n} \varphi^N d\tau - \rho d\psi.
	\end{align}
\end{subequations}
The aim of this section is to show that when $(M,g)$ is the sphere $\Sp^n$ or harmonic with $\Ric \le 0$, the conformal equations \eqref{CE with null sigma} can be well treated in the space of radial functions. Moreover, in this context, we can reduce them to a single non-linear equation of functions of the distance, which is much easier to solve. The outline of this section is as follows. In Subsection \ref{subsection calculate CKO}, we will deal with the conformal killing operator $L$ for a given symmetric $1-$form $W$. Next, in Subsection \ref{subsection Lic-type eq}, based on our computations of $L$ and the primary functions in Section \ref{sec-Harmonic_mfds}, the full system of the conformal equations will be addressed and reduced  to a single equation.

%
\subsection{Calculations on the conformal killing operator} \label{subsection calculate CKO}
%
In this section we compute $\Div(LW)$ and $|LW|^2$ for a $1$-form $W$ of the form $W=u(r)dr$. The computation will use the basic results on the shape operator of geodesic spheres, and only the two following properties of the distance function $r$: its level hypersurfaces have constant mean curvature $h(r)$ and $|\nabla r|=1$.

For the convenience of the reader, letting $\dr$ denote the gradient of $r$, we recall the following observations:
\begin{itemize}
	\item For any tangent vector $v$ we have
	\[
	g(D_{\dr}\dr,v)=D^2r(\dr,v)=g(D_v\dr,\dr)=\frac{1}{2}v.|\dr|^2=0.
	\]
	Therefore $D_{\dr}\dr=0$ and the integral curves of $\dr$ are geodesics.
	\item The tangent space to the level hypersurfaces at a point $x\in M$ is $\dr^\bot$ and, since $|\dr|=1$, the map
	\[
	S:\left\{\begin{array}{rcl}
				\dr^\bot & \to & \dr^\bot \\
				u & \mapsto & D_u\dr
	\end{array}\right.
	\]
	is well defined and it is the shape operator of the hypersurface. Its trace is $\tr(S(r))=h(r)$, and since $D_{\dr}\dr=0$, we also have $\Delta r=h(r)$.
	
	\item Along a geodesic $\gamma$ orthogonal to the hypersurfaces of $r$, the operators $S(r)$ satisfy the Ricatti equation:
	\[
	S^\p(r) + S^2(r) + R(\dr,.)\dr = 0
	\]
	In particular, taking the trace we get
	\[
	h^\p(r) + |S(r)|^2 + \Ric(\dr,\dr) = 0.
	\]
\end{itemize}

\begin{proposition} \label{proposition 1-form W calculations}
	If $W=u(r)dr$ then
	\begin{enumerate}[(a)]
		\item $LW=2u^\p(r)dr\otimes dr + 2u(r)D^2r - \frac{2}{n}(u^\p(r)+u(r)h(r))g,$
		
		\item $\Div(L W)=2\frac{n-1}{n}\big(u^\pp(r)+u^\p(r)h(r)+u(r)h^\p(r)\big)dr + 2u(r)\Ric(\partial r,.),$
		
		\item $\frac{1}{4}|LW|^2  = \frac{(n-1)}{n}\big(u^\p(r) - \frac{1}{n-1}h(r)u(r)\big)^2 - \big( h^\p(r) + \frac{h(r)^2}{n-1} + \Ric(\dr,\dr) \big)u(r)^2.$
	\end{enumerate}
\end{proposition}
Since $M$ is a harmonic manifold, then it is Einstein and we have $\Ric(\dr,.)=\Ric dr$ where $\Ric$ is constant. Therefore, $\Div(LW)$ is a radial 1-form, and $|LW|$ is a radial function.
\begin{proof}
	Let $W^\#=u(r)\dr$ be the dual of the 1-form $W$. The 2-form $LW$ is the traceless part of the Lie derivative of the metric $g$ in the direction of $W^\#$:
	\[
	LW(X,Y)=g(D_XW^\#,Y) + g(X,D_YW^\#) - \frac{2}{n}\Div(W^\#)g(X,Y)
	\]
	for any vector field $X,Y$. From the definition of $W^\#$ we get
	\[
	\Div(W^\#) = u(r)\Delta r + u^\p(r) = u(r)h(r) + u^\p(r)
	\]
	and
	\begin{align*}
		g(D_XW^\#,Y) & = g((X.u(r))\dr,Y) + g(u(r)D_X\dr, Y) \\
		 & = u^\p(r)g(\dr,X)g(\dr,Y) + u(r)D^2r(X,Y).
	\end{align*}
	Summing all the terms we get
	\begin{multline*}
		LW(X,Y) = 2u^\p(r)g(\dr,X)g(\dr,Y) + 2u(r)D^2r(X,Y) \\
		- \frac{2}{n}(u^\p(r)+u(r)h(r))g(X,Y),
	\end{multline*}
	which gives point (a).
	
	Let $X$ be a vector field, from the definition of the divergence of a symmetric 2-form we have
	\[
	\Div(LW)(X) = \sum_{i=1}^{n}(D_{E_i}LW)(X,E_i)
	\]
	where $(E_1,\dots,E_n)$ is a local field of orthonormal frame. In order to compute
	\begin{equation} \label{eqn-b}
		(D_{E_i}LW)(X,E_i) = E_i.(LW(X,E_i)) - LW(D_{E_i}X,E_i) - LW(X,D_{E_i}E_i),
	\end{equation}
	we compute separately the three terms. Using point (a) and the standard rules of derivation we get
	\begin{multline} \label{eqn-term1}
		E_i.(LW(X,E_i)) = 2u^\pp(r)g(E_i,\dr)^2g(X,\dr) \\
		+ 2u^\p(r)\Big(g(D_{E_i}X,\dr)+g(X,D_{E_i}\dr)\Big)g(E_i,\dr) \\
		+ 2u^\p(r)g(X,\dr)\Big(g(D_{E_i}E_i,\dr)+g(E_i,D_{E_i}\dr)\Big) \\
		+ 2u^\p(r)g(E_i,\dr)g(D_X\dr,E_i) + 2u(r)\Big(g(D_{E_i}D_{E_i}\dr,X)+g(D_{E_i}\dr,D_{E_i}X)\Big) \\
		- \frac{2}{n}\Big(u^\pp(r) + u^\p(r)h(r) + u(r)h^\p(r)\Big)g(E_i,\dr)g(X,E_i) \\
		- \frac{2}{n}\big(u^\p(r) + u(r)h(r)\big)\Big(g(D_{E_i}X,E_i)+g(X,D_{E_i}E_i)\Big).
	\end{multline}
	Point (a) gives the last two terms
	\begin{multline} \label{eqn-term2}
		LW(D_{E_i}X,E_i) = 2u^\p(r)g(D_{E_i}X,\dr)g(E_i,\dr) + 2u(r)g(D_{E_i}\dr,D_{E_i}X) \\
		- \frac{2}{n}(u^\p(r)+u(r)h(r))g(D_{E_i}X,E_i)
	\end{multline}
	and
	\begin{multline} \label{eqn-term3}
		LW(X,D_{E_i}E_i) = 2u^\p(r)g(X,\dr)g(D_{E_i}E_i,\dr) + 2u(r)g(D_X\dr,D_{E_i}E_i) \\
		- \frac{2}{n}(u^\p(r)+u(r)h(r))g(X,D_{E_i}E_i)
	\end{multline}
	Equation \eqref{eqn-b} together with equalities \eqref{eqn-term1}, \eqref{eqn-term2} and \eqref{eqn-term3} now give (after basic cancellations)
	\begin{multline} \label{eqn-sum}
		(D_{E_i}LW)(X,E_i) = 2u^\pp(r)g(E_i,\dr)^2g(X,\dr) + 2u^\p(r)g(X,\dr)g(E_i,D_{E_i}\dr) \\
		+4u^\p(r)g(E_i,\dr)g(D_X\dr,E_i) \\
		+2u(r)\Big(g(D_{E_i}D_{E_i}\dr,X)-g(D_X\dr,D_{E_i}E_i)\Big) \\
		- \frac{2}{n}(u^\pp(r) + u^\p(r)h(r) + u(r)h^\p(r))g(E_i,\dr)g(X,E_i).
	\end{multline}
	In order to sum over $i$, we make the following observations, using that $(E_1,\dots,E_n)$ is an orthonormal frame field:
	\begin{itemize}
		\item $\sum_{i=1}^ng(E_i,\dr)^2=|\dr|^2=1$ ;
		\item $\sum_{i=1}^ng(E_i,D_{E_i}\dr)=\tr(D^2r)=h(r)$ ;
		\item $\sum_{i=1}^ng(E_i,\dr)g(D_X\dr,E_i) = g(D_X\dr,\dr)=0$ ;
		\item From the definition of the rough Laplacian $D^*Ddr$ of $dr$ we get
		\begin{align*}
			\sum_{i=1}^n\Big(g(D_{E_i}D_{E_i}\dr,X) & -g(D_X\dr,D_{E_i}E_i)\Big) \\
			 & = \sum_{i=1}^n\Big(g(D_{E_i}D_{E_i}\dr,X)-g(D_{D_{E_i}E_i}\dr,X)\Big) \\
			 & = D^*Ddr(X),
		\end{align*}
		and according to Bochner formula we have
		\[
		D^*D(dr) = \Delta dr + \Ric(\dr,.) = d\Delta r + \Ric(\dr,.) = h^\p(r)dr + \Ric(\dr,.),
		\]
		so that
		\[
		\sum_{i=1}^n\Big(g(D_{E_i}D_{E_i}\dr,X)-g(D_X\dr,D_{E_i}E_i)\Big) = h^\p(r)g(X,\dr) + \Ric(\dr,X)\ ;
		\]
		\item $\sum_{i=1}^ng(E_i,\dr)g(X,E_i) = g(X,\dr)$.
	\end{itemize}
	Now, taking the sum over $i$ of Equality \eqref{eqn-sum} we get
	\[
	\Div(LW)(X) = 2\tfrac{n-1}{n}\big(u^\pp(r)+u^\p(r)h(r)+u(r)h^\p(r)\big)g(\dr,X) + 2u(r)\Ric(\dr,X)
	\]
	which gives point (b).
	
	The norm $|LW|$ at $x\in M$ is given by $|LW|^2=\sum_{i,j=1}^nLW(e_i,e_j)^2$ where $(e_1,\dots,e_n)$ is an orthonormal basis of $T_xM$. Since $|\dr|$ is constant, the map $S:u\mapsto D_u\dr$ is a symmetric endomorphism of $\dr^\bot$, therefore we can choose $(e_1,\dots,e_{n-1})$ which diagonalizes $S$, and $e_n=\dr$. Denoting $\lambda_1,\dots,\lambda_n$ the eigenvalues of $S$ we have $\sum_{i=1}^{n-1}\lambda_i=h(r)$ and, using Ricatti equation, $\sum_{i=1}^{n-1}\lambda_i^2=|S|^2=-h^\p(r)-\Ric(\dr,\dr)$. For $1\le i,j\le n-1$ with $i\neq j$, the point (a) gives
	\begin{itemize}
		\item $LW(e_n,e_n)=LW(\dr,\dr)= 2u^\p(r) - \frac{2}{n}(u^\p(r)+u(r)h(r))$
		\item $LW(e_1,e_i)=2u(r)\lambda_i - \frac{2}{n}(u^\p(r)+u(r)h(r))$
		\item $LW(e_i,e_n)=LW(e_i,e_j)=0$.
	\end{itemize}
	From this we get
	\begin{align*}
		\frac{1}{4}|LW|^2 & = \Big(u^\p(r) - \frac{1}{n}(u^\p(r)+u(r)h(r))\Big)^2 \\
			& \qquad + \sum_{i=1}^{n-1}\Big(u(r)\lambda_i - \tfrac{1}{n}(u^\p(r)+u(r)h(r))\Big)^2 \\
		 & = \tfrac{(n-1)^2}{n^2}u^\p(r)^2 - 2\tfrac{n-1}{n^2}u^\p(r)u(r)h(r) + \tfrac{1}{n^2}u(r)^2h(r)^2 \\
		 & \qquad + \sum_{i=1}^{n-1}\Big(u(r)^2\lambda_i^2 - \frac{2\lambda_i}{n}u(r)(u^\p(r)+u(r)h(r)) + \tfrac{1}{n^2}(u^\p(r) \\
		 & \qquad +u(r)h(r))^2\Big) \\
		 & = \tfrac{(n-1)^2}{n^2}u^\p(r)^2 - 2\tfrac{n-1}{n^2}u^\p(r)u(r)h(r) + \tfrac{1}{n^2}u(r)^2h(r)^2
		 - u(r)^2h^\p(r) \\
		 & \qquad - u(r)^2\Ric(\dr,\dr) - \tfrac{2}{n}h(r)u(r)\big(u^\p(r)+u(r)h(r)\big) \\
		 & \qquad + \tfrac{n-1}{n^2}\big(u^\p(r)+u(r)h(r)\big)^2 \\
		 & = \tfrac{(n-1)}{n}\Big(u^\p(r) - \frac{1}{n-1}h(r)u(r)\Big)^2 \\
		 & \qquad - \Big( h^\p(r) + \frac{h(r)^2}{n-1} + \Ric(\dr,\dr) \Big)u(r)^2,
	\end{align*}
	which proves point (c). The proof is completed.
\end{proof}

\begin{remark}
	Let $M$ be a complete Riemannian manifold and let $r:M\to\RR$ be a function such that $|\nabla r|=1$ and whose level hypersurfaces have constant mean curvature $h(r)$. All the previous computations still hold for a $1-$form $W=u(r)dr$ on $M$. Besides the distance function to any point of a harmonic manifold, there are other examples of such situations.
	
	If the metric $g$ is rotationally symmetric around a point $x_0\in M$, then the distance to $x_0$ satisfies these conditions. The metric has the form $g=dr^2+f(r)^2\sigma_s$ where $\sigma_s$ is the standard metric of $\Sp^{n-1}$ and $f>0$, the spheres centered at $x_0$ are totally umbilical with mean curvature $h(r)=(n-1)\frac{f^\p(r)}{f(r)}$. Classical computations give $\Ric(\dr,.)=-(n-1)\frac{f^\pp(r)}{f(r)}dr$, and we get that $\Div(LW)$ is a radial 1-form and $|LW|$ is a radial function. The previous computations may be useful in this context. 
	
	Another typical example of such a function is the distance function to a totally geodesic submanifold in a space form. Note that the Euclidean space, the round sphere and the hyperbolic space are the only Riemannian manifolds which are both harmonic and spherically symmetric.
\end{remark}

%
\subsection{Reduction of the conformal equations} \label{subsection Lic-type eq}
We consider now the conformal equations \eqref{CE with null sigma} for a radial data set $(V, \tau(r), \psi(r), \rho(r))$. Let us begin by considering the vector equations
\begin{equation} \label{equation vectors}
	\Div(L W) = v(r) dr,	
\end{equation}
where $v(r)$ is a given bounded radial function in $C^0(M)$. Based on the primary functions in Subsection \ref{subsection primary functions}, the result below  allows us to express explicitly solutions to \eqref{equation vectors} as follows.
\begin{proposition} \label{proposition resolve vector equations}
	Let  $v(r) \in RC^0(M)$ be a bounded radial function. Assume that $(M,g)$ satisfies one of two conditions \eqref{condition on u0}. If $\Ric > 0$, assume furthermore that 
	\begin{equation} \label{vanishing condition sphere}
		\int_0^{\dM} f_v u_0 \theta \, dr = 0,
	\end{equation}
	where $u_0$ is the unique solution to \eqref{def u0} and $f_v$ is defined by \eqref{def f}.
	
	Then the vector equations \eqref{equation vectors} admit a solution $W$ given by
	\begin{equation} \label{solution to vector equations}
		W=\tfrac{n}{n - 1}F_v^\p(r)dr,
	\end{equation}
	where $F_v$ is given by Definition \ref{definition F}. Moreover, if $(M,g)$ is the sphere $\Sp^n$, then the vanishing condition \eqref{vanishing condition sphere} is necessary for the existence of solutions to \eqref{equation vectors}.
\end{proposition}
\begin{proof} 
	It is clear by Proposition \ref{proposition 1-form W calculations}(b) that $W \in C^2(M)$ defined in \eqref{solution to vector equations} is a solution to \eqref{equation vectors} if and only if
	\begin{equation} \label{vector equation 2}
		\bigg(F^\ppp_v(r)+ F^\pp_v(r)h(r)+F^\p_v(r)h^\p(r)\bigg)dr + \tfrac{n}{n-1}F^\p_v(r)\Ric(\partial r,.) = \frac{v}{2} \,dr.
	\end{equation}
	Since $(M,g)$ is harmonic, we have $\Ric(\del r, .) = \Ric dr$, and hence, \eqref{vector equation 2} may be rewritten as
	$$
	\bigg(F^\pp_v+ hF^\p_v  + \tfrac{n}{n-1}\Ric F_v \bigg)^\p=  \frac{v}{2}.
	$$
	Clearly, this equation automatically holds as Proposition \ref{proposition properties of F} asserts that
	$$
	F^\pp_v+ h(r)F^\p_v  + \tfrac{n}{n-1}\Ric F^\p_v(r) = f_v.
	$$
	
	Now, to complete the proof, we need to prove that, if $(M,g)$ is the standard sphere, the vanishing condition \eqref{vanishing condition sphere} is necessary for the existence of solutions $W \in C^2(M)$ to \eqref{equation vectors}. In fact, observing that by definition, 
	$$
	u_0\theta = - \frac{n -1}{n \Ric} \big(\theta u_0^\p \big)^\p.
	$$
	Then, we have by integral by parts
	\begin{align*}
		\int_0^{\dM} f_v u_0 \theta \, dr & = - \frac{n-1}{n \Ric }\int_0^{\dM} f_v (u_0^\p \theta)^\p \, dr \\
		& = - \frac{n-1}{n \Ric} \bigg( \big(f_v u_0^\p \theta \big) (\dM) -  \tfrac12 \int_0^{\dM} v u_0^\p \theta \, dr \bigg).
	\end{align*}
	Since, by \eqref{h at dM}, $\theta (\dM) = 0$, the first term of the right-hand side vanishes. Therefore, it follows that the condition \eqref{vanishing condition sphere} is equivalent to 
	$$
	\int_0^{\dM} v u_0^\p \theta \, dr = 0.
	$$ 
	On the other hand, remark that when $(M,g)$ is a standard sphere, we have $u_0 = \cos(r)$. Moreover, it is easy to check by Proposition \ref{proposition 1-form W calculations}(c) that 
	$$
	V_0 := u_0^\p dr = -\sin(r) dr
	$$ 
	is a CKVF on the sphere. Therefore, if \eqref{equation vectors} admits a solution $V \in C^2(M)$, we must have
	$$
	\int_Mv(r)u_0^\p(r)d\mu = \int_{M} v \langle V_0 , dr \rangle \, d\mu = \int_{M} \langle V_0, \Div(L V) \rangle \, d\mu,
	$$
	and since $\int_Mv(r)u_0^\p(r)d\mu=|\Sp^{n-1}|\int_0^{\dM} v u_0^\p \theta \, dr$, we get
	$$
	\int_0^{\dM} v u_0^\p \theta \, dr = -\frac{1}{|\Sp^{n-1}|} \int_{M} \langle LV_0, L V \rangle \, d\mu = 0.
	$$
	The proof is completed.
\end{proof}
We have solved the vector equations with a symmetric $1$-form source. Thanks to this result and our calculations on the conformal killing operator in Proposition \ref{proposition 1-form W calculations}, we now reduce the full system \eqref{CE with null sigma}. Letting $v(r)$ be a radial function such that $(v, u_0)$ satisfies all assumptions in Definition \ref{definition F}, we set
\begin{equation} \label{def Fcal}
	\Fcal_v :=  \tfrac{4n}{n-1} \Big( f_v - \tfrac{n}{n-1} h F^\p_v - \tfrac{n}{n-1}\Ric F_v \Big)^2
	-  4 \big(\tfrac{n}{n-1}\big)^2 \big( F^\p_v \big)^2 \Big( h^\p + \frac{h^2}{n-1} + \Ric \Big),
\end{equation}
with $f_v$ and $F_v$ given by \eqref{def f} and Definition \ref{definition F} respectively. The following is our main result in this section.
\begin{theorem} \label{theorem general}
	Let $(M,g)$ be either the sphere $\Sp^n$, or a harmonic manifold with $\Ric \le 0$. Let $(V, \psi(r), \tau(r), \rho(r))$ be radial data set with 
	$$
	V \in C^0(\RR), \quad \psi(r) \in C^1(M), \quad \tau(r) \in \ C^1(M) \quad \text{and} \quad \rho(r) \in C^0(M).
	$$ 
	Given a positive radial function $\varphi(r) \in C_+^3 (M)$, we write
	\begin{equation} \label{def v}
		v =\tfrac{n-1}{n}\varphi^N\tau^\p-\rho\psi^\p.
	\end{equation}
	
	Then $\varphi$ is a solution to the conformal equations \eqref{CE}  if and only if $\varphi$ satisfies the Lichnerowicz-type equation
	\begin{equation}\label{RCE}
		-\tfrac{4(n-1)}{n-2}\varphi^{N+1}(\varphi^\pp+h\varphi^\p) + \mathcal{R}_\psi\varphi^{N+2} - \mathcal{B}_{\tau,\psi}\varphi^{2N} = \Fcal_v + \rho^2,
	\end{equation}
	with $\Fcal_v$ defined in \eqref{def Fcal}. Moreover, in this case, the $1-$form $W$ can be computed by \eqref{solution to vector equations}.
\end{theorem}
\begin{proof}
	Since $(\psi(r), \rho(r), \tau(r), \varphi(r))$ are radial functions, the vector equations \eqref{veceq} can be rewritten
	\begin{equation} \label{vector equation conformal}
		\Div (LW) = v(r)dr.
	\end{equation}
	If $(M,g)$ is a harmonic manifold with $\Ric \le 0$, it follows from Proposition \ref{proposition sign of u0} that $u_0$ is positive for all $r \ge 0$. Therefore, thanks to Proposition \ref{proposition resolve vector equations}, (up to a CKVF) the $1$-form $W$ given by \eqref{solution to vector equations} is a solution of the vector equations.	Moreover, by Proposition \ref{proposition 1-form W calculations}(c), we can compute 
	$$
	|LW|^2 =  \tfrac{4n}{n-1} \Big(F^\pp_v - \frac{1}{n-1}hF^\p_v\Big)^2 - 4 \Big(\tfrac{n}{n-1} \Big)^2\Big( h^\p + \frac{h^2}{n-1} + \Ric \Big)\big(F^\p_v\big)^2.
	$$
	Since, by Proposition \ref{proposition properties of F},
	$$
	F^\pp_v =  f_v - h F^\p_v - \tfrac{n}{n-1}\Ric F_v,
	$$
	it follows by definition that
	$$
	|LW|^2 = \Fcal_v.
	$$
	Plugging this equality into the Lichnerowicz equation, we obtain that the conformal equations \eqref{CE with null sigma} is equivalent to the nonlinear single equation \eqref{RCE} as desired.
	
	\medskip
	
	In the case where $(M,g)$ is a round sphere, the process of the proof follows same way, the only difference being that we need to point out the role of the vanishing condition \eqref{vanishing condition sphere} in the theorem. In fact, it is clear by Proposition \ref{proposition resolve vector equations} that the condition \eqref{vanishing condition sphere} is necessary for the solvability of the vector equations, and hence, so is it for that of the conformal equations \eqref{CE with null sigma}.
	
	\medskip
	
	Conversely, once \eqref{vanishing condition sphere} holds, 
	$F_v$ is well-defined by Proposition \ref{proposition properties of F}. Therefore, the equivalence between \eqref{CE with null sigma} and \eqref{RCE} is also obtained similarly to arguments above. The proof is completed.
\end{proof}

\begin{remark}
	The Banach spaces used in the theorem are basic. In some concrete situations they should be strengthened dependently on which models we work. For instance, when we are interested in the AH or AF initial data setting,  we will use the weighted H\"{o}lder or Sobolev spaces, as we will see in Sections \ref{section hyperbolic} and \ref{section Euclidean}.
\end{remark}

%
%
\section{Solutions in the round sphere} \label{section Sphere}
%
%
In the previous section, we reduced the conformal equations to the Lichnerowicz-type \eqref{RCE} on a general harmonic manifold. A natural question to ask is, in concrete cases, whether we can completely solve \eqref{CE with null sigma} by addressing \eqref{RCE}. For answering this question, let us consider the problem in the simple models where 
$$ 
h^\p + \frac{h^2}{n-1} + \Ric = 0,
$$
that is $(M , g)$ is a Euclidean, Hyperbolic or spherical space (see Lemma \ref{lemma inequality h}). In this section, we will consider first the Lichnerowicz-type equation \eqref{RCE} on the round sphere, and then, in the next two sections, the remaining space cases will be treated.

\medskip

The main result of this section is the proof of Main Theorem \ref{main theorem} on the sphere, which allows us to construct a large class of explicit solutions and provides a necessary and sufficient condition for a radial positive function $\varphi$ to solve \eqref{RCE}. Building on this result, many properties of the conformal equations are derived, including the existence and nonexistence of solutions, as well as issues of stability and instability. The outcomes we obtained differ significantly from what is known for the conformal equations on compact manifolds without conformal Killing vector fields. This highlights the complexity of the system on compact models, where the behavior of solutions depends not only on the mean curvature $\tau$, but also on the underlying metric $g$. 

\medskip

The outline of this section is as follows. In Subsection \ref{subsection Lich-type on the sphere}, we present a precise statement of Main Theorem \ref{main theorem} on the sphere and provide its proof. In Subsections~\ref{subsection Non-existence on the sphere}--\ref{subsection stability and instability on the sphere}, some applications of this result will be obtained, that is the proofs of Main Theorems \ref{main theorem nonexistence near CMC}--\ref{main theorem instability}. Throughout the discussion, for the convenience of the reader, we will present the new results established on the sphere alongside the well-known results for the conformal equations on compact manifolds without CKVF, in order to highlight both the similarities and the differences between the two settings.
\subsection{Lichnerowicz-type equation on the sphere} \label{subsection Lich-type on the sphere}
Let $\Sp^n$ denote the round sphere of dimension $n$, with $n \ge 3$.  It follows that 
\begin{equation} \label{computing in sphere}
	\begin{gathered}
		\theta(r) = \sin^{n-1}(r), \quad h(r) = (n-1)\cot(r), \quad u_0(r) = \cos(r) \\
		\Ric = n-1, \quad \Scal = n(n-1).
	\end{gathered}
\end{equation}
As explained in Introduction, although $\Sp^n$ is basic, generating initial data on it is very challenging due to the presence of conformal Killing vector fields, which prevents the elliptic theories developed in \cite{HolstNagyTsogtgerel, Maxwellcompact, DahlGicquaudHumbert, NguyenFPT} from working in this case, particularly when $\tau$ is non-constant. To the best of the authors’ knowledge, the only known solutions to the conformal equations \eqref{CE} on $\Sp^n$ are the CMC ones obtained by R. Beig, P. Bizon and W. Simon \cite{BeigBizonSimon}, where the fact $d \tau = 0$ gives $LW = 0$ in the vector equations no matter whether the seed data has CKVF, and so the system \eqref{CE} becomes uncoupled.  In the aim of constructing non-CMC solutions to the conformal equations in $\Sp^n$, especially for the radial mean curvatures, thanks to Theorem \ref{theorem general}, we will study the Lichnerowicz-type equation \eqref{RCE} instead of the full system \eqref{CE}. 

\medskip

Before stating the main result of this section, we introduce some notations. From a radial data set $(V,\tau(r),\psi(r),\rho(r))\in C^0(\RR) \times RC^1(\Sp^n) \times RC^1(\Sp^n) \times RC^0(\Sp^n)$ and a radial function $\varphi(r) \in RC_+^3(\Sp^n)$, we define new radial functions $I_{V,\psi,\rho,\varphi}(r)$ and $S_{V,\psi,\rho,\varphi}(r)$ as follow:
\begin{equation} \label{I on sphere}
	I_{V,\psi,\rho,\varphi} := \tfrac{n}{n-1} \Big( 2 V(\psi) \varphi^{2N} + |\psi^\p|^2 \varphi^{N+2} + \rho^2 \Big),
\end{equation}
and
\begin{equation} \label{S on sphere}
	S_{V,\psi,\rho,\varphi} := N^2 (\varphi^\p)^2 - n^2 \varphi^2 + 2 n N \cot \varphi \varphi^\p + \frac{1}{\sin^n} \int_0^r \frac{I_{V,\psi,\rho,\varphi}}{\varphi^{2N}} (\varphi^N \sin^n)^\p  \, ds.
\end{equation}
The result playing a central role in this section is as follows.
\begin{theorem}[Solutions with freely specified conformal factor] \label{theorem existence on sphere}
	Consider a radial data set $(V,\tau(r),\psi(r),\rho(r))\in C^0(\RR) \times RC^1(\Sp^n) \times RC^1(\Sp^n) \times RC^0(\Sp^n)$ with $\rho \psi^\p \equiv 0$. For $\varphi(r) \in RC_+^3(\Sp^n)$ we have:
	\begin{enumerate}[(i)]
		\item $\varphi(r)$ is a solution to the conformal constraint equations \eqref{CE with null sigma} if and only if
		\begin{multline} \label{rce sphere}
			- 2N \varphi^{N + 1} (\varphi^\pp + h \varphi^\p ) + n^2\varphi^{N + 2} + \tau^2  \varphi^{2N} \\
			= \bigg( \frac{1}{\sin^n} \int_0^r \tau^\p \varphi^N  \sin^n \, ds \bigg)^2 + I_{V,\psi,\rho,\varphi}.
		\end{multline}
		
		\item If $\varphi(r)$ is a solution to \eqref{rce sphere}, then $S_{V,\psi,\rho,\varphi} \ge 0$.
		
		\item If $S_{V,\psi,\rho,\varphi}>0$ on $[0,\pi]$ and $(\varphi^N\sin^n(r))^\p\neq0$ a.e. on $[0,\pi]$, then $\varphi(r)$ is a solution to \eqref{rce sphere} if and only if the function $\tau$ satisfies
		$$
		\tau= \pm \frac{ S_{V,\psi,\rho,\varphi} + 2N  \big( \varphi^\pp + (n - 1) \cot \varphi^\p \big) \varphi - n^2 \varphi^2 + I_{V,\psi,\rho,\varphi} \varphi^{-N}}{2\sqrt{\varphi^N S_{V,\psi,\rho,\varphi}}}.
		$$
	\end{enumerate}
	Moreover, in this case, the $1-$form $W$ can be computed by \eqref{solution to vector equations} with $v=\tfrac{n-1}{n}\varphi^N\tau^\p$.
\end{theorem}
\begin{proof}
	Following Subsection \ref{subsection primary functions}, let us consider $f_v$, $F_v$ and $\Fcal_v$ given by \eqref{def f}, \eqref{def F compact} and \eqref{def Fcal} respectively, where  
	$$
	v =\tfrac{n-1}{n}\varphi^N\tau^\p - \rho \psi^\p.
	$$ 
	Thanks to Theorem \ref{theorem general}, we keep in mind that a radial positive function $\varphi$ is a solution to the conformal equations \eqref{CE with null sigma} if and only if it solves the Lichnerowicz-type equation \eqref{RCE}.
	
	\medskip
	
	In this regards, recall that since the geodesic spheres of $\Sp^n$ are totally umbilical, Riccati equation gives 
	$$
	h^\p(r)+\frac{h^2(r)}{n-1}+\Ric=0.
	$$ 
	So, we may rewrite $\Fcal_v$ as
	\begin{equation} \label{Fcal sphere}
		\Fcal_v =  \tfrac{4n}{n-1} \Big(  f_v - \tfrac{n}{n-1} h F^\p_v - n F_v \Big)^2.
	\end{equation}
	By equation \eqref{def F compact} defining the function $F_v$, and from \eqref{computing in sphere} we have
	\begin{align*}
		h F^\p_v & = h\frac{u_0^\p}{u_0}F_v +  \frac{h}{u_0 \theta } \int_0^r f_v u_0 \theta \, ds \\
		& =  (n - 1) \bigg( \frac{1}{n \sin^n} \int_0^r f_v \big(\sin^n\big)^\p \, ds - F_v \bigg),
	\end{align*}
	therefore, integrating by part the integral term on the right-hand side, we get  
	$$
	h F^\p_v = (n - 1) \bigg( \frac{f_v}{n} - \frac{1}{2 n \sin^n} \int_0^r v \sin^n \, ds - F_v \bigg).
	$$
	Taking into \eqref{Fcal sphere}, it follows that
	\begin{equation} \label{|LW| in sphere}
		\Fcal_v =  \tfrac{n-1}{n} \bigg(  \frac{1}{\sin^n} \int_0^r \bigg(  \varphi^N \tau^\p - \tfrac{n}{n-1}\rho \psi^\p \bigg) \sin^n \, ds \bigg)^2.
	\end{equation}
	Hence, the Lichnerowicz-type equation \eqref{RCE} becomes
	\begin{multline*}
		- \tfrac{4(n - 1)}{n - 2} \varphi^{N + 1} (\varphi^\pp + h \varphi^\p ) + \big(n(n-1) - |\psi^\p|^2 \big) \varphi^{N + 2} + \bigg(\tfrac{n-1}{n} \tau^2 - 2V(\psi)\bigg) \varphi^{2N} \\
		= \tfrac{n-1}{n} \bigg( \frac{1}{\sin^n} \int_0^r \bigg(  \varphi^N \tau^\p - \tfrac{n}{n-1}\rho \psi^\p \bigg) \sin^n \, ds \bigg)^2 + \rho^2.
	\end{multline*}
	Now, since $\rho \psi^\p \equiv 0$, by the definition of $I_{V,\psi,\rho,\varphi}$, the equation is rewritten as
	$$
	- 2N \varphi^{N + 1} (\varphi^\pp + h \varphi^\p ) + n^2\varphi^{N + 2} + \tau^2  \varphi^{2N} = \bigg( \frac{1}{\sin^n(r)} \int_0^r \tau^\p \varphi^N  \sin^n \, ds \bigg)^2 + I_{V,\psi,\rho,\varphi},
	$$
	thus proving point $(i)$.
	
	\medskip
	
	Continuing to integrate by parts the integral term on the right-hand side, we get
	\begin{multline*}
		- 2N \varphi^{N + 1} (\varphi^\pp + h \varphi^\p ) + n^2 \varphi^{N + 2} + \tau^2 \varphi^{2N} \\
		= \Big( \tau \varphi^N - \frac{1}{\sin^n} \int_0^r \tau \big( \varphi^N \sin^n \big)^\p \, ds \Big)^2 + I_{V,\psi,\rho,\varphi},
	\end{multline*}
	or equivalently
	\begin{multline} \label{eq.5a1}
		\frac{2 \tau \varphi^N}{\sin^n} \bigg(\int_0^r  \tau \big( \varphi^N \sin^n \big)^\p \, ds \bigg) - \frac{1}{\sin^{2n}} \bigg( \int_0^r  \tau \big( \varphi^N \sin^n \big)^\p \, ds \bigg)^2 \\
		= 2N \varphi^{N + 1} (\varphi^\pp + h \varphi^\p ) - n^2 \varphi^{N + 2} 
		+ I_{V,\psi,\rho,\varphi}.
	\end{multline}
	Multiplying \eqref{eq.5a1} by $\frac{( \varphi^N \sin^n)^\p}{\varphi^{2N}}$,  it follows that
	\begin{multline} \label{eq.5a2}
		\bigg(\frac{1}{\varphi^N \sin^n} \Big( \int_0^r  \tau \big( \varphi^N \sin^n\big)^\p \, ds \Big)^2 \bigg)^\p
		= \big( \varphi^N \sin^n \big)^\p \Big( 2N \big( \varphi^\pp + (n - 1)\cot \varphi^\p \big) \varphi^{- N + 1} \\
		- n^2 \varphi^{- N + 2} + I_{V,\psi,\rho,\varphi} \varphi^{-2N} \Big).
	\end{multline}	
	Now, since
	$$
	\lim_{r \to 0}  \bigg(\frac{1}{\varphi^N(r) \sin^n (r)} \Big( \int_0^r  \tau \big( \varphi^N \sin^n (s) \big)^\p \, ds \Big)^2 \bigg) = 0,
	$$
	the equation \eqref{eq.5a2} is equivalent to
	\begin{multline} \label{eq 5b}
		\Big( \int_0^r  \tau \big( \varphi^N \sin^n \big)^\p \, ds \Big)^2 =  \varphi^N \sin^n \int_0^r \big( \varphi^N \sin^n \big)^\p \Big( 2N  \big( \varphi^\pp + (n - 1) \cot \varphi^\p \big) \varphi^{- N + 1} \\
		- n^2 \varphi^{- N + 2} + I_{V,\psi,\rho,\varphi} \varphi^{-2N} \Big) \, ds.
	\end{multline}
	To simplify the right-hand side, we check at once that
	\begin{multline} \label{derivative of S}
		\big( \sin^n S_{V,\psi,\rho,\varphi} \big)^\p = \big( \varphi^N \sin^n \big)^\p \bigg( 2N  \big( \varphi^\pp + (n - 1) \cot \varphi^\p \big) \varphi^{- N + 1} \\
		- n^2 \varphi^{- N + 2} + I_{V,\psi,\rho,\varphi} \varphi^{-2N} \bigg).
	\end{multline}
	Taking this into account, Equation \eqref{eq 5b} becomes
	\begin{equation} \label{eq.5a3}
		\Big( \int_0^r  \tau \big( \varphi^N \sin^n \big)^\p \, ds \Big)^2	
		=  \varphi^N \sin^{2n} S_{V,\psi,\rho,\varphi}.
	\end{equation}
	Therefore, $S_{V,\psi,\rho,\varphi} \ge 0$ is a necessary condition for existence of solution, proving point $(ii)$.
	
	In order to prove $(iii)$, assume $S_{V,\psi,\rho,\varphi}>0$ and $(\varphi^N\sin^n)^\p\neq0$ a.e.. Following the previous computation we have that equations \eqref{eq.5a1} and \eqref{eq.5a2} are equivalent and $\varphi(r)$ is a solution if and only if equation \eqref{eq.5a3} holds. Since $\varphi^N \sin^{2n}S_{V,\psi,\rho,\varphi}\neq0$ on $\oo{0,\pi}$, we have that $\int_0^r  \tau \big( \varphi^N \sin^n (s) \big)^\p \, ds$ does not vanish and has constant sign. Taking the square root we get
	\[
	\pm\int_0^r\tau(s)\big(\varphi^N\sin^n\big)^\p(s)\,ds = \sqrt{\varphi(r)^N\sin^{2n}(r)S_{V,\psi,\rho,\varphi}(r)}.
	\]
	Taking the derivative with respect to $r$ we obtain
	\begin{align*}
		\tau & = \pm \frac{\bigg(\sqrt{\big(\varphi^N \sin^n \big) \big(\sin^n S_{V,\psi,\rho,\varphi} \big)} \bigg)^\p}{\big(\varphi^N \sin^n \big)^\p} \\
		& =\pm  \frac{\big(\varphi^N \sin^n \big)^\p \big(\sin^n S_{V,\psi,\rho,\varphi} \big) + \big(\varphi^N \sin^n\big) \big(\sin^n S_{V,\psi,\rho,\varphi} \big)^\p}{2\big(\varphi^N \sin^n\big)^\p \sin^n\sqrt{\varphi^N S_{V,\psi,\rho,\varphi}}}.	
	\end{align*}
	Finally, using \eqref{derivative of S}, we get
	\[
	\tau = \pm \frac{S_{V,\psi,\rho,\varphi} + 2N  \big( \varphi^\pp + (n - 1) \cot \varphi^\p \big) \varphi - n^2 \varphi^2 + I_{V,\psi,\rho,\varphi} \varphi^{-N}}{2\sqrt{\varphi^N S_{V,\psi,\rho,\varphi}}}.
	\]
	The proof is competed.
\end{proof}
\begin{remark}
	It is not surprising that the relation between $\varphi$ and $\tau$ be only determined up to the sign. In fact, in the constraint equations \eqref{constraint}, $\hat{k}$, $\hat{\rho}$ and $\tau=\tr(\hat{k})$ depend on the chosen orientation in time: if $(\hat{g},\hat{k},\hat{\rho})$ is a solution, then so is $(\hat{g},-\hat{k},-\hat{\rho})$. Considering the conformal approach we get: if $(\varphi,W)$ is a solution with data $(\tau,\rho)$ (and $\sigma=0$), then $(\varphi,-W)$ is a solution with data $(-\tau,-\rho)$. In our setting of radial solutions with radial data $\tau(r),\rho(r)$, this reflected by the fact that $F_{-v}=-F_v$ (for the expression of $W$) and $\Fcal_{-v}=\Fcal_v$ (for the equation of $\varphi$). Therefore it is natural that the relation we find between $\varphi$ and $\tau$ works for both couple $(\varphi,\tau)$ and $(\varphi,-\tau)$.
\end{remark}

\begin{remark}\label{remark vanishing condition on Is}
	We have shown that \eqref{vanishing condition sphere} is a necessary condition for existence of solutions to the Lichnerowicz-type equation \eqref{RCE}, and so \eqref{rce sphere}, on the sphere. Another vanishing condition for existence of solutions to the equation is
	$$
	\int_0^\pi \frac{I_{V,\psi,\rho,\varphi}}{\varphi^{2N}} (\varphi^N \sin^n(s))^\p  = 0.
	$$
	In fact, assume that it is not true. Then the condition $S_{V,\psi,\rho,\varphi} \ge 0$ give us that
	$$
	\int_0^\pi \frac{I_{V,\psi,\rho,\varphi}}{\varphi^{2N}} (\varphi^N \sin^n(s))^\p > 0,
	$$
	and hence 
	$$
	\lim_{r \to \pi} \bigg( \frac{1}{\sin^n(r)} \int_0^r \frac{I_{V,\psi,\rho,\varphi}}{\varphi^{2N}} (\varphi^N \sin^n(s))^\p  \, ds \bigg) = +\infty.
	$$
	However, taking this limit into the expression of $\tau$ by $(\rho, \psi, V, \varphi)$ above, it follows that $|\tau(\pi)| = +\infty$, which is a contradiction.
\end{remark}
%
%
%
\subsection{Non-existence of solutions} \label{subsection Non-existence on the sphere}
Let us consider the Equations \eqref{CE} on the sphere with $(V, \psi, \sigma)$ zero everywhere, that is
\begin{subequations}\label{CE with null sigma in sphere}
	\begin{align}
	- \tfrac{4(n-1)}{n-2} \Delta \varphi + \Scal_g \varphi  
	& = - \tfrac{n-1}{n} \tau^2 \varphi^{N-1}  + \frac{|LW|^2+\rho^2}{\varphi^{N + 1}} \\
	\Div(L W) & = \tfrac{n-1}{n} \varphi^N d\tau.
	\end{align}
\end{subequations}
As we have mentioned in the introduction, we can heuristically think of $\rho$ as the norm of the TT-tensor $\sigma$ in the vacuum case studied in \cite{IsenbergMoncrief, DahlGicquaudHumbert, Maxwellcompact, HolstNagyTsogtgerel}, where $(M,g)$ is a compact manifold with no CKVF. For that, according to the existence results proven in \cite{IsenbergMoncrief} and \cite{HolstNagyTsogtgerel, Maxwellcompact, DahlGicquaudHumbert}, it is easy to show that as long as $(M,g)$ has no CKVF, such an associated system \eqref{CE with null sigma in sphere} will admit a solution if one of following conditions holds:
\begin{itemize}
	\item Near--CMC existence: $\rho \ne 0$ and $\frac{\max|d\tau|}{\min \tau}$ is small,
	\item Small time-derivative existence: $\Scal_g > 0$, $\rho \ne 0$ and $\rho$ is small.
\end{itemize}
However, in the presence of the CKVF, these well-known results are no longer true in general as we will see in the following result. 
\begin{theorem}[Non-existence of solution in the near-CMC regime] \label{theorem non-existence near CMC}
	Let $\tau(r)\in RC^1(\Sp^n)$ be radial function satisfying $\tau^\p \not \equiv 0$ and  $\tau^\p \ge 0$. Let $C>0$ be a large constant such that $\frac{\max|d \tau|}{\min \tau +C}$ is sufficiently small.
	
	Then for any $\rho(r)\in RC^0(\Sp^n)$, the conformal equations \eqref{CE with null sigma in sphere} associated with $(\rho(r), \tau(r) + C)$ have no solution.
\end{theorem}
\begin{proof}
	Since  $\frac{\max|d \tau|}{\min \tau +C}$ is small, by the same arguments used for the uniqueness of solutions as in Isenberg--Moncrief \cite{IsenbergMoncrief}, we get that the conformal equations \eqref{CE with null sigma in sphere} associated with  $(\rho, \tau + C)$  have at most one solution $\varphi > 0$. Therefore, since Laplace's operator $\Delta$ and the conformal one $\Div(L)$ are invariant under rotations, it follows that as long as $(\rho, \tau)$ are radial, if a solution $\varphi$ to \eqref{CE with null sigma in sphere} exists, then it must be radial, otherwise, any rotation of $\varphi$ is also a solution, and so the set of solutions is infinite, a contradiction. Consequently, as we have shown above, existence of solutions to \eqref{CE with null sigma in sphere} in this case will lead to existence of solutions to the Lichnerowicz-type equation \eqref{rce sphere} associated with the data set $(V \equiv 0, \psi \equiv 0, \rho(r), \tau(r) + C)$. In particular, by Proposition \ref{proposition resolve vector equations}, we deduce from the vanishing condition \eqref{vanishing condition sphere} on the sphere that if such a radial solution $\varphi$ exists, we must have 
	$$
	\int_0^\pi (\tau + C)^\p \varphi^N \sin^n(r) dr = 0.
	$$
	However, since $\tau^\p \not \equiv 0$ and $\tau^\p \ge 0$, we get 
	$$
	\int_0^\pi (\tau + C)^\p \varphi^N \sin^n(r) dr = \int_0^\pi \tau^\p \varphi^N \sin^n(r) dr > 0
	$$
	for all $\varphi > 0$. Therefore, we conclude that the Lichnerowicz-type equation \eqref{rce sphere} has no solution, and so, neither do the conformal equations \eqref{CE with null sigma in sphere}. The proof is completed.
\end{proof}

We next discuss the possibility of the existence of solutions to the vacuum conformal equations with the TT-tensor $\sigma \equiv 0$, that is the Equations \eqref{vacuum CE null sigma}. This question was posed first by D. Maxwell in \cite{Maxwell: Modelproblem} and appears to have a major impact on the solvability of \eqref{CE} in general as proven in \cite{NguyenNonexistence}. When $(M,g)$ has no CKVF, it is well-known that \eqref{vacuum CE null sigma} admit no solution in the CMC and near--CMC regimes. However, in the far--from--CMC case, an example of the existence of solutions to \eqref{vacuum CE null sigma} is given by the second author in \cite{NguyenNonexistence}. More precisely, the result states that as long as $(M,g)$ is compact and has no CKVF, if $\tau > 0$ satisfies the inequality
\begin{equation} \label{non-generic condition}
	\bigg | L\bigg(\frac{d\tau}{\tau}\bigg) \bigg| \le c \bigg |\frac{d\tau}{\tau}\bigg|^2
\end{equation}
with some constant $c>0$ depending only on $g$, then for all sufficiently large constant $a > 0$,  the vacuum equations \eqref{vacuum CE null sigma} associated with $\tau^a$ have at least one solution. Motivated by the models in \cite{NguyenNonexistence}, J. Dilts, M. Holst, T.~Kozareva and D. Maxwell \cite{DiltsHolstKozarevaMaxwell} obtain numerical results which suggest that existence of solutions to vacuum Equations \eqref{CE} with null $\sigma$ may still hold for a large class of far-from-CMC data sets, without the non-generic inequality \eqref{non-generic condition}. 

\medskip

In the following result, we will give a different perspective on the conformal equations on the sphere. In fact, in the presence of CKVF makes the above results are no longer valid. Otherwise, it would lead to an unexpected situation where the set of solutions to the vacuum conformal equations with null $\sigma$ is infinite. To see this, we first remark that the set of all radial mean curvatures $\tau$ on the sphere well covers the non-generic inequality considered in \cite{NguyenNonexistence, NguyenProgress}, that is, there exists a class of radial mean curvatures satisfying \eqref{non-generic condition}. A simple example of this is $\tau_0 = \cos(r) + 2$. By straightforward calculations, since 
$$
 h^\p(r) + \frac{h(r)^2}{n-1} + \Ric(\dr,\dr) = 0,
$$
it follows from Proposition \ref{proposition 1-form W calculations}(c) that
\begin{align*}
	\bigg|L\bigg(\frac{d\tau_0}{\tau_0}\bigg) \bigg|  &= \bigg|L\bigg(\frac{\tau_0^\p}{\tau_0} dr \bigg) \bigg|
	\\
	& = 2 \sqrt{\tfrac{(n-1)}{n}}\bigg|\frac{\tau_0^\pp}{\tau_0} - \bigg(\frac{\tau_0^\p}{\tau_0} \bigg)^2 - \cot(r) \frac{\tau_0^\p}{\tau_0}\bigg|
	\\
	&= 2 \sqrt{\tfrac{(n-1)}{n}} \bigg(\frac{\tau_0^\p}{\tau_0} \bigg)^2.
\end{align*}
Therefore, the inequality \eqref{non-generic condition} automatically holds for $\tau_0$ provided that $c \ge 2$. We now prove the non-existence of radial solutions to the vacuum conformal equations with $\sigma \equiv 0$ as follows.
\begin{theorem}[Non-existence of spherically symmetric solution] \label{theorem no solution without TT-tensor}
	For any $\tau(r) \in RC^1(\Sp^n)$, the vacuum conformal equations \eqref{vacuum CE null sigma} have no radial solution. In particular, if a solution to \eqref{vacuum CE null sigma} exists, the set of solutions is infinite.
\end{theorem}
\begin{proof}
	Let us consider the Lichnerowicz-type equation \eqref{rce sphere} with $(\psi, V, \rho)$ zero everywhere and $\tau \in RC^1(\Sp^n)$. We first prove that this equation has no radial solution $\varphi >0$. To do so, we observe by Theorem \ref{theorem existence on sphere} that $S_{V,\psi,\rho,\varphi} \ge 0$ is a necessary condition for $\varphi$ to be a solution to \eqref{rce sphere}. Therefore, the claim will follow as long as we can show that this condition never holds in $[0, \pi]$ in our case. In fact, since $V\equiv0$, $\psi\equiv0$ and $\rho\equiv0$, the condition $S_{V,\psi,\rho,\varphi} \ge 0$ may be rewritten as
	\begin{equation} \label{ineq absence of k}
		N^2 \bigg(\frac{\varphi^\p(r)}{\varphi(r) } \bigg)^2  + 2 n N \cot(r) \bigg(\frac{\varphi^\p(r)}{\varphi(r) } \bigg) - n^2 \ge 0.
	\end{equation}
	Since the quadratic equation $N^2 x^2 + 2n N \cot (r) x - n^2 = 0$ has two solutions
	$$
	x_ 1 = - \frac{n}{N} \bigg(\cot (r) + \frac{1}{\sin (r)} \bigg), \qquad x_2 = - \frac{n}{N} \bigg(\cot (r) - \frac{1}{\sin (r)} \bigg),
	$$
	by the continuity of $\frac{\varphi^\p}{\varphi}$, the inequality \eqref{ineq absence of k} implies that either
	$$
	\frac{\varphi^\p(r)}{\varphi(r)} \le - \frac{n}{N} \bigg(\cot (r) + \frac{1}{\sin (r)} \bigg) \quad \text{for all $r \in [0, \pi]$}
	$$
	or
	$$
	\frac{\varphi^\p(r)}{\varphi(r)} \ge - \frac{n}{N} \bigg(\cot (r) - \frac{1}{\sin (r)} \bigg) \quad \text{for all  $r \in [0, \pi]$}.
	$$
	However, this is impossible since
	$$
	\lim_{r \to 0} \bigg(\cot (r) + \frac{1}{\sin (r)} \bigg) = + \infty \quad \text{and} \quad \lim_{r \to \pi} \bigg(\cot (r) - \frac{1}{\sin (r)} \bigg) = - \infty,
	$$
	which contradicts the regularity of $\frac{\varphi^\p}{\varphi}$ at $0$ and $\pi$. 
	
	\medskip
	
	Now, to show the theorem, it suffices to show that if the vacuum conformal equations \eqref{vacuum CE null sigma} admit a solution $\varphi > 0$, then it must be non-radial and the set of solutions to \eqref{vacuum CE null sigma} is infinite. In fact, suppose that $\varphi$ is radial. Since $\tau$ is radial, it follows from Theorem \ref{theorem existence on sphere} that $\varphi$ is also a solution to \eqref{rce sphere}. However, this cannot happen as we have shown above. So, we can deduce that $\varphi$ is non-radial. Now, since Laplace's operator $\Delta$ and the conformal one $\Div(L)$ are invariant under rotations and since $\tau$ is radial, any rotation of $\varphi$ is also a solution to \eqref{vacuum CE null sigma}, hence, the set of solutions to \eqref{vacuum CE null sigma} is infinite as claimed. The proof is completed.
\end{proof} 

\subsection{Existence of solutions} \label{subsection existence of solutions on the sphere}
As we have shown in Theorem \ref{theorem non-existence near CMC}, due to the vanishing condition \eqref{vanishing condition sphere}, not every choice of $\tau$ gives a solution to the equations~\eqref{CE with null sigma in sphere}, even in the near-CMC case with $\rho \not \equiv 0$. Therefore, a natural question is how to find a class of radial mean curvatures $\tau$ such that \eqref{CE with null sigma in sphere} admit at least one solution. In this subsection, we will address this problem. With regard to what we have done in compact manifold having no CKVF, the main idea in our analysis is also constructing a continuous and compact operator which gives solutions to \eqref{CE with null sigma in sphere} throughout its fixed points. The biggest challenge here is how to construct such an operator while ensuring the vanishing condition \eqref{vanishing condition sphere} in the procedure. For addressing this difficulty, let us decompose $\tau$ as
\begin{equation} \label{tau decompose}
	\tau = \tau_0 - \alpha \tau_1
\end{equation}
where $\tau_0, \tau_1$ are some radial compatible functions with $\tau$, and $\alpha$ is a constant. Clearly, if $\varphi$ is a radial solution to \eqref{CE with null sigma in sphere}, the vanishing condition \eqref{vanishing condition sphere} becomes
$$
\int_0^\pi \tau_0^\p \varphi^N \sin^n dr = \alpha \int_0^\pi \tau_1^\p \varphi^N \sin^n dr,
$$
which determines $\alpha$ as
$$
\alpha = \frac{\int_0^\pi \tau_0^\p \varphi^N \sin^n dr}{\int_0^\pi \tau_1^\p \varphi^N \sin^n dr}.
$$
This identity combined with the finding in Theorem \ref{theorem non-existence near CMC} indicates that for fixed $\tau_0$ and $\tau_1$, not every choice of $\alpha$ yields a function $\tau$ of the form \eqref{tau decompose} for which the equations \eqref{CE with null sigma in sphere} admit a solution. Motivated by this observation, and as mentioned in Introduction, rather than seeking a class of radial mean curvatures that produce solutions to \eqref{CE with null sigma in sphere}, we fix two radial functions $(\tau_0, \tau_1)$, and instead look for the constant $\alpha$ that makes the corresponding equations associated with $\tau$ in \eqref{tau decompose} solvable. For this purpose, let us define $T$ from $RC^0(\Sp^n)$ into itself as follows. Let $\tau_0(r), \tau_1(r) \in RC^1(\Sp^n)$ be two radial functions with $\tau_1$ non-constant and 
\begin{equation} \label{tau1 and tau0}
	|\tau_0^\p| \le c \tau_1^\p
\end{equation}
for some constant $c > 0$. For any $\phi(r) \in RC^0(\Sp^n)$, we set
\begin{equation} \label{def T1}
	T_1 (\phi)(r) := \frac{1}{\sin^n(r)} \int_0^r \big(\tau_0 - \alpha (\phi) \tau_1 \big)^\p |\phi|^N \sin^n \, ds,
\end{equation}
where $\alpha(\phi)$ is the constant defined by
\begin{equation} \label{def of alpha}
	\alpha(\phi) :=  \left\{\begin{array}{lr}
		\frac{\int_0^{\pi} \tau_0^\p |\phi|^N \sin^n \, dr}{\int_0^{\pi} \tau_1^\p |\phi|^N \sin^n \, dr} & \text{if $\tau_1^\p|\phi|^N \not \equiv 0$,}
		\\
		0 &\text{otherwise.}
	\end{array}\right.
\end{equation}
As long as such a $T_1(\phi)$ is well-defined (we will prove it later), in view of the Lichnerowicz equation, there exists a unique $\varphi(r) \in RC^2(\Sp^n)$ such that
\begin{multline} \label{T1 to T}
	- \tfrac{4(n-1)}{n-2} \varphi^{N + 1} (\varphi^\pp + h \varphi^\p )+ n(n-1)\varphi^{N + 2}  \\
	+ \tfrac{n-1}{n} \alpha_\star(\phi) \big(\tau_0 - \alpha(\phi) \tau_1 \big)^2 \varphi^{2N} = \tfrac{n-1}{n} \big(T_1(\phi)\big)^2 + \rho^2,
\end{multline}
where $\alpha_\star (\phi)$ is given by
\begin{equation}\label{def alpha_star}
	\alpha_\star (\phi) := \frac{\|\tau_1^\p |\phi|^N \|_{C^0}}{\max \{\|\tau_1^\p |\phi|^N \|_{C^0},\, \|\tau_1^\p \varphi_\star^N \|_{C^0}\}}
\end{equation}
with $\varphi_\star(r) \in RC^2_+(\Sp^n)$ the unique solution to the equation
\begin{equation} \label{def varphi_star}
	\tfrac{4(n-1)}{n-2} \varphi_\star^{N + 1} (\varphi_\star^\pp + h \varphi_\star^\p ) + n(n-1)\varphi_\star^{N + 2} + \tfrac{n-1}{n} (2c\tau_1)^2 \varphi_\star^{2N} = \rho^2.
\end{equation}
We define
\begin{equation}\label{definition of T}
	T(\phi) := \varphi.
\end{equation}
Note that, as we will see below, the presence of $\alpha_\star(.)$ here is just a dispensable technique for ensuring the continuity of $T$ at $\phi$, when $\phi$ satisfies $\tau_1^\p \phi \equiv 0$. In practice, we may simply regard it as $1$ by the following reason. In fact, similarly to that on a compact manifold having no CKVF, our approach is identifying fixed points of $T$ with solutions to the conformal \eqref{CE with null sigma}. So, when $\phi$ is a fixed point of $T$, it follows that
\begin{multline*}
	- \tfrac{4(n-1)}{n-2} \phi^{N + 1} (\phi^\pp + h \phi^\p ) + n(n-1)\phi^{N + 2} \\
	+ \tfrac{n-1}{n} \alpha_\star(\phi) \big(\tau_0 - \alpha(\phi) \tau_1 \big)^2 \phi^{2N} = \tfrac{n-1}{n} \big(T_1(\phi)\big)^2 + \rho^2.
\end{multline*}
Observing that by \eqref{tau1 and tau0}, we have 
\begin{equation} \label{alpha bounded}
	|\alpha (\phi)| \le  c \quad \text{and} \quad |\tau_0 - \alpha(\phi) \tau_1| \le 2c |\tau_1|.
\end{equation}
Therefore, by the definition of $\varphi_\star$, it follows that $\phi$ is a super-solution to \eqref{def varphi_star}, and hence, $\varphi_\star \le \phi$. By definition, this gives us that
\begin{equation}\label{alpha_star = 1}
	\alpha_\star(\phi) = 1
\end{equation} 
as desired. The following result ensures the compactness of $T_1$ and $T$.
\begin{proposition} \label{proposition T1 compact}
	$T_1$ and $T$ are continuous compact operators from $C^0$ into itself.
\end{proposition}
\begin{proof}
	We first show that $T_1(\phi)(r) \in RC^0(\Sp^n)$ for all $\phi(r) \in RC^0(\Sp^n)$. In fact, from \eqref{tau1 and tau0} and \eqref{def of alpha}, we have
	\begin{equation} \label{T1 vanishing codition}
		\int_0^\pi \bigg(\tau_0^\p - \alpha (\phi) \tau_1^\p \bigg) \phi^N \sin^n \, ds = 0.
	\end{equation}
	Then, it follows from l'Hôpital's rule that  for any $a \in \{0,\pi\}$
	\begin{equation}\label{T1 in C0}
		\lim_{r \to a} T_1(\phi) = \lim_{r \to a}\frac{\Big(\int_0^r \big(\tau_0 - \alpha (\phi) \tau_1 \big)^\p \phi^N \sin^n \, ds\Big)^\p}{(\sin^n)^\p(r)} = 0.
	\end{equation}
	By definition, this means that $T_1$ is well-defined and $T_1(\phi)(r) \in RC^0(\Sp^n)$ as claimed. 
	
	\medskip
	
	Next, in order to verify $(T_1(\phi))^\p(r) \in RC^0(\Sp^n)$, differentiating $T_1(\phi)$ on $(0,\pi)$,  we get
	\begin{equation} \label{compute T1'}
		(T_1(\phi) )^\p(r) = \big(\tau_0 - \alpha (\phi) \tau_1 \big)^\p (r) \phi^N(r) - \frac{n\cos(r)}{\sin^{n+1}(r)} \int_0^r \big(\tau_0 - \alpha (\phi) \tau_1 \big)^\p \phi^N \sin^n \, ds,
	\end{equation}
	Since $\tau_0(r), \tau_1(r) \in RC^1(\Sp^n)$, we have $(\tau_0 - \alpha (\phi) \tau_1 \big)^\p(0) = (\tau_0 - \alpha (\phi) \tau_1 \big)^\p(\pi) = 0.
	$ 
	Therefore, also using l'Hôpital's rule to \eqref{compute T1'}, we obtain that for all $a \in \{0, \pi\}$
	\begin{align} \label{T1 in C1}
		\lim_{r \to a} (T_1(\phi))^\p & = \lim_{r \to a}\bigg( \big(\tau_0 - \alpha (\phi) \tau_1 \big)^\p (r) \phi^N(r) \notag \\
		& \qquad \qquad + n\cos(r)\frac{\big(\int_0^r \big(\tau_0 - \alpha (\phi) \tau_1 \big)^\p \phi^N \sin^n \, ds\big)^\p}{(\sin^{n+1})^\p(r)} \bigg) \notag \\
		& = \lim_{r \to a}\Big( \big(\tau_0 - \alpha (\phi) \tau_1 \big)^\p (r) \phi^N(r) + \tfrac{n}{n+1}(\tau_0 - \alpha (\phi) \tau_1 \big)^\p (r) \phi^N (r) \Big) \notag \\
		& = 0.
	\end{align}
	Combining \eqref{compute T1'} and \eqref{T1 in C1}, it follows that $T_1(\phi)(r) \in RC^1(\Sp^n)$.  
	
	\medskip
	
	For the continuity of $T_1$ on $RC^0(\Sp^n)$, since $\alpha(.)$ is continuous at all $\phi \in RC^0(\Sp^n)$ satisfying $\tau_1^\p |\phi|^N \not \equiv 0$, by definition, so does $T_1$ at these $\phi$. Now, given $\phi \in RC^0(\Sp^n)$ with $\tau_1^\p |\phi|^N \equiv 0$, we have by \eqref{tau1 and tau0} 
	$$
	\tau_0^\p |\phi|^N \equiv 0.
	$$ 
	Therefore, letting $\{\phi_i \}$ be a sequence of radial continuous functions converging to this $\phi$,  it follows
	$$
	\big(\tau_0^\p |\phi_i|^N, \tau_1^\p |\phi_i|^N \big) \to \big(\tau_0^\p |\phi|^N, \tau_1^\p |\phi|^N \big) = (0,0) \quad \text{as} \quad i \to +\infty,
	$$
	and hence, by the definition \eqref{def T1},
	$$
	T_1(\phi_i) \to T_1(\phi) = 0 \quad \text{as} \quad i \to +\infty.
	$$ 
	This means that $T_1$ is continuous on $RC^0(\Sp^n)$.
	
	\medskip
	
	Next, to show $T_1$ is compact, we first observe that by the mean value theorem
	\begin{itemize}
		\item if $r \in [0,\pi/2]$, then there exists $\xi_r \in [0,r]$ satisfying
		\begin{equation}\label{xi}
			\int_0^r \big(\tau_0 - \alpha (\phi) \tau_1 \big)^\p \phi^N \sin^n \, ds = r\big(\tau_0 - \alpha (\phi) \tau_1 \big)^\p(\xi_r) \phi^N(\xi_r) \sin^n(\xi_r);
		\end{equation}
		
		\item if $r \in [\pi/2, \pi]$, then in view of \eqref{T1 vanishing codition}, there exists $\xit_r \in [r, \pi]$ such that
		\begin{align}\label{xit}
			\int_0^r \big(\tau_0 - \alpha (\phi) \tau_1 \big)^\p \phi^N \sin^n \, ds & = - \int_r^\pi \big(\tau_0 - \alpha (\phi) \tau_1 \big)^\p \phi^N \sin^n \, ds \notag \\
			& = - (\pi-r)\big(\tau_0 - \alpha (\phi) \tau_1 \big)^\p(\xit_r) \phi^N(\xit_r) \sin^n(\xit_r).
		\end{align}
	\end{itemize}
	Since
	$$
	\begin{cases}
	\max \Big(\sin(\rt), \, \frac{2r}{\pi} \Big) \le \sin(r) & \text{for all} \quad 0 \le \rt \le r \le {\pi/2},
	\\
	\max \Big(\sin(\rt), \, \frac{2(\pi-r)}{\pi} \Big) \le \sin(r) &\text{for all} \quad \pi/2 \le r \le \rt \le \pi,
	\end{cases}
	$$
	we get from (\ref{xi}--\ref{xit}) that 
	\begin{equation} \label{T1 is bounded}
		\bigg| \frac{1}{\sin^{n+1}(r)} \int_0^r \big(\tau_0 - \alpha (\phi) \tau_1 \big)^\p \phi^N \sin^n \, ds \bigg| \le \frac{\pi}{2} \big| \big(\tau_0 - \alpha (\phi) \tau_1 \big)^\p(\xi) \phi^N(\xi) \big|,
	\end{equation}
	where the value of $\xi$ is assigned as $\xi_r$ or $\xit_r$, depending whether $r$ is smaller or greater than $\pi/2$, respectively. 
	
	\medskip
	
	Taking into the definition of $T_1$ and \eqref{compute T1'}, we have
	\begin{subequations}
		\begin{align}
			|T_1(\phi)| &\le \frac{\pi}{2} \max\bigg(  \big|\big(\tau_0 - \alpha (\phi) \tau_1 \big)^\p |\phi|^N \bigg), \label{es T1 ogirin} \\
			|(T_1(\phi))^\p| &\le \Big(\frac{n\pi}{2} + 1 \Big) \max\bigg(  \big|\big(\tau_0 - \alpha (\phi) \tau_1 \big)^\p |\phi|^N \bigg). \label{es T1' origin}
		\end{align}
	\end{subequations}
	Combined with \eqref{alpha bounded}, we obtain
	\begin{subequations}
		\begin{align}
			|T_1(\phi)| &\le c\pi \max\bigg(\tau_1^\p |\phi|^N \bigg), \label{es T1} \\
			|(T_1(\phi))^\p| &\le c(n \pi + 2 ) \max\bigg(\tau_1^\p |\phi|^N \bigg), \label{es T1'}
		\end{align}
	\end{subequations}
	and so,
	\begin{align*}
		\|T_1(\phi)\|_{C^1} & \le c\big((n+1)\pi + 2 \big) \max\big(\tau_1^\p |\phi|^N \big) \\
		& \le c\big((n+1)\pi + 2 \big) \|\tau_1^\p\|_{C^0}\|\phi\|_{C^0}^N.
	\end{align*}
	Therefore, thanks to the Sobolev embedding, we conclude that $T_1$ is compact as desired.
	
	\medskip
	
	To prove that $T$ is continuous and compact, we observe that similarly to continuity of $T_1$, the operator 
	$$
	RC^0(\Sp^n) \ni \phi \longmapsto \alpha_\star(\phi) \big(\tau_0 - \alpha(\phi) \tau_1 \big) \in RC^0(\Sp^n)
	$$ 
	is also continuous in $RC^0(\Sp^n)$. Therefore, since $T(\phi)$ is just the image of $T_1(\phi)$ by the Lichnerowicz equation \eqref{Liceq}, thanks to the arguments in \cite{Maxwellcompact}, $T$ is continuous and compact as claimed. The proof is completed.
\end{proof}

Based on the operator $T$ defined above, we first establish the existence results for solutions to \eqref{CE with null sigma in sphere} on the sphere, which are similar to the near--CMC and small TT-tensor existence results proven for compact manifolds without CKVF. We remark that these results will be generalized in Theorems \ref{theorem limit equation criterion sphere} and \ref{theorem smooth time-derivative sphere} below, but for the convenience of the readers who may not be familiar with techniques used in this field, let us begin by considering the following basic case.
\begin{theorem}[Existence of solutions in the near--CMC and small time-derivative cases] \label{theorem near CMC sphere}
	Let $(\rho(r), \tau_0(r), \tau_1(r)) \in RC^0(\Sp^n) \times \big(RC^1(\Sp^n)\big)^2$ be radial functions with $\rho \not \equiv 0$ and  $|\tau_0^\p| \le c\tau_1^\p$ for some constant $c>0$. Assume that one of following conditions holds
	\begin{enumerate}[(i)]
		\item $\sup_{a \in [-c, c]} \bigg(\frac{\max |(\tau_0 - a \tau_1)^\p|}{\min |(\tau_0 - a \tau_1)|} \bigg)<\frac2\pi$ 
		\item $\max |\rho|$ is sufficiently small depending only on $(\tau_1, c)$.
	\end{enumerate}
	
	Then there exists a constant $\alpha \in [-c,c]$ such that the conformal equations \eqref{CE with null sigma in sphere} associated with $(\rho, \tau = \tau_0 - \alpha \tau_1)$ admit at least one solution $(\varphi(r),u(r)dr)$ with $(\varphi(r),u(r))\in RC_+^3(\Sp^n) \times RC^2(\Sp^n)$. Moreover, in this case, $\alpha$ can be expressed by
	$$
	\alpha = \frac{\int_0^{\pi} \tau_0^\p \varphi^N \sin^n \, dr}{\int_0^{\pi} \tau_1^\p \varphi^N \sin^n \, dr}.
	$$
\end{theorem}
\begin{proof}
	Let $T:RC^0(\Sp^n) \to RC^0(\Sp^n)$ be defined as in \eqref{definition of T}. In regard to \eqref{alpha_star = 1}, it is clear from Theorem \ref{theorem existence on sphere} that to show existence of solutions to \eqref{CE with null sigma in sphere}, it suffices to show that $T$ admits a fixed point $\phi >0$. 
	
	\medskip
	
	We will first prove (ii). Given $\eps > 0$, we set
	$$
	\Omega_\eps := \{ \phi(r) \in RC^0(\Sp^n) \ |\ 0 \le \phi \le \eps\}.
	$$
	For all $\phi(r) \in \Omega_\eps$, we have by \eqref{alpha bounded} and \eqref{T1 is bounded}
	\begin{align*}
		\bigg(\frac{1}{\sin^n(r)} \int_0^r \big(\tau_0 - \alpha (\phi) \tau_1 \big)^\p \phi^N \sin^n\, ds \bigg)^2 
		& \le  \frac{\pi^2}{4} \big\| \big(\tau_0 - \alpha (\phi) \tau_1 \big)^\p  \big\|^2_{C^0} \|\phi\|_{C^0}^{2N} \\
		& \le c^2 \pi^2 \| \tau_1^\p \|^2 \eps^{2N}.
	\end{align*} 
	Therefore, as long as $\eps$ is sufficiently small such that 
	$$
	n(n-1) \eps^{N + 2} >  c^2 \tfrac{n-1}{n} \pi^2 \| \tau_1^\p \|^2 \eps^{2N},
	$$
	for any $\rho$ satisfying
	$$
	\rho^2 \le n(n-1) \eps^{N + 2} - c^2\tfrac{n-1}{n} \pi^2 \| \tau_1^\p \|^2 \eps^{2N},
	$$
	we obtain that
	$$
	n(n-1) \eps^{N + 2} + \tfrac{n-1}{n}\alpha_\star(\phi)\big(\tau_0 - \alpha(\phi) \tau_1 \big)^2 \eps^{2N} >  c^2 \tfrac{n-1}{n} \pi^2 \| \tau_1^\p \|^2 \eps^{2N} + \rho^2
	$$
	for all $r \in  [0, \pi]$. This means that the constant function $\eps$ is a super-solution to the Lichnerowicz equation \eqref{T1 to T}, and so $T(\phi) \le \eps$. In other words, we have $T(\Omega_\eps) \subset \Omega_\eps$. Therefore, by Schauder's fixed point theorem, it follows that $T$ has a fixed point in $\Omega_\eps$, which completes (ii).
	
	\medskip
	
	Next, we will prove (i) by using the Leray-Schauder's fixed point instead of Schauder's fixed point as in (ii). In fact, assume that $T$ has no fixed point. Then, according to Leray-Schauder's fixed point, there exists a sequence $\{(t_i,\phi_i)\} \subset [0,1)\times C^0(\Sp^n)$ such that 
	$$
	\phi_i = t_i T(\phi_i) \quad \text{and} \quad \max|T(\phi_i)| \to +\infty.
	$$
	By definition, this fact can be rewritten in view of \eqref{T1 to T} as
	\begin{multline} \label{eq T(varphi_i)}
		- \tfrac{4(n-1)}{n-2} \varphi_i^{N + 1} (\varphi_i^\pp + h \varphi_i^\p )+ n(n-1)\varphi_i^{N + 2} \\
		+ \tfrac{n-1}{n} \alpha_\star(\phi) \big(\tau_0 - \alpha(\phi) \tau_1 \big)^2 \varphi_i^{2N} = \tfrac{n-1}{n} \big(T_1(\phi_i)\big)^2 + \rho^2,
	\end{multline}
	where $\varphi_i := T(\phi_i)$ and  $\max|\varphi_i | \to +\infty$.
	
	\medskip
	
	We observe that by \eqref{es T1}, if $\{\|\tau_1^\p \phi_i^N \|_{C^0} \}$  is bounded, so is $\{\|T_1(\phi_i)\|_{C^0}\}$. Then, letting $\varphi_\infty$ be the unique solution to the equation
	$$
	- \tfrac{4(n-1)}{n-2} \varphi_\infty^{N + 1} (\varphi_\infty^\pp + h \varphi_\infty^\p )+ n(n-1)\varphi_\infty^{N + 2} 
	= \tfrac{n-1}{n} \big(\sup_{i}\|T_1(\phi_i)\|_{C^0}\big)^2 + \rho^2,
	$$
	we obtain that $\varphi_\infty$ is a supersolution to \eqref{eq T(varphi_i)}, and so $\varphi_i \le \varphi_\infty$, which contradicts $\max|\varphi_i | \to +\infty$. Therefore, we must have $\|\tau_1^\p \phi_i^N \|_{C^0} \to +\infty$. In particular, it implies by definition that 
	\begin{equation} \label{alpha_star equals 1}
		\alpha_\star (\phi_i) = 1.
	\end{equation} 
	Now, letting $r_i \in [0,\pi]$ be a maximum point of $\varphi_i$, since
	$$
	- (\varphi_i^\pp + h \varphi_i^\p )(r_i) = - \Delta \varphi_i (r_i) \ge 0,
	$$
	it follows from \eqref{eq T(varphi_i)} and \eqref{alpha_star equals 1} that
	$$
	(\tau_0 - \alpha(\phi)\tau_1)^2(r_i)\varphi_i^{2N} (r_i) \le (T_1(\phi_i))^2(r_i) + \tfrac{n}{n-1} \rho^2(r_i).
	$$
	Combined with \eqref{es T1 ogirin} and the fact that $\phi_i = t_i \varphi_i$, we get
	$$
	(\tau_0 - \alpha(\phi)\tau_1)^2(r_i) \varphi_i^{2N} (r_i) \le  \frac{t_i^{2N}\pi^2}{4} \max \big| \big(\tau_0 - \alpha (\phi) \tau_1 \big)^\p \big|^2 \varphi_i^{2N}(r_i) + \tfrac{n}{n-1} \rho^2(r_i).
	$$
	However, since $\alpha(\phi) \in [-c, c]$ and since we assumed that
	$$
	\sup_{a \in [-c, c]} \Big(\frac{\max |(\tau_0 - a \tau_1)^\p|}{\min |(\tau_0 - a \tau_1)|} \Big) < \frac{2}{\pi},
	$$ 
	the inequality contradicts the fact that $\varphi_i (r_i) \to +\infty$. Therefore, $T$ must have a fixed point. The proof is completed.
\end{proof}

Besides the near-CMC and small TT-tensor results, another important achievement for the vacuum conformal equations on compact manifolds having no CKVF is the limit equation criterion proven by M. Dahl, R. Gicquaud and E.~Humbert \cite{DahlGicquaudHumbert}, which states that, given a compact manifold without CKVF, if $\tau$ does not change sign and if the limit equation
\begin{equation} \label{limit equation}
	\Div(L W) = \lambda \sqrt{\tfrac{n-1}{n}} |LW| \frac{d\tau}{\tau}, 
\end{equation}	
has no nontrivial $W$ for all $\lambda \in (0,1]$, then the vacuum conformal equations \eqref{CE} admit at least one solution $(\varphi, W)$. Moreover, in this case, the set of solutions to \eqref{CE} is compact. The results in \cite{NguyenNonexistence, NguyenProgress, DiltsHolstKozarevaMaxwell} also show that the limit equation \eqref{limit equation} appears to play a central role in the existence of solutions to the conformal equations, at least for $\tau$ not changing sign. Technically, as we may see in the proof of the main result in \cite{DahlGicquaudHumbert} or \cite{NguyenFPT}, a nontrivial solution $W$ to the limit equation \eqref{limit equation} arises from a blow-up phenomenon of a sequence of solutions $\{\varphi_i \}$, for which, after an appropriate rescaling, one obtains
\begin{equation} \label{convergence to LW}
	\bigg( \frac{\varphi_i}{\|\varphi_i\|_{L^\infty}}\bigg)^N \longrightarrow \sqrt{\tfrac{n}{n-1}} \frac{|LW|}{\tau}.
\end{equation}
Therefore, in our setting of a radial data set and considering radial solutions, we have that the functions $\varphi_i$ and $\tau$ in \eqref{convergence to LW} are radial. Therefore, the convergence tells us that $|LW|$ is also radial, and hence by Proposition \ref{proposition resolve vector equations}, $W$ must be of the form $W = u(r) dr$, for some radial function $u(r)$. However, in this case, the following proposition shows that such a non-trivial solution $W$ does not exist.
\begin{proposition}\label{proposition limit equation insight}
	Assume that $\tau(r) \in RC^1(\Sp^n)$. Then the limit equation \eqref{limit equation} has no nontrivial $W$ of the form $W = u(r)dr$ with $u(r)\in RC^2(\Sp^n)$.
\end{proposition}
\begin{proof}
	We argue by contradiction. Assume that the proposition is not true. Then, there exists a $1-$form $W = u(r)dr$ with $u \in RC^2(\Sp^n) \setminus \{0\}$ such that
	\[
	\Div(L W) = \lambda \sqrt{\tfrac{n-1}{n}} |LW| \frac{d\tau}{\tau(r)} = \lambda \sqrt{\tfrac{n-1}{n}} |LW| \frac{\tau^\p(r)}{\tau(r)} dr
	\]
	Setting
	\begin{equation} \label{v --> LW}
		v(r) := \lambda \sqrt{\tfrac{n-1}{n}} |LW| \frac{\tau^\p(r)}{\tau(r)},
	\end{equation}
	it is clear by Proposition \ref{proposition 1-form W calculations}(c) that $v(r) \in RC^0(\Sp^n)$. Therefore, calculating similarly to \eqref{|LW| in sphere} gives us that
	$$
	|LW| = \sqrt{|\Fcal_v(r)|} = \sqrt{\tfrac{n-1}{n}} \bigg| \frac{1}{\sin^n(r)} \int_0^r v \sin^n \, ds \bigg|.
	$$
	Now, getting back to \eqref{v --> LW}, the equation becomes
	$$
	|LW|= \tfrac{n-1}{n} \bigg|  \frac{\lambda}{\sin^n(r)} \int_0^r |LW| \frac{\tau^\p}{\tau} \sin^n \, ds \bigg|,
	$$
	and so,
	$$
	\sin^n(r) |LW| = \lambda \tfrac{n-1}{n} \bigg|\int_0^r (\sin^n|LW|) \frac{\tau^\p}{\tau} \, ds \bigg|, \qquad \forall r \in [0,\pi].
	$$
	Since $|LW| \ne 0$, without loss of generality, we may assume that there exists $a \in (0, \pi)$ such that
	\begin{subequations} \label{limit equation sphere}
		\begin{align}
			\sin^n(r) |LW| & = \lambda\tfrac{n-1}{n} \int_0^r (\sin^n|LW|) \frac{\tau^\p}{\tau} \, ds \quad \text{in $[0,a]$} \label{limit sphere a} \\
			\sin^n(r) |LW| & \not\equiv 0 \quad \text{in $[0,a]$}. \label{limit sphere b}
		\end{align}
	\end{subequations}
	Differentiating \eqref{limit sphere a} gives us
	$$
	\big(\sin^n|LW| \big)^\p - \lambda\tfrac{n-1}{n}\frac{ \tau^\p}{\tau} \big(\sin^n |LW| \big) = 0, \quad \forall r \in [0,a].
	$$
	However, since $\sin^n(0)|LW|(0) = 0$, we deduce from this linear equation that $\sin^n(r)|LW| \equiv 0$ in $[0,a]$, which contradicts \eqref{limit sphere b}. The proof is completed. 
\end{proof}

Although the limit equation criterion cannot be applied for the conformal equations on the sphere due to the presence of CKVF, this proposition provides insight that the equations \eqref{CE with null sigma in sphere} seems to be solvable in some sense. In the next two results, we will give existence of solutions to \eqref{CE with null sigma in sphere} in the far--from--CMC regime by generalizing the results of Theorem \ref{theorem near CMC sphere}, in which the conditions requiring mean curvature to be near constant and time-derivative $\rho$ to be small may be dropped; instead, all we need to assume is that $(\tau_0 + \alpha \tau_1)$ does not change sign or that $\rho$ is $C^1-$regular.  More precisely, our results are as follows.
\begin{theorem}[Solution with mean curvature of constant sign] \label{theorem limit equation criterion sphere}
	Let $\rho(r)\in RC^0(\Sp^n)$ and $\tau_0(r), \tau_1(r)\in RC^1(\Sp^n)$ be radial functions with $\rho \not \equiv 0$ and  $|\tau_0^\p| \le c\tau_1^\p$ for some constant $c>0$. Assume that
	\begin{equation} \label{condition on tau for limit equation criterion}
		\inf_{(a,r) \in [-c, c] \times [0, \pi]} \big\{ (\tau_0 - a\tau_1)^2 \big\} > 0.
	\end{equation}
	
	Then, there exists a constant $\alpha \in [-c,c]$ such that the conformal equations \eqref{CE with null sigma in sphere} associated with $(\rho, \tau = \tau_0 - \alpha \tau_1)$ admit at least one solution $\varphi(r) \in RC_+^3(\Sp^n)$. Moreover, in this case, $\alpha$ can be expressed by
	$$
	\alpha = \frac{\int_0^{\pi} \tau_0^\p \varphi^N \sin^n \, dr}{\int_0^{\pi} \tau_1^\p \varphi^N \sin^n \, dr}.
	$$
\end{theorem}
\begin{proof}
	Let $T$ be defined as in \eqref{definition of T}. We set
	$$
	K = \big\{ \phi \in C^0 ~~|~~ \exists t \in [0,1] ~~ \text{such that} ~~  \phi = tT(\phi) \big\}.
	$$
	If $K$ is bounded in $C^0$, it follows from Leray--Schauder's fixed point theorem that $T$ has a fixed point $\phi$, and hence the theorem follows. We will now show that the situation where $K$ is unbounded cannot happen. We argue by contradiction. Assume that $K$ is unbounded. Then, by the definitions of $T$ and $K$, there exists a sequence $(\varphi_i, t_i)$ such that
	\begin{multline}\label{varphi_i unbounded}
		- \tfrac{4(n-1)}{n-2} \varphi_i^{N + 1} (\varphi_i^\pp + h \varphi_i^\p ) + n(n-1) \varphi_i^{N + 2} \\
		+ \tfrac{n-1}{n} \big(\tau_0 - \alpha(\varphi_i) \tau_1\big)^2 \varphi_i^{2N} = \tfrac{n-1}{n} \big(t_i^{N} T_1(\varphi_i)\big)^2 + \rho^2,
	\end{multline}
	with $\alpha(\varphi_i)$ defined as in \eqref{def of alpha} and $\|\varphi_i\|_{C^0} \to +\infty$ as $i \to +\infty$. Here, we note that in view of the Lichnerowicz equation \eqref{T1 to T}, we have taken $\alpha_\star(\varphi_i) = 1$ by the same argument as in  \eqref{alpha_star equals 1}. Now, similarly to the process in \cite{DahlGicquaudHumbert}, letting $\gamma_i = \|\varphi_i\|_{C^0}$, we rescale the functions $\varphi_i$ and $\rho$ as
	$$
	\varphit_i = \frac{\varphi_i}{\gamma_i} \quad \text{and} \quad \rhot_i = \frac{\rho}{\gamma_i^N}.
	$$
	Then the Lichnerowicz-type \eqref{varphi_i unbounded} can be rewritten as
	\begin{multline*}
		\gamma_i^{2-N}\Big(-\tfrac{4(n-1)}{n-2} \Delta \varphit_i + n(n-1)\varphit_i \bigg) + \tfrac{n-1}{n} \big(\tau_0 - \alpha(\varphi_i) \tau_1\big)^2 \varphit_i^{N-1} \\
		= \frac{n-1}{n\varphit_i^{N+1}} \big(t_i^N T_1(\varphit_i) \big)^2 + \frac{\rhot^2}{\varphit_i^{N+1}}.
	\end{multline*}
	Since $\{\varphit_i\}$ is bounded in $C^0(\Sp^n)$, we get by Proposition \ref{proposition T1 compact} (up to a subsequence) 
	\begin{equation} \label{varphii to y 1}
		T_1(\varphit_i) \longrightarrow y(r) \quad  \text{in $C^0$}.
	\end{equation}
	Now, since 
	$$
	(\tau_0 - \alpha(\varphi_i) \tau_1)^2 \ge \inf_{(a,r) \in [-c, c] \times [0, \pi]} \big\{ (\tau_0 - a\tau_1)^2 \big\} > 0.
	$$ 
	an analysis similar to that in \cite[Lemma 2.6]{DahlGicquaudHumbert} or \cite[Theorem 3.3]{NguyenFPT} shows that (up to a subsequence)
	\begin{equation}\label{varphii to y 2}
		\varphit_i^N \longrightarrow \frac{t_0 |y(r)|}{(\tau_0 - \alpha_0 \tau_1)} \quad \text{in $C^0$},
	\end{equation}
	where (also up to a subsequence) $(t_0, \alpha_0) \in [0,1]\times[-c,c]$ is the limit of $\{(t_i,\alpha(\varphi_i)) \}$. Note that since $\|\varphit_i\|_{C^0} = 1$, we have $t_0 \ne 0$ and $y \not \equiv 0$. Combining \eqref{varphii to y 1} and \eqref{varphii to y 2}, it follows from the continuity of $T_1$ that
	$$
	|y(r)| = \bigg| \frac{t_0^N}{\sin^n(r)} \int_0^r \frac{\big(\tau_0 - \alpha_0 \tau_1 \big)^\p}{\big(\tau_0 - \alpha_0 \tau_1 \big)} |y(s)| \sin^n(s) \, ds \bigg|.
	$$
	Now, since $y \not \equiv 0$, without loss of generality, we may assume that there exists $a \in (0, \pi)$ such that
	\begin{align*}
		y(r & = \frac{t_0^N}{\sin^n(r)} \int_0^r \frac{\big(\tau_0 - \alpha_0 \tau_1 \big)^\p}{\big(\tau_0 - \alpha_0 \tau_1 \big)} |y(s)| \sin^n(s) \, ds \quad \text{in $[0,a]$}, \\
		y(r) & \not\equiv 0 \quad \text{in $[0,a]$}.
	\end{align*}
	However, an analysis similar to the nonexistence of solution to the equation \eqref{limit equation sphere} in the proof of Proposition \ref{proposition limit equation insight} shows that this is a contradiction. The proof is completed.
\end{proof}
\begin{theorem}[$C^1$-regular time-derivative solutions] \label{theorem smooth time-derivative sphere}
	Let $(\rho(r), \tau_0(r), \tau_1(r))$ be radial functions on $\Sp^n$ with 
	$\rho \not \equiv 0$, $(\tau_0(r), \tau_1(r)) \in \big(RC^1(\Sp^n)\big)^2$ and  $|\tau_0^\p| \le c\tau_1^\p$ for some constant $c>0$. Assume that $\rho(r) \in RC^1(\Sp^n)$.
	
	Then, there exists a constant $\alpha \in [-c,c]$ such that the conformal equations \eqref{CE with null sigma in sphere} associated with $(\rho, \tau = \tau_0 - \alpha \tau_1)$ admit at least one solution $\varphi(r) \in RC_+^3(\Sp^n)$. Moreover, in this case, $\alpha$ can be expressed by
	$$
	\alpha = \frac{\int_0^{\pi} \tau_0^\p \varphi^N \sin^n \, dr}{\int_0^{\pi} \tau_1^\p \varphi^N \sin^n \, dr}.
	$$
\end{theorem}
\begin{proof}
	We first modify the definition of $T$ in \eqref{definition of T} to adapt it to our case. Let $T_1$ be defined by \eqref{def T1}. We define $\Tt:[0,1] \times RC^0(\Sp^n) \to RC^0(\Sp^n)$ as follows. Given $(t, \phi(r)) \in [0,1] \times RC^0(\Sp^n)$, there exists a unique $\varphi(r) \in RC^2(\Sp^n)$ such that
	\begin{multline} \label{modify T step}
		- \tfrac{4(n-1)}{n-2} \varphi^{N + 1} (\varphi^\pp + h \varphi^\p )+ n(n-1)\varphi^{N + 2} \\
		+ \tfrac{n-1}{n} t^{2N}\alpha_\star(\phi) \big(\tau_0 - \alpha(\phi) \tau_1 \big)^2 \varphi^{2N} = \tfrac{n-1}{n} \big(T_1(\phi)\big)^2 + \rho^2,
	\end{multline}
	where $\alpha(.)$ and $\alpha_\star(.)$ are defined as in \eqref{def of alpha} and \eqref{def alpha_star} respectively. Then, we set
	$$
	\Tt(t, \phi)(r) := t\varphi(r).
	$$
	Of course, similarly to Proposition \ref{proposition T1 compact}, $\Tt$ is continuous and compact. Moreover, as explained in \eqref{alpha_star = 1}, a fixed point of $\Tt(1,.)$ is a solution to \eqref{CE with null sigma in sphere}. 
	
	\medskip
	
	We now set
	$$
	K = \big\{ (t,\phi) \in [0,1] \times C^0 \ |\ \phi(r)\in RC^0(\Sp^n) \mbox{ and } \phi = \Tt(t,\phi) \big\}.
	$$ 
	By Leray--Schauder's fixed point theorem, to show that $\Tt(1,.)$ has a fixed point, it suffices to show that $K$ is bounded. In order to prove so, given a couple $(t,\phi) \in K$, we divide into two cases depending on whether  $\max\big(\tau_1 |\phi|^N \big)$ is smaller or larger than $\max\big(\tau_1 \varphi_\star^N\big)$, where $\varphi_\star$ is recalled to be the unique solution to \eqref{def varphi_star}.
	
	\medskip
	
	\textbf{Case 1.} If $(t,\phi) \in K$ satisfies 
	$$
	\max\big(\tau_1 |\phi|^N \big) < \max\big(\tau_1 \varphi_\star^N\big),
	$$ 
	then it follows from \eqref{es T1} that 
	\begin{equation} \label{case1 K bounded sphere}
		\max|T_1(\phi)| < c\pi \max\big(\tau_1^\p \varphi_\star^N \big).
	\end{equation}
	Therefore, letting $\varphi := \Tt(t,\phi)$ and letting $\varphi_\infty$ be the unique solution to the equation
	$$
	-\tfrac{4(n-1)}{n-2} \varphi_\infty^{N + 1} (\varphi_\infty^\pp + h \varphi_\infty^\p ) + n(n-1) \varphi_\infty^{N + 2} 
	= \tfrac{n-1}{n} \bigg(c\pi \max\big(\tau_1^\p \varphi_\star^N \big)\bigg)^2 + \rho^2,
	$$
	since
	\begin{multline}\label{varphi_i unbounded 2}
		- \tfrac{4(n-1)}{n-2} \varphi^{N + 1} (\varphi^\pp + h \varphi^\p ) + n(n-1) \varphi^{N + 2} \\
		+ \tfrac{n-1}{n}t^{2N} \alpha_\star(\phi) \big(\tau_0 - \alpha(\varphi) \tau_1\big)^2 \varphi^{2N} = \tfrac{n-1}{n} \big(T_1(\phi)\big)^2 + \rho^2,
	\end{multline}
	it is clear that $\varphi_\infty$ is a supersolution to \eqref{varphi_i unbounded 2}, and so 
	$$
	\phi = t\varphi \le \varphi \le \varphi_\infty.
	$$
	
	\textbf{Case 2.}  If $(t,\phi) \in K$ satisfies 
	$$
	\max\big(\tau_1 |\phi|^N \big) \ge \max\big(\tau_1 \varphi_\star^N\big),
	$$ 
	then, by definition $\alpha_\star(\phi) = 1$. Therefore, letting $\varphi := \Tt(t,\phi)$, the fact that $\phi = t \Tt(t,\phi)$ gives
	\begin{multline}\label{eqs varphi and t}
		- \tfrac{4(n-1)}{n-2} \varphi^{N + 1} (\varphi^\pp + h \varphi^\p ) + n(n-1) \varphi^{N + 2} \\
		+ \tfrac{n-1}{n} t^{2N}\big(\tau_0 - \alpha(\varphi) \tau_1\big)^2 \varphi^{2N} = \tfrac{n-1}{n} \big(t^{N} T_1(\varphi)\big)^2 + \rho^2.
	\end{multline}
	Observing that this equation is indeed the Lichnerowicz-type equation  \eqref{rce sphere} associated with 
	$$ V \equiv 0, \quad \psi \equiv 0 \quad \text{and} \quad \tau = t_i^N (\tau_0 - \alpha(\varphi_i)\tau_1
	.
	$$ 
	Then, since $\varphi$ is a solution to \eqref{eqs varphi and t}, and according to Theorem \ref{theorem existence on sphere}, we have $S_{V,\psi,\rho,\varphi} \ge 0$, and since  $V \equiv 0$ and $\psi \equiv 0$ we get
	\begin{equation}\label{ineq S sphere diff rho 1}
		N^2 (\varphi^\p)^2 - n^2 \varphi^2 + 2 n N \cot \varphi \varphi^\p + \frac{n}{(n-1)\sin^n} \int_0^r \frac{\rho^2}{\varphi^{2N}} (\varphi^N \sin^n)^\p  \, ds \ge 0.
	\end{equation}
	Now, let $r_0 \in [0,\pi]$ be a maximum point of $\varphi$. We will consider three situations where $r_0$ is either $0$ or $\pi$ or in $(0,\pi)$.
	
	\medskip
	
	\textit{-- Sub-case 2.1.} When $r_0 = 0$, we have by l'Hôpital rule
	\begin{align*}
		\lim_{r \to 0} \bigg( \frac{1}{\sin^n(r)} \int_0^r \frac{\rho^2}{\varphi^{2N}} (\varphi^N \sin^n)^\p  \, ds \bigg) 
		& = \lim_{r \to 0} \bigg( \frac{\varphi^N(0)}{\varphi^N(r)\sin^n(r)} \int_0^r \frac{\rho^2}{\varphi^{2N}} (\varphi^N \sin^n)^\p  \, ds \bigg) \\
		& = \frac{\rho^2(0)}{\varphi^N(0)}.
	\end{align*}
	Assuming that $n^2 \varphi^2(0) > \frac{n}{n-1} \frac{\rho^2(0)}{\varphi^N(0)}$, and since $\varphi^\p(0) = 0$, it follows from \eqref{ineq S sphere diff rho 1} that for all $r$  near $0$
	$$
	2 n N \cot \varphi \varphi^\p \ge n^2 \varphi^2 -  N^2 (\varphi^\p)^2  -  \frac{n}{(n-1)\sin^n} \int_0^r \frac{\rho^2}{\varphi^{2N}} (\varphi^N \sin^n)^\p  \, ds > 0. 
	$$
	This means that $\varphi^\p > 0$ near $0$, which contradicts the fact that $r_0 = 0$ is a maximum point of $\varphi$. Therefore, we must have
	$$
	n^2 \varphi^2(0) \le \tfrac{n}{n-1} \frac{\rho^2(0)}{\varphi^N(0)},
	$$ 
	and so
	$$
	\phi \le \|\varphi\|_{C^0} = \varphi(0) \le \bigg(\frac{\rho^2(0)}{n(n-1)}\bigg)^{1/(N+2)}.
	$$
	
	\textit{-- Sub-case 2.2.} Similarly, when $r_0 = \pi$, since, by Remark \ref{remark vanishing condition on Is},
	$$
	\int_0^\pi \frac{\rho^2}{\varphi^{2N}} (\varphi^N \sin^n(s))^\p  \, ds = 0,
	$$
	applying the l'Hôpital rule, we have
	\begin{align*}
		\lim_{r \to \pi} \bigg( \frac{1}{\sin^n} \int_0^r \frac{\rho^2}{\varphi^{2N}} (\varphi^N \sin^n)^\p  \, ds \bigg) 
		& = \varphi^N(\pi)\lim_{r \to \pi} \bigg( \frac{1}{\varphi^N\sin^n} \int_0^r \frac{\rho^2}{\varphi^{2N}} (\varphi^N \sin^n)^\p  \, ds \bigg) \\
		& = \frac{\rho^2(\pi)}{\varphi^N(\pi)}.
	\end{align*}
	By an analysis similar to Sub-case 1, assuming that $n^2 \varphi^2(\pi) > \tfrac{n}{n-1} \frac{\rho^2(\pi)}{\varphi^N(\pi)}$, we get $\cot(r)\varphi \varphi^\p > 0$ near $\pi$, and so, $\varphi^\p < 0$ near $\pi$, which contradicts the fact that $\pi$ is a maximum point of $\varphi$. Hence, we must have 
	$$
	\phi \le \|\varphi\|_{C^0} = \varphi(\pi) \le \bigg(\frac{\rho^2(\pi)}{n(n-1)}\bigg)^{1/(N+2)}.
	$$
	
	\textit{-- Sub-case 2.3.} When $\varphi$ gets maximum at $r_0 \in (0,\pi)$, since $\varphi^\p(r_0) = 0$, we obtain from \eqref{ineq S sphere diff rho 1} that
	$$
	n^2 \varphi^2(r_0) \le \frac{n}{(n-1)\sin^n(r_0)} \int_0^{r_0} \frac{\rho^2}{\varphi^{2N}} (\varphi^N \sin^n)^\p  \, ds.
	$$
	Observing that 
	\begin{equation}\label{change integration by part}
		\frac{\rho^2}{\varphi^{2N}} (\varphi^N \sin^n)^\p = - \big(\rho^2 \sin^{2n} \big) \bigg(\frac{1}{\varphi^N \sin^n}\bigg)^\p,
	\end{equation}
	by integration by parts, the inequality becomes
	\begin{equation}\label{ineq 2 sphere}
		n^2 \varphi^2(r_0) \le  \frac{n}{(n-1)\sin^n(r_0)} \int_0^{r_0} \frac{\big(\rho^2 \sin^{2n}(r)\big)^\p}{\varphi^N\sin^n(s)}  \, ds - \frac{n\rho^2(r_0)}{(n-1)\varphi^N(r_0)}.
	\end{equation}
	Moreover, also thanks to \eqref{change integration by part} and integration by parts, the second vanishing condition (see Remark \ref{remark vanishing condition on Is})
	$$
	\int_0^{\pi} \frac{\rho^2}{\varphi^{2N}} (\varphi^N \sin^n(s))^\p  \, ds = 0
	$$
	gives us that
	\begin{equation}\label{second vanishing condition}
		\int_0^{\pi} \frac{\big(\rho^2 \sin^{2n}(s)\big)^\p}{\varphi^N\sin^n(s)}  \, ds = 0
	\end{equation}
	We next argue similarly to \eqref{T1 is bounded}, that is, by the mean value theorem
	\begin{itemize}
		\item if $r_0 \in (0,\pi/2]$, then there exists $r_1 \in [0,r_0]$ satisfying
		\begin{align}\label{r1}
			\int_0^{r_0} \frac{\big(\rho^2 \sin^{2n}\big)^\p}{\varphi^N\sin^n}  \, ds & = r_0 \frac{\big(\rho^2 \sin^{2n}\big)^\p(r_1)}{\varphi^N(r_1)\sin^n(r_1)} \notag \\
			& =  r_0 \frac{2\big(n\rho^2(r_1) \sin^{n-1}(r_1)\cos(r_1) + \rho(r_1)\rho^\p(r_1) \sin^{n}(r_1)\big)}{\varphi^N(r_1)};
		\end{align}
		
		\item if $r_0 \in (\pi/2, \pi)$, then in view of \eqref{second vanishing condition}, there exists $r_2 \in [r_0, \pi]$ such that
		\begin{align}\label{r2}
			\int_0^{r_0} \frac{\big(\rho^2 \sin^{2n}\big)^\p}{\varphi^N\sin^n}  \, ds & = -  \int_{r_0}^\pi \frac{\big(\rho^2 \sin^{2n}\big)^\p}{\varphi^N\sin^n}  \, ds \notag \\
			& = (\pi - r_0 ) \frac{\big(\rho^2 \sin^{2n}\big)^\p(r_2)}{\varphi^N(r_1)\sin^n(r_2)} \notag \\
			& = (\pi - r_0 )\frac{2\big(n\rho^2(r_2) \sin^{n-1}(r_2)\cos(r_2) + \rho(r_2)\rho^\p(r_2) \sin^{n}(r_2)\big)}{\varphi^N(r_2)}.
		\end{align}
	\end{itemize}
	Therefore, since
	$$
	\begin{cases}
		\max \Big(\sin(\rt), \, \frac{2r}{\pi} \Big) \le \sin(r)  &\text{for all} \quad 0 \le \rt \le r \le {\pi/2},
		\\
		\max \Big(\sin(\rt), \, \frac{2(\pi-r)}{\pi} \Big) \le \sin(r)  &\text{for all} \quad \pi/2 \le r \le \rt \le \pi,
	\end{cases}
	$$
	we get from (\ref{r1}--\ref{r2}) that 
	$$
	\bigg| \frac{1}{\sin^n(r_0)}  \int_0^{r_0} \frac{\big(\rho^2 \sin^{2n}\big)^\p}{\varphi^N\sin^n}  \, ds \bigg| 
	\le \pi \frac{\big(n\rho^2(\rt) \cos(\rt) + \rho(\rt)\rho^\p(\rt) \sin(\rt)\big)}{\varphi^N(\rt)},
	$$
	where the value of $\rt$ is assigned as $r_1$ or $r_2$, contingent upon whether $r_0$ is smaller or greater than $\pi/2$, respectively.
	
	\medskip
	
	Since $\varphi \ge \varphi_\star$, with $\varphi_\star$ defined as in \eqref{def varphi_star}, it follows that 
	$$
	\bigg| \frac{1}{\sin^n(r_0)}  \int_0^{r_0} \frac{\big(\rho^2 \sin^{2n}(r)\big)^\p}{\varphi^N\sin^n(s)}  \, ds \bigg| 
	\le (n+1)\pi \frac{\|\rho\|_{C^1}^2}{(\min\varphi_\star)^N}.
	$$
	Combined with \eqref{ineq 2 sphere}, we obtain
	$$
	n^2 \varphi^2(r_0) \le  \tfrac{n(n+1)}{(n-1)} \frac{\pi\|\rho\|_{C^1}^2}{(\min\varphi_\star)^N},
	$$
	and so
	$$
	\phi \le \|\varphi\|_{C^0} = \varphi(r_0) \le \sqrt{\frac{(n+1)\pi}{n(n-1)}}  \frac{\|\rho\|_{C^1}}{(\min\varphi_\star)^{N/2}}.
	$$
	In three sub-cases, we have shown that all $(t,\phi) \in K$ in Case 2 are uniformly bounded, therefore, we conclude that $K$ is bounded, which completes the proof. 
\end{proof}

\subsection{Stability and Instability of conformal equations} \label{subsection stability and instability on the sphere}
Another application of Theorem \ref{theorem existence on sphere} that we would like to present before ending this section is the investigation of the stability and instability of the conformal equations \eqref{CE} in the non-CMC case. This problem has attracted interest from many authors and carries physical significance in the context of general relativity. From the perspective of elliptic analysis, the (in)stability behavior of the equations \eqref{CE} heavily depends on the sign of $\mathcal{B}_{\tau,\psi}$. To the best of the author’s knowledge, a complete description of the stability of \eqref{CE} on compact manifolds has been obtained only in the case where $\mathcal{B}_{\tau,\psi} > 0$, through the work of E. Hebey, O. Druet, and B. Premoselli in \cite{DruetHebey, DruetPremoselli, Premoselli_Stability}. In contrast, when $\mathcal{B}_{\tau,\psi} < 0$
or when$\mathcal{B}_{\tau,\psi}$ changes sign, the conformal equations are known to be stable only in the CMC and/or near--CMC regimes. When $\tau$ is far-from-CMC, the questions of stability and instability of \eqref{CE} remain largely open.

\medskip

In this subsection, we would like to give answers, or at least compelling evidence, to these questions by considering the conformal equations on the sphere within the class of radial solutions. For the case where $\mathcal{B}_{\tau,\psi} \le 0$, we are particularly interested in the equations \eqref{CE with null sigma in sphere} which are closely related to the vacuum conformal equations \eqref{CE} as studied in \cite{Maxwellcompact, HolstNagyTsogtgerel, DahlGicquaudHumbert, NguyenFPT, NguyenNonexistence}.  As the reader may have remarked in the proof of Theorem \ref{theorem smooth time-derivative sphere}, the result we obtained in this case establishes the stability of solutions to \eqref{CE with null sigma in sphere}. When $\mathcal{B}_{\tau,\psi}$ changes sign, we will show that the conformal equations \eqref{CE} are not stable in general by constructing an example of a sequence of blowing-up solutions to \eqref{CE} with non-CMC. More precisely, the statements of our results are as follows.
\begin{theorem}[Non-positive $\mathcal{B}_{\tau, \psi}$ and stability] \label{theorem stability on sphere}
	Let $\rho(r) \in RC^1(\Sp^n)$ be a radial function. For any $\tau(r) \in RC^1(\Sp^n)$ we denote by $\mathcal{S}_{\Sp^n}(\tau)$ the set of all solutions $\varphi(r) \in RC^3_+(\Sp^n)$ to the conformal equations \eqref{CE with null sigma in sphere} associated with $(\rho, \tau)$. Then, there exists a constant $C>0$ depending only on $\rho$ such that
	$$
	\sup\bigg\{\|\varphi\|_{C^0(\Sp^n)} ~~ \bigg|~~ \text{$\exists \tau \in RC^1(\Sp^n)$ such that $\varphi \in \mathcal{S}_{\Sp^n}(\tau)$}  \bigg\} \le C.	
	$$
\end{theorem}
\begin{proof}
	The proof is exactly similar to arguments in Case 2 in the proof of Theorem \ref{theorem smooth time-derivative sphere} where we prove that solutions of the equations \eqref{eqs varphi and t} are uniformly bounded by a function of $\|\rho\|_{C^1(\Sp^n)}$.
\end{proof}

\begin{theorem}[Instability under sign change of $\mathcal{B}_{\tau, \psi}$]\label{theorem instability sphere}
	Given any $\Lambda > 0$, there exist sequences $(\tau_i(r))_i$ of radial $C^1$ functions, $(\varphi_i(r))_i$ of radial $C^2$ functions and $(W_i)_i$ of 1-forms such that
	\begin{itemize}
		\item for each $i$, $(\varphi_i(r),W_i)$ is a solution to the conformal constraint equations with data set $V=\Lambda$, $\tau = \tau_i$, $\psi\equiv0$ and $\rho\equiv0$;
		\item the sequence $(\tau_i(r))$ converges in $C^1(\Sp^n)$;
		\item $\|\varphi_i\|_{L^\infty} \to + \infty$.
	\end{itemize}
\end{theorem}
\begin{proof}
Given two positive constants $a$ and $b$, we define
	$$
	\varphi_\eps := \frac{b}{\sin^a(r) + \eps }.
	$$
	Of course, $(\varphi^N_\eps (\sin^n(r)))^\p$ has finite zero-points in $[0,\pi]$ and 
	$$
	\lim\limits_{\eps\to0} \varphi_\eps (0) = \lim\limits_{\eps\to0} \varphi_\eps (\pi) = +\infty.
	$$ 
	Therefore, the theorem will follow if we can find a sequence of mean curvatures $\{\tau_\eps \}$ such that $\tau_\eps\to\tau_0$ in $C^1$ and  $\varphi_\eps$ satisfies the Lichnerowicz-type equation \eqref{rce sphere} with $V=\Lambda$, $\tau = \tau_\eps$, $\psi\equiv0$ and $\rho\equiv0$. 
	
	\medskip
	
	To do so, we first calculate
	\begin{align*}
		\varphi_\eps^\p & = - \frac{\sin^{a - 1} (r) \cos (r)}{\big( \sin^a(r) + \eps \big)^2 } ab \\
		\varphi_\eps^\pp & = \frac{\sin^a (r)}{\big( \sin^a(r) + \eps \big)^2 } ab - \frac{\sin^{a - 2} (r) \cos^2 (r)}{\big( \sin^a(r) + \eps \big)^2 } a (a - 1) b + \frac{2 \sin^{2(a - 1)} (r) \cos^2 (r)}{\big( \sin^a(r) + \eps \big)^3 } a^2 b.
	\end{align*}	
	Computing $S_{\Lambda, 0, 0, \varphi_\eps}(r)$ given in \eqref{S on sphere},  we get that $S_{\Lambda, 0, 0, \varphi_\eps}(r) > 1$ if and only if
	\begin{multline} \label{eqn-S_phi_eps_positive}
		\tfrac{2n\Lambda}{n-1}b^{N-2} > (\sin^a(r)+\eps)^{N-2-\frac2a}
		\bigg((\sin^a(r) + \eps)^{{\frac2a}} -\frac{N^2a^2(\sin^a(r))^{2-\frac2a}\cos^2(r)}{(\sin^a(r)+\eps)^{2-\frac2a}} \\
		+ \frac{2nNa(\sin^a(r))^{1-\frac2a}\cos^2(r)}{(\sin^a(r)+\eps)^{1-\frac2a}} +\frac{n^2}{(\sin^a(r)+\eps)^{-\frac2a}}\bigg).
	\end{multline}
	Since $N-2=\frac{4}{n-2}>0$, it is possible to chose $a$ such that $N-2-\frac2a>0$ and $1-\frac2a>0$. With this choice, the right-hand-side of \eqref{eqn-S_phi_eps_positive} is bounded above by $2^{N-2-\frac2a}\big(2nNa+(n^2 + 1)2^{\frac2a}\big)$ for any $r\in\ff{0,\pi}$ and any $\eps<1$. Therefore, choosing $b$ large enough, Inequality \eqref{eqn-S_phi_eps_positive} is satisfied and $S_{\Lambda, 0, 0, \varphi_\eps}(r)>1$ on $\ff{0,\pi}$.
	
	From Theorem \ref{theorem existence on sphere}--(iii), the function $\varphi_\eps$ is a solution to the Lichnerowicz-type equation with data $\tau_\eps$ given by
	\[
	\tau_\eps = \frac{2N\varphi_\eps^\pp\varphi_\eps + N^2(\varphi_\eps^\p)^2 + 2N(2n-1)\cot\varphi_\eps^\p\varphi_\eps - 2n^2\varphi_\eps^2 + \tfrac{4n\Lambda}{n-1}\varphi_\eps^N}%
		{2\sqrt{\varphi_\eps^N S_{\Lambda, 0, 0, \varphi_\eps}}}
	\]
	
	From the expression of $\varphi_\eps$ and its derivatives we have the following upper bounds (where $C_i$ are constants):
	\begin{align*}
		|\varphi_\eps(r)| & = b\big(\sin^a(r)+\eps\big)^{-1} \\
		|\varphi_\eps^\p(r)| & \le C_1\big(\sin^a(r)+\eps\big)^{-1-\frac1a} \\
		|\cot(r)\varphi_\eps^\p(r)| & \le C_2\big(\sin^a(r)+\eps\big)^{-1-\frac2a} \\
		|\varphi_\eps^\pp(r)| & \le C_3\big(\sin^a(r)+\eps\big)^{-1} + C_4\big(\sin^a(r)+\eps\big)^{-1-\frac2a}
	\end{align*}
	from which we get
	\begin{align*}
		|\varphi_\eps(r)^{2-N}| & = b^{2-N}\big(\sin^a(r)+\eps\big)^{N-2} \\
		|\varphi_\eps^\p(r)^2\varphi_\eps(r)^{-N}| & \le C_5\big(\sin^a(r)+\eps\big)^{N-2-\frac2a} \\
		|\cot(r)\varphi_\eps^\p(r)\varphi_\eps(r)^{1-N}| & \le C_6\big(\sin^a(r)+\eps\big)^{N-2-\frac2a} \\
		|\varphi_\eps^\pp(r)\varphi_\eps(r)^{1-N}| & \le C_7\big(\sin^a(r)+\eps\big)^{N-2} + C_8\big(\sin^a(r)+\eps\big)^{N-2-\frac2a}
	\end{align*}
	Observing that
	\[
	\tau_\eps = \frac{2N\varphi_\eps^\pp\varphi_\eps^{1-N} + N^2(\varphi_\eps^\p)^2\varphi_\eps^{-N} + 2N(2n-1)\cot\varphi_\eps^\p\varphi_\eps^{1-N} - 2n^2\varphi_\eps^{2-N} + \tfrac{4n\Lambda}{n-1}}%
	{2\sqrt{N^2(\varphi_\eps^\p)^2\varphi_\eps^{-N} - n^2\varphi_\eps^{2-N} + 2nN\cot\varphi_\eps^\p\varphi_\eps^{1-N} + \tfrac{2n\Lambda}{n-1}}}
	\]
	and since $N-2-\frac2a>0$, the upper bounds above give
	\[
	\lim\limits_{(r,\eps)\to(0,0)}\tau_\eps(r) = \lim\limits_{(r,\eps)\to(\pi,0)}\tau_\eps(r) = \sqrt{\frac{2n\Lambda}{n-1}}.
	\]
	Noting $\varphi_0=\frac{b}{\sin^a}$ and $\tau_0$ the function defined by
	\[
	\left\{\begin{array}{l}
			\tau_0 = \frac{2N\varphi_0^\pp\varphi_0 + N^2(\varphi_0^\p)^2 + 2N(2n-1)\cot\varphi_0^\p\varphi_0 - 2n^2\varphi_0^2 + 	\tfrac{4n\Lambda}{n-1}\varphi_0^N}{2\sqrt{\varphi_0^N S_{\Lambda, 0, 0, \varphi_0}}} \qquad\mbox{ on }\oo{0,\pi} \\
			\tau_0(0)=\tau_0(\pi)=\sqrt{\frac{2n\Lambda}{n-1}}
		\end{array}
	\right.
	\]
	we have that the function $(r,\eps)\mapsto\tau_\eps(r)$ is uniformly continuous on $\ff{0,\pi}\times\ff{0,1}$ and therefore $\tau_\eps\longrightarrow\tau_0$ in $C^0$.
	
	Taking the derivatives of $\tau_\eps$ and $\tau_0$ on $\ff{0,\pi}$, and following the same approach we get (after direct but tricky computations) that $\tau_\eps^\p\longrightarrow\tau_0^\p$ in $C^0$ as soon as $N-2-\frac3a>0$. The proof is completed.  
\end{proof}

\begin{remark}
	For any $k\in\mathbb{N}$, it is possible to have the previous result with the functions $(\tau_i(r))$ converging in $C^k(\Sp^n)$. It is sufficient to follow the proof above with $a$ and $b$ great enough so that $N-2-\frac{k+1}{a}>0$ and $S_{\Lambda, 0, 0, \varphi_\eps} > 1$.
\end{remark}

%
%
\section{Solutions in the hyperbolic space} \label{section hyperbolic}
%
%
In this section, we will examine Theorem \ref{theorem general} in the context of hyperbolic space. Unlike the spherical case, the existence of solutions to the Lichnerowicz-type equation \eqref{RCE} in hyperbolic space is not constrained by the condition~\eqref{vanishing condition sphere}, arising from the presence of CKVFs. Thus, a key advantage of this setting is that, besides being able to use arguments similar to those of the previous section, we can also directly apply well-known techniques developed for the general conformal equations on AH manifolds having no CKVFs, without any additional modifications. In particular, one of these important techniques we employ is the limit equation criterion established by R. Gicquaud and A. Sakovich \cite{GicquaudSakovich}, which corresponds exactly to \eqref{limit equation} when adapted to AH models.

\medskip

Building on this approach, besides proving Main Theorem \ref{main theorem} in hyperbolic spaces, which provides a necessary and sufficient condition for a positive function $\varphi$ to solve \eqref{RCE}, an important outcome in this section is that, as long as $\tau$ does not change sign, the conformal equations in the radial setting on hyperbolic space are shown to be solvable, with a compact set of solutions. Although this result is limited to the radial setting, it provides strong evidence for the solvability of \eqref{CE} in general data, with the limit equation criterion offering a promising tool. This indicates that, in contrast to the conclusions obtained in \cite{NguyenNonexistence, NguyenProgress, DiltsHolstKozarevaMaxwell} for compact manifolds, the conformal method remains a powerful approach for parametrizing constraint solutions on AH spaces, not only in the CMC or near-CMC regimes, but also in the far-from-CMC regime.

\medskip

Another achievement we are concerned with in this section is the role of the decay rate of constraint solutions at infinity in the sign of mass in AH manifolds. Using the construction of solutions to the conformal equations provided by Main Theorem~\ref{main theorem}, we will show that if the decay rate of the initial data $(g,k)$ is critical, then the mass can take any sign. In particular, this result highlights that the commonly assumed decay conditions on $(g,k)$ in the Positive Mass Theorem for AH manifolds are sharp.

\medskip

The organization of the section is as follows. In Subsection \ref{subsection Lieq-type hyperbolic}, we recall some basic notation on hyperbolic space and then provide the proof of Main Theorem \ref{main theorem} in this setting. In Subsection \ref{subsection existence solutions AH}, provided that $\tau$ does not change sign, the solvability of the conformal equations in the radial data setting will be obtained. Finally, in the last subsection, we review the conjecture on the positivity of the AH mass and show how the decay rate of constraint solutions at infinity derives the sign of the mass.

\subsection{Lichnerowicz-type equation in the hyperbolic space}  \label{subsection Lieq-type hyperbolic}
Let us first denote by $\HH^n$ the $n$-dimensional hyperbolic space, with $n \ge 3$. We fix a point in
$\HH^n$ as an origin. Then, in geodesic normal coordinates at this point, the hyperbolic metric reads
$b = dr^2 + \sinh^2 (r) \sigma_s$, where $\sigma_s$ is the standard round metric on $\mathbb{S}^{n-1}$ and $r$ is the distance from the origin. The functions $h, \theta, u_0$ and the constants $\Ric$, $\Scal_g$ are given by
\begin{equation}\label{computing in hyperbolic}
	\begin{gathered}
		\theta(r) = \sinh^{n - 1} (r), \quad h(r) = (n - 1) \coth (r), \quad u_0(r) = \cosh (r) \\
		\Ric = - (n - 1), \quad \Scal = - n (n -1).
	\end{gathered}
\end{equation}
Similarly to Theorem \ref{theorem existence on sphere} for the sphere, we first give a necessary and sufficient condition for a radial positive function $\varphi(r)$ to be a solution to the Lichnerowicz-type equation \eqref{RCE} in the hyperbolic space.

\medskip

From a radial data set $(V,\tau(r),\psi(r),\rho(r))\in C^0(\RR) \times RC^1(\HH^n) \times RC^1(\HH^n) \times RC^0(\HH^n)$ and a radial function $\varphi(r) \in RC_+^3(\HH^n)$ we define
\[
I_{V,\psi,\rho,\varphi} := \tfrac{n}{n-1} \bigg( 2 V(\psi) \varphi^{2N} + |\psi^\p|^2 \varphi^{N+2} + \rho^2 \bigg),
\]
and
\[
S_{V,\psi,\rho,\varphi} := N^2 (\varphi^\p)^2 + n^2 \varphi^2 + 2 n N \coth \varphi \varphi^\p + \frac{1}{\sinh^n} \int_0^r \frac{I_{V,\psi,\rho,\varphi}}{\varphi^{2N}} (\varphi^N \sinh^n)^\p  \, ds.
\]
\begin{theorem}[Solutions with freely specified conformal factor] \label{theorem existence on hyperbolic}
	Consider a radial data set $(V,\tau(r),\psi(r),\rho(r))\in C^0(\RR) \times RC^1(\HH^n) \times RC^1(\HH^n) \times RC^0(\HH^n)$ with $\rho \psi^\p \equiv 0$. For $\varphi(r) \in RC_+^3(\HH^n)$ we have:
	\begin{enumerate}[(i)]
		\item $\varphi(r)$ is a solution to the conformal constraint equations \eqref{CE with null sigma} if and only if
		\begin{multline} \label{rce hyperbolic}
		- 2N \varphi^{N + 1} (\varphi^\pp + h \varphi^\p ) - n^2\varphi^{N + 2} + \tau^2  \varphi^{2N} \\
		= \bigg( \frac{1}{\sinh^n} \int_0^r \tau^\p \varphi^N  \sinh^n \, ds \bigg)^2 + I_{V,\psi,\rho,\varphi}.
		\end{multline}
		
		\item If $\varphi$ is a solution to \eqref{rce hyperbolic}, then $S_{V,\psi,\rho,\varphi} \ge 0$.
		
		\item If $S_{V,\psi,\rho,\varphi}>0$ on $\RR_+$ and $(\varphi^N \sinh^n)^\p \neq 0$ a.e. on $\RR_+$, then $\varphi$ is a solution to \eqref{rce hyperbolic} if and only if the function $\tau$ satisfies
		\begin{equation}\label{AHidenity}
			\tau= \pm \frac{ S_{V,\psi,\rho,\varphi} + 2N  \big( \varphi^\pp + (n - 1) \coth \varphi^\p \big) \varphi + n^2 \varphi^2 + I_{V,\psi,\rho,\varphi} \varphi^{-N}}{2\sqrt{\varphi^N S_{V,\psi,\rho,\varphi}}}.
		\end{equation}
	\end{enumerate}
	Moreover, in this case, the $1-$form $W$ can be computed by \eqref{solution to vector equations} with $v=\tfrac{n-1}{n}\varphi^N\tau^\p$. 
\end{theorem}
The proof of this theorem is broadly similar to that of Theorem \ref{theorem existence on sphere} line by line, with the only difference being that we replace $(\sik, \cok)=(\sin(r), \cos(r))$ on the sphere by 
$(\sinh(r),\cosh(r))$ in hyperbolic space, as defined in \eqref{def sik and cok}. Here, for the reader’s convenience, we detail the proof as follows.

\begin{proof}
	Consider $f_v$, $F_v$ and $\Fcal_v$ given by \eqref{def f}, \eqref{def F nonpositive Ric} and \eqref{def Fcal} respectively, where  
	$$
	v =\tfrac{n-1}{n}\varphi^N\tau^\p - \rho \psi^\p.
	$$
	Thanks to Theorem \ref{theorem general}, $\varphi$ is a solution to the conformal equations \eqref{CE with null sigma} if and only if it solves the Lichnerowicz-type equation \eqref{RCE}.
	
	\medskip
	
	In this regards, we recall that $h^\p(r)+\frac{h^2(r)}{n-1}+\Ric=0$.  Therefore, $\Fcal_v$ becomes
	\begin{equation} \label{Fcal hyperbolic}
		\Fcal_v =  \tfrac{4n}{n-1} \Big(  f_v - \tfrac{n}{n-1} h F^\p_v - n F_v \Big)^2.
	\end{equation}
	Observing that, by the definition \eqref{def F nonpositive Ric} and \eqref{computing in hyperbolic}
	\begin{align*}
		h F^\p_v & = h\frac{u_0^\p}{u_0}F_v +  \frac{h}{u_0 \theta } \int_0^r f_v u_0 \theta \, ds \\
		& =  (n - 1) \bigg( \frac{1}{n \sinh^n} \int_0^r f_v \big(\sinh^n\big)^\p \, ds - F_v \bigg).
	\end{align*}
	Integrating by part the integral term on the right-hand side, we obtain  
	$$
	h F^\p_v = (n - 1) \bigg( \frac{f_v}{n} - \frac{1}{2n\sinh^n} \int_0^r v \sinh^n \, ds - F_v \bigg).
	$$
	Taking into \eqref{Fcal hyperbolic}, it follows that
	\begin{equation} \label{|LW| in hyperbolic}
		\Fcal_v = \tfrac{n-1}{n} \bigg( \frac{1}{\sinh^n} \int_0^r \bigg(  \varphi^N \tau^\p - \tfrac{n}{n-1}\rho \psi^\p \bigg) \sinh^n \, ds \bigg)^2.
	\end{equation}
	Therefore, since $\rho \psi^\p$ is assumed to be zero everywhere, by definition, we can rewrite the Lichnerowicz-type equation as
	\begin{multline*}
		- 2N \varphi^{N + 1} (\varphi^\pp + h \varphi^\p ) - n^2\varphi^{N + 2} + \tau^2  \varphi^{2N} \\
		= \bigg( \frac{1}{\sinh^n} \int_0^r \tau^\p \varphi^N  \sinh^n \, ds \bigg)^2 + I_{V,\psi,\rho,\varphi},
	\end{multline*}
	thus proving point $(i)$. Continuing to integrate by parts the integral term on the right-hand side, we get
	\begin{multline*}
		- 2N \varphi^{N + 1} (\varphi^\pp + h \varphi^\p ) - n^2 \varphi^{N + 2} + \tau^2 \varphi^{2N} \\
		= \Big( \tau \varphi^N - \frac{1}{\sinh^n} \int_0^r \tau \big( \varphi^N \sinh^n \big)^\p \, ds \Big)^2 + I_{V,\psi,\rho,\varphi},
	\end{multline*}
	equivalently,
	\begin{multline}\label{eq.5a1 hyperbolic}
		- 2N \varphi^{N + 1} (\varphi^\pp + h \varphi^\p ) - n^2 \varphi^{N + 2} = \frac{1}{\sinh^{2n}} \bigg( \int_0^r  \tau \big( \varphi^N \sinh^n \big)^\p \, ds \bigg )^2 \\
		- \frac{2\tau\varphi^N}{\sinh^n} \bigg( \int_0^r  \tau \big( \varphi^N \sinh^n \big)^\p \, ds \bigg) + 	I_{V,\psi,\rho,\varphi}.
	\end{multline}
	Now, multiplying by $\big( \varphi^N \sinh^n \big)^\p / \varphi^{2N}$, it follows that
	\begin{multline} \label{eq.5a2 hyperbolic}
		\bigg( \frac{1}{\varphi^N \sinh^n} \Big( \int_0^r  \tau \big( \varphi^N \sinh^n \big)^\p \, ds \Big)^2 \bigg)^\p \\
		= \frac{\big( \varphi^N \sinh^n \big)^\p}{\varphi^{N - 1}} \Big( 2N  \big( \varphi^\pp + (n - 1) \coth \varphi^\p \big) + n^2 \varphi + \frac{I_{V,\psi,\rho,\varphi}}{\varphi^{2N}} \Big).
	\end{multline}
	Since
	$$
	\lim_{r \to 0} \bigg( \frac{1}{\varphi^N(r) \sinh^n (r)} \Big( \int_0^r  \tau \big( \varphi^N \sinh^n \big)^\p \, ds \Big)^2 \bigg) = 0,
	$$
	we deduce
	\begin{multline*}
		\Big(\int_0^r  \tau \big( \varphi^N \sinh^n \big)^\p \, ds \Big)^2  \\
		= \varphi^N \sinh^n \int_0^r \frac{\big( \varphi^N \sinh^n \big)^\p}{\varphi^{N - 1}} \Big( 2N  \big( \varphi^\pp + (n - 1) \coth \varphi^\p \big) + n^2 \varphi + \frac{I_{V,\psi,\rho,\varphi}}{\varphi^{2N}} \Big) \, ds
	\end{multline*}
	Observing that
	\begin{equation}\label{derivative S hyperbolic}
		\big(\sinh^n S_{V,\psi,\rho,\varphi}\big)^\p = \frac{\big( \varphi^N \sinh^n \big)^\p}{\varphi^{N - 1}} \Big( 2N  \big( \varphi^\pp + (n - 1) \coth \varphi^\p \big) + n^2 \varphi + \frac{I_{V,\psi,\rho,\varphi}}{\varphi^{2N}}\Big),
	\end{equation}
	the equation becomes
	\begin{equation}\label{increasing condition}
		\Big(\int_0^r  \tau \big( \varphi^N \sinh^n \big)^\p \, ds \Big)^2 = \varphi^N\sinh^{2n} S_{V,\psi,\rho,\varphi}.
	\end{equation}
	Therefore, $ S_{V,\psi,\rho,\varphi} \ge 0$ is a necessary condition for existence of solution, proving point $(ii)$.
	
	\medskip
	
	In order to prove $(iii)$, assume $S_{V,\psi,\rho,\varphi}>0$ and $(\varphi^N\sinh^n)^\p\neq0$ a.e.. Following the previous computation we have that equations \eqref{eq.5a1 hyperbolic} and \eqref{eq.5a2 hyperbolic} are equivalent and $(\varphi(r),W)$ is a solution if and only if equation \eqref{increasing condition} holds. Since $\varphi^N \sinh^{2n}S_{V,\psi,\rho,\varphi}\neq0$ on $\oo{0,\pi}$, we have that $\int_0^r  \tau \big( \varphi^N \sinh^n (s) \big)^\p \, ds$ does not vanish and has constant sign. Taking the square root we get
	\[
	\pm\int_0^r\tau(s)\big(\varphi^N\sinh^n\big)^\p(s)\,ds = \sqrt{\sinh^{2n} (r)\varphi(r)^N  S_{V,\psi,\rho,\varphi}(r)}.
	\]
	Taking the derivative with respect to $r$ we obtain
	\begin{align*}
		\tau &=  \pm \frac{\bigg(\sqrt{\big(\varphi^N \sinh^n\big) \big(\sinh^n S_{V,\psi,\rho,\varphi} \big)} \bigg)^\p}{\big(\varphi^N \sinh^n\big)^\p} \\
		& = \pm \frac{\big(\varphi^N \sinh^n\big)^\p \big(\sinh^n S_{V,\psi,\rho,\varphi} \big) + \big(\varphi^N \sinh^n\big) \big(\sinh^n S_{V,\psi,\rho,\varphi} \big)^\p}{2\big(\varphi^N \sinh^n\big)^\p \sinh^n\sqrt{\varphi^N S_{V,\psi,\rho,\varphi}}}.	
	\end{align*}
	In view of \eqref{derivative S hyperbolic}, we get
	\[
	\tau = \pm \frac{S_{V,\psi,\rho,\varphi} + 2N  \big( \varphi^\pp + (n - 1) \coth \varphi^\p \big) \varphi + n^2 \varphi^2 + I_{V,\psi,\rho,\varphi} \varphi^{-N}}{2\sqrt{\varphi^N  S_{V,\psi,\rho,\varphi}}}.
	\]
	The proof is competed.
\end{proof}

Similarly to Appendix \ref{appendix construct explicit solutions}, it is not difficult to find a set of $(V, \psi, \rho, \varphi)$ such that $S_{V,\psi,\rho,\varphi}$ is positive on hyperbolic spaces. Therefore, thanks to the theorem, a large class of explicit solutions to \eqref{CE} can be easily obtained, and so it gives a variety of models in general relativity, which can help us to find out important properties of initial data in both theoretical and numerical sense. 
\begin{remark} \label{remark S > 0 in vacuum on hyperbolic}
	In the vacuum case, the condition $S_{V,\psi,\rho,\varphi} > 0$ can be considerably simplified, making it quite simple to choose $\varphi$ so that the condition is satisfied. Since this result will be used later in the discussion of the sign of the AH mass, we write it here as an observation.

	If $V \equiv 0$ and $\psi = \rho = 0$, then $S_{V, \psi, \rho, \varphi} > 0$ if and only if  $\big( \big(\cosh(\frac{r}{2})\big)^{n-2} \varphi\big)^\p > 0$. In fact, in this case, by definition, $S_{V, \psi, \rho, \varphi} > 0$ is rewritten as
	$$
	N^2 \sinh(r) (\varphi^\p)^2 + n^2 \sinh(r) \varphi^2 + 2nN \cosh(r) \varphi^\p \varphi > 0.
	$$ 
	Dividing this inequality by $n^2\sinh(r) \varphi^2$, we obtain that
	\begin{equation} \label{vacuum S>0 hyperbolic}
		\bigg(\tfrac{N}{n} \frac{\varphi^\p}{\varphi} \bigg)^2 +2\coth(r) \bigg(\tfrac{N}{n} \frac{\varphi^\p}{\varphi} \bigg) + 1 > 0.
	\end{equation}
	Similarly to the proof of Theorem \ref{theorem no solution without TT-tensor}, since $x_1 = - \coth(r/2)$ and $x_2 = - \tanh(r/2)$ are two solutions to the quadratic equation $x^2 + 2 \coth(r) x + 1 = 0$, and since $\lim_{r \to 0^+} x_1 = - \infty$, it follows from the continuity of $\varphi^\p / \varphi$ that \eqref{vacuum S>0 hyperbolic} holds if and only if
	$$
	\tfrac{N}{n} \frac{\varphi^\p}{\varphi} > - \tanh \big(\tfrac{r}{2} \big).
	$$
	By straightforward calculations, this is equivalent to $\bigg( \big(\cosh(\tfrac{r}{2})\big)^{n-2} \varphi\bigg)^\p  > 0$, which is our claim. 
	
	\medskip
	
	Obviously, constructing such a positive function $\varphi$  is not difficult. For instance, letting $p$, $q$ and $k$ be positive constants with $(q,k)$ sufficiently large and defining
	\begin{equation} \label{example of varphi for sign of AH}
		\varphi := 1 \pm \frac{1}{kq} \bigg( \frac{r^q}{r^q + 1} \bigg) e^{-pr},
	\end{equation}
	we see immediately that this function satisfies the required condition.
\end{remark}

\medskip

In the next subsections, using Theorem  \ref{theorem existence on hyperbolic}, we address the problem of existence of spherically symmetric solutions to the conformal equations \eqref{CE} in the far--from--CMC regime, and we study the sign of the mass, as mentioned above.

\subsection{Existence of solutions with mean curvature of constant sign} \label{subsection existence solutions AH}
Before going further, we first define the standard Banach spaces on $\HH^n$. Given an integer $l \ge 0$, a H\"{o}lder exponent $\alpha \in [ 0 , 1]$, and a decay exponent $p \in ( 0 , n)$,  we will use is the weighted H\"{o}lder space $( C^{l , \alpha}_{-p}(\HH^n) , \|.\|_{C^{l , \alpha}_{-p}})$ to capture asymptotic of functions and tensors near infinity: a function $f$ is in $C^{l , \alpha}_{-p}(\HH^n)$ if $f\in C^{l , \alpha}_{\mathrm{loc}}$ and $\|f\|_{C^{l , \alpha}_{-p}} < \infty$, where
\begin{equation}\label{weighted Holder}
	\| f \|_{C^{l , \alpha}_{-p}} := \sup_{x \in  \HH} \Big( e^{p r(x)} \| f \|_{C^{l , \alpha}(B(x,1))} \Big).
\end{equation}
It will be clear from the context if the notation refers to a space of functions on $\HH^n$, or a space of sections of some bundle over $\HH^n$. We can easily check that $( C^{l , \alpha}_{-p} , \|.\|_{C^{l , \alpha}_{-p}})$ satisfies the weighted Sobolev embedding and Rellich theorem on $\HH^n$. A reference for this can be found in \cite[Proposition 2.3]{GicquaudSakovich}.

Based on these spaces, similarly to the previous section, we are interested in the conformal equations \eqref{CE} with $(V, \psi, \sigma)$ zero everywhere, that is
\begin{subequations}\label{CE with null sigma in hyperbolic}
	\begin{align}
		- \tfrac{4(n-1)}{n-2} \Delta \varphi + \Scal_g \varphi  
		& = - \tfrac{n-1}{n} \tau^2 \varphi^{N-1}  + \frac{|LW|^2 + \rho^2}{\varphi^{N + 1}} \\
		\hspace{1.8cm} \Div(L W) & = \tfrac{n-1}{n} \varphi^N d\tau,
	\end{align}
\end{subequations}
with $\varphi \longrightarrow 1$ at infinity. Heuristically, the analysis of \eqref{CE with null sigma in hyperbolic} is essentially the same as that of the vacuum system \eqref{CE} studied in \cite{IsenbergPark, Gicquaud, GicquaudSakovich}, where the existence of solutions is known only when $\tau$ is constant or near constant. Therefore, to obtain a more global understanding of the equations — particularly to seek convincing evidence to the existence of solutions to \eqref{CE} in the far-from-CMC regime — we will consider \eqref{CE with null sigma in hyperbolic} in the simpler setting where both $\tau$ and $\rho$ are radial. In this framework, since the analysis takes place within the space of radial functions as we have seen in Theorem \ref{theorem existence on hyperbolic}, rather than seeking a general solution $\varphi$ to \eqref{CE with null sigma in hyperbolic}, it is natural to restrict our consideration to the set of radial solutions.

\medskip

In this context, thanks to the limit equation criterion established by R. Gicquaud and A. Sakovich \cite{GicquaudSakovich}, we will follow the approach of Theorem \ref{theorem limit equation criterion sphere} and show that, as long as $\tau$ does not change sign, the conformal equations \eqref{CE with null sigma in hyperbolic} are always solvable, with solutions satisfying the expected properties. More precisely, the statement of our result is as follows.
\begin{theorem}[Solution with mean curvature of constant sign] \label{theorem limit equation criterion hyperbolic}
	Given a H\"{o}lder exponent $\alpha \in [ 0 , 1]$ and a decay exponent $p \in ( 0 , n)$, let $(\rho(r), \tau(r))$ be radial functions satisfying $(\rho(r), \tau(r) - n) \in  C^{0,\alpha}_{-p} (\HH^n) \times C^{1,\alpha}_{-p} (\HH^n)$. Assume that $\min \tau > 0$.
	
	Then, the conformal equations \eqref{CE with null sigma in hyperbolic} admit at least one solution $\varphi(r) > 0$ with $\varphi(r) - 1 \in C^{2,\alpha}_{-p}(\HH^n)$. Moreover, in this case, the set of solutions is compact.
\end{theorem}
\begin{proof}
	Similarly to the proof of Theorem \ref{theorem limit equation criterion sphere}, we first define the operator $T:RC^0(\HH^n) \to RC^0(\HH^n)$ as follows. For any radial function $\phi(r) \in RC^0(\HH^n)$, there exists a unique $1-$form $W \in C^{2,\alpha}_{-p}(\HH^n)$ satisfying
	$$
	\Div(L W) = \tfrac{n-1}{n} |\phi|^N d\tau.
	$$
	Recall that since $\phi(r)$ and $\tau(r)$ are radial, it follows from Proposition \ref{proposition 1-form W calculations}(c) that $|LW|$ is also radial. Therefore, taking $|LW|$ into the Lichnerowicz equation, there exists a unique $\varphi > 0$ such that $\varphi - 1 \in RC^{2,\alpha}_{-p}$ and 
	$$
	- \tfrac{4(n-1)}{n-2} \Delta \varphi + \Scal_g \phi + \frac{n-1}{n} \tau^2  \varphi^{N-1} 	= \big(|LW|^2 + \rho^2) \varphi^{- N - 1}.
	$$
	So we define
	\begin{equation} \label{compact operator}
		T(\phi) := \varphi.
	\end{equation}
	Following R. Gicquaud and A. Sakovich \cite[Section 3.3]{GicquaudSakovich}, we obtain that $T$ is continuous and compact. It is clear that a fixed point of $T$ is a solution to \eqref{CE with null sigma in hyperbolic}, therefore, the theorem will follows if we can show that $T$ has a fixed point. We will argue by contradiction. In fact, assume that $T$ has no fixed point. Then, according to the Leray--Schauder's fixed point theorem, there exists a sequence $\{(t_i,\phi_i(r))\} \subset (0,1) \times RC^{2,\alpha}_{-p}(\HH^n)$ such that 
	\begin{equation}\label{contradiction assumption hyperbolic}
		\phi_i = t_i T(\phi_i) \quad \text{and} \quad \|\phi\|_{C^0} \to +\infty.
	\end{equation} 
	By the definition of $T$, an analysis similar to \eqref{rce hyperbolic} shows that \eqref{contradiction assumption hyperbolic} can be rewritten as
	\begin{multline} \label{rce hyperbolic i}
		- 2N \varphi_i^{N + 1} (\varphi_i^\pp + h \varphi_i^\p ) - n^2\varphi_i^{N + 2} + \tau^2  \varphi_i^{2N} \\
		= \bigg( \frac{t_i^N}{\sinh^n} \int_0^r \tau^\p \varphi_i^N  \sinh^n \, ds \bigg)^2 + \tfrac{n}{n-1} \rho^2.
	\end{multline}
	where $\varphi_i := T(\phi_i)$ satisfies $\varphi_i(r) - 1 \in RC^{2,\alpha}_{-p}(\HH^n)$ and $\|\varphi_i\|_{C^0} \to +\infty$.
	
	\medskip
	
	Now, as in \cite{DahlGicquaudHumbert, GicquaudSakovich}, setting $\gamma_i := \|\varphi_i\|_{C^0}$, we rescale the functions $\varphi_i$ and $\rho$
	$$
	\varphit_i = \frac{\varphi_i}{\gamma_i} \quad \text{and} \quad \rhot_i = \frac{\rho}{\gamma_i^N}.
	$$
	Then the Lichnerowicz-type \eqref{rce hyperbolic i} can be rewritten as
	\begin{multline*}
		- \gamma_i^{2-N}\big(2N \varphit_i^{N + 1} (\varphit_i^\pp + h \varphit_i^\p ) - n^2\varphit_i^{N + 2} \big) + \tau^2  \varphit_i^{2N} \\
		= \bigg( \frac{t_i^N}{\sinh^n} \int_0^r \tau^\p \varphit_i^N  \sinh^n \, ds \bigg)^2 + \tfrac{n}{n-1} \rhot_i^2.
	\end{multline*}
	In view of \cite[Section 4]{GicquaudSakovich}, when $i \to +\infty$,  similarly to the proof of Theorem \ref{theorem limit equation criterion sphere}, (up to a subsequence) there exists $y(r) \in RC^0(\HH^n) \setminus \{0\}$ and $t_0 \in (0,1]$ such that, when $i\to+\infty$,
	\begin{align*}
		t_i &\longrightarrow t_0 \in (0,1], \\
		\frac{1}{\sinh^n} \int_0^r \tau^\p \varphit_i^N  \sinh^n \, ds  &\longrightarrow y \quad  \text{in $C^{0,\alpha}_{-p}$}, \\
		\varphit_i^N &\longrightarrow \frac{t_0 |y|}{\tau} \quad \text{in $C^0$}.
	\end{align*}
	In particular, it follows that
	\begin{subequations}
		\begin{align}
			|y(r)| & = \frac{t_0^N}{\sinh^n(r)} \bigg| \int_0^r \frac{\tau^\p(s)}{\tau(s)} |y(s)| \sinh^n(s) \, ds \bigg|, \label{eqn-y(r)-a} \\
			y(r) & \not\equiv 0. \label{eqn-y(r)-b} 
		\end{align}
	\end{subequations}
	However, an analysis similar to that in the proof of Theorem \ref{theorem limit equation criterion sphere} shows that Equation \eqref{eqn-y(r)-a} has no nontrivial solution, which is a contradiction with \eqref{eqn-y(r)-b}. The proof is completed.
\end{proof}

\subsection{The sharpness of the decay-exponent in the positive mass theorem} \label{subsection negative AH mass}
We now consider the mass functional for AH manifolds, a topical issue in general relativity. For our purpose, we restrict our discussion to the mass functional in $\HH^n$. Let us recall the definition of the mass functional as follows. 
Set
$$
\Ncal := \{ V \in C^\infty (\HH^n) ~|~ \Hess^b (V) = V b \},
$$
where $b$ is the standard metric on $\HH^n$. This is a vector space with a basis given by the functions
$$
V_0 = \cosh (r), \quad V_1 = x^1 \sinh (r), ..., V_n = x^n \sinh (r),
$$
where the functions $x^1, . . . , x^n$ are the coordinate functions on $\RR^n$ restricted to $\mathbb{S}^{n - 1}$. We define the linear mass functional $\Hgk$ of the AH  vacuum initial data $(\HH^n , g , k)$ on $\Ncal$ by
\begin{multline} \label{def mass functional H}
	\Hgk (V) := \lim_{r \to +\infty} \int_{\mathbb{S}_r} \Big( V (\Div^b e - d \tr^b e) + (\tr^b e) dV \\
	- (e + 2 \pi)(\nabla^b V , .) \Big) (\dr) \, d \mu^b,
\end{multline}
where $e := g - b$ and $\pi := (k - g) - \tr_g (k - g) g$. For $\Hgk$ to be finite and independent of the chart at infinity, the standard decay assumption of $(g , k)$ usually added to our study is 
\begin{equation} \label{decay rate for AH mass}
	(e , \pi) \in C^{2 ,\alpha}_{-p} \times C^{1 ,\alpha}_{-p}, 	
\end{equation}
where $\alpha \in (0 , 1)$, $p > {n /2}$ and $C^{2 ,\alpha}_{-p}$, $C^{1 ,\alpha}_{-p}$ are the weighted H\"{o}lder spaces defined in \eqref{weighted Holder}. 
Now, writing
$$
\Mgk := \Hgk (V_0),
$$
we will refer to $\Mgk$ as the mass of $(\HH^n , g , k)$. A long-standing conjecture in general relativity states that the mass of an AH vacuum initial data set is positive unless it is a Cauchy hypersurface of Minkowski space. Although there is substantial evidence supporting this conjecture, so far it has only been proven to be true under the decay condition \eqref{decay rate for AH mass} for vacuum three-dimensional initial data settings, following Sakovich \cite{SakovichMassAH}. The aim of this subsection is not to address the conjecture, but we will be concerned with it, investigating how the decay rate $p$ in \eqref{decay rate for AH mass} plays an essential role in the sign of mass. More precisely, in the spirit of the recent work of the second author for ADM mass of AF manifolds \cite{NguyenRadialAF}, we will show that when the decay rate of $(e , \tau)$ is not super-critical, that is $p \le n/2$, the AH mass $\Mgk$ fails to be positive in general.  
 
\begin{theorem}[Arbitrary sign of AH mass]\label{thm-mass_hyperbolic} \label{theorem mass sign AH}
	Let $\varphi(r) > 0$ be a radial function on $(\HH^n, b)$. Assume that $\varphi(r) - 1 \in RC^{3, \alpha}_{p}(\HH^n)$ with $\alpha \in (0 , 1)$ and $p>0$, and that $V\equiv0$, $\psi\equiv0$, $\rho\equiv0$ and $S_{V,\psi,\rho,\varphi}>0$ on $\RR_+$. Let $\tau$ be given by \eqref{AHidenity} and set
	$$
	W := \tfrac{n}{n-1} F^\p_v dr
	$$
	where $v = \tfrac{n-1}{n} \varphi^N\tau^\p$ and $F_v$ is defined in \eqref{def F nonpositive Ric}. Then 
	\begin{equation} \label{H solution}
		(g, k):= \big( \varphi^{N - 2} b, \, \frac{\tau}{n} \varphi^{N - 2} b + \varphi^{-2} LW \big)
	\end{equation} 
	is an AH solution to the vacuum Einstein constraint equations \eqref{constraint}. Moreover, the AH mass of this solution has the following property
	\begin{itemize}
		\item if $p > \frac{n}{2}$,  then $\Mgk = 0$,
		\item if $p = \frac{n}{2}$, then $\Mgk$ is finite and can have an arbitrary sign,
		\item if $ p < \frac{n}{2}$, then $\Mgk$ can reach $\pm \infty$.
	\end{itemize}
\end{theorem}
\begin{proof}
	Thanks to Theorem \ref{theorem existence on hyperbolic} and the conformal method, it is clear that $(g,k)$ defined in \eqref{H solution} is a solution to the vacuum constraint equations \eqref{constraint} on the hyperbolic manifold. Let us now consider the mass of $(g,k)$. We first compute $e$ and $\pi$ from their definitions
	\begin{align*}
		e & = g - b = (\varphi^{N-2}-1)b, \\
		\pi & = (k-g) - (\tau-n)g = \frac{n-1}{n}(n-\tau)\varphi^{N-2}b + \varphi^{-2} LW.
	\end{align*}
	Using the fact that $V_0(r)=\cosh r$, direct computations give
	\begin{align*}
		\Div^b(e)(\partial r) & = (N-2)\varphi^{N-3}(r)\varphi^\p(r), \\
		d\tr^b(e)(\partial r) & = n(N-2)\varphi^{N-3}(r)\varphi^\p(r), \\
		(e+2\pi)(\nabla^bV_0,\partial r) & = \sinh(r)\bigg((2n-1)\varphi^{N-2}(r) - 1 - \tfrac{2(n-1)}{n}\tau(r)\varphi^{N-2}(r) \\
			& \qquad + 2\varphi^{-2}(r)LW(\partial r,\partial r)\bigg),
	\end{align*}
	and therefore
	\begin{equation} \label{conformal mass}
		H_{(g,k)}(V_0) = (n-1)\lim_{r\to+\infty} \int_{\mathbb{S}_r}
		A(r)\sinh(r) \, d\mu^b,
	\end{equation}
	with
	\begin{multline} \label{eqn-integrand}
		A(r)= \frac{2}{n}\tau(r)\varphi^{N-2}(r) - \frac{2}{n-1}\varphi^{-2}(r)LW(\partial r,\partial r) \\
		- \Big( 1+\varphi^{N-2}(r)+(\varphi^{N-2})^\p(r)\coth(r)\Big).
	\end{multline}
	Since $\varphi(r)-1 \in C^{3 ,\alpha}_{p}(\HH^n)$, we can write $\varphi$ at infinity in the form 
	\begin{equation} \label{decay of varphi H}
		\varphi(r) = 1 + a(r) e^{- p r},
	\end{equation}
	where $a$ is a $C^2-$bounded function. We get
	\begin{equation}  \label{decay of varphi H-2}
		\varphi^\p = (a^\p(r) - pa(r)) e^{-pr} \quad \mbox{and} \quad \varphi^\pp = (a^\pp(r) - 2pa^\p(r) + p^2a(r)) e^{-pr},
	\end{equation}
	and applying Taylor's expansion near $0$
	\begin{equation} \label{Taylor}
		(1 + \eps)^\beta = 1 + \beta \eps + \frac{\beta(\beta-1)}{2} \eps^2 + O(\eps^3),
	\end{equation}
	we have at infinity
	\begin{multline} \label{eqn-integrand1}
		1+\varphi^{N-2}(r)+(\varphi^{N-2})^\p(r)\coth(r) = 2 + (N-2)\big(a(r)+(a^\p(r)-pa(r))\big)e^{-pr} \\
			+ (N-2)(N-3)\big(\frac{1}{2}a^2(r) + (a^\p(r)-pa(r))\coth(r)\big)e^{-2pr} + O(e^{-3pr}).
	\end{multline}
	
	\medskip
	
	Now, thanks to Proposition \ref{proposition 1-form W calculations}(a), we calculate $(LW)_{rr}$
	\begin{align*}
		LW(\partial r,\partial r) 
		& = 2(n - 1)\bigg(\frac1n f_{\varphi,\tau}(r) - \coth(r)F_v^\p(r) + F_v(r) \bigg) \\
		& = \tfrac{n-1}{n}\bigg(\tau(r)\varphi^N(r) - \frac{1}{\sinh^n (r)}\int_0^r \tau(s)(\varphi^N\sinh^n)^\p(s) \,ds \bigg)
	\end{align*}
	and using Equation \eqref{increasing condition} we get
	\begin{multline}
		\tfrac2n\tau(r)\varphi^{N-2}(r) - \tfrac2{n-1}\varphi^{-2}(r)LW(\partial r,\partial r) \\
			= 2\varphi^{-2}(r)\sqrt{N^2(\varphi^\p)^2(r)\varphi^N(r) + n^2\varphi^{N + 2}(r) + 2nN\coth(r)\varphi^\p(r)\varphi^{N +1}(r)}.
	\end{multline}
	Applying again Taylor's expansion \eqref{Taylor}, we obtain by \eqref{decay of varphi H} and \eqref{decay of varphi H-2}
	\begin{multline} \label{eqn-integrand2}
		\tfrac2n\tau(r)\varphi^{N-2}(r) - \tfrac2{n-1}\varphi^{-2}(r)LW(\partial r,\partial r)  \\
			= 2 + (N-2)\big(a(r)+(a^\p(r)-pa(r))\big)e^{-pr} + B(r)e^{-2pr} + O(e^{-3pr}),
	\end{multline}
	where
	\begin{multline*}
		B(r) = \tfrac{(N-2)(N-4)}{4}a^2(r)+\tfrac{(N-2)(N-4)}{4}\big(a^\p(r)-pa(r)\big)\coth(r) \\
		+ \tfrac{(N-2)^2}{4}\big(a^\p(r)-pa(r)\big)(1-\coth^2(r)).
	\end{multline*}
	Combining Equations \eqref{eqn-integrand}, \eqref{eqn-integrand1} and \eqref{eqn-integrand2} we obtain
	\begin{align*}
		A(r) & = \tfrac{(N-2)^2}{4}\Big(-a^2(r)-2\big(a^\p(r)-pa(r)\big)\coth(r) \\
			& \qquad + \big(a^\p(r)-pa(r)\big)\big(1-\coth^2(r)\big)\Big)e^{-2pr} + O(e^{-3pr}) \\
		 & = \tfrac{(N-2)^2}{4}\Big(-a^2(r)-2\big(a^\p(r)-pa(r)\big)\Big)e^{-2pr} + O(e^{-2(p+1)r}) + O(e^{-3pr})
	\end{align*}
	and, since $|\mathbb{S}_r|=\sigma_n\sinh^{n-1}(r)$,
	\begin{multline*}
		(n-1)\int_{\mathbb{S}_r}A(r)\sinh(r)\,d\mu^b \\
		= \tfrac{(n-1)(N-2)^2\sigma_n}{2^{n+2}}\Big((-a^2(r)-2(a^\p(r)-pa(r))\Big)e^{n-2pr} \\
		+ O(e^{n-2(p+1)r}) + O(e^{n-3pr}).
	\end{multline*}
	If the limit exists, it follows that
	$$
	H_{(g,k)}(V_0) = \tfrac{(n-1)(N-2)^2\sigma_n}{2^{n+2}}\lim_{r\to+\infty}\Big(-a^2(r)-2(a^\p(r)-pa(r))\Big)e^{n-2pr}.
	$$
	Since $a$ and $a^\p$ are bounded, this implies that $\Mgk=0$ as long as $p>\frac{n}{2}$. 
	
	\medskip
	
	When $p\le\frac{n}{2}$, the limit may not exist, however, if $a$ and $a^\p$ have limits, then we must have
	$$
	\lim\limits_{r\to+\infty} a^\p(r) = \lim\limits_{r\to+\infty} \frac{a}{r} = 0
	$$ 
	and
	$$
	H_{(g,k)}(V_0) = \tfrac{(n-1)(N-2)^2\sigma_n}{2^{n+2}}\ell(2p-\ell)\lim_{r\to+\infty}e^{n-2pr},
	$$
	where $\ell=\lim\limits_{r\to+\infty}a(r)$. Here, we observe that, as can be easily verified through Example \ref{example of varphi for sign of AH} in Remark \ref{remark S > 0 in vacuum on hyperbolic}, this limit $\ell$ can have an arbitrary sign. Therefore, it follows in particular that if $p=\frac{n}{2}$, then $\Mgk$ is finite and can also have an arbitrary sign, and if $p<\frac{n}{2}$, then $\Mgk$ can reach $\pm \infty$ as claimed. The proof is completed.
\end{proof}

\begin{remark} \label{remark Birkhoff theorem}
	It is worth noting that since the vacuum solution $(g,k)$ constructed in the proof of Theorem \ref{theorem mass sign AH} is spherically symmetric and regular, by Birkhoff's theorem, such a $(\HH^n, g, k)$ in the theorem is indeed Cauchy initial data for Minkowski space.
\end{remark}

\section{Solutions in the Euclidean space} \label{section Euclidean}
In this section, we will study the Lichnerowicz-type equation \eqref{RCE} in the Euclidean space, that is providing the proof of Main Theorem \ref{main theorem} in this setting and exploring its applications. The main results obtained here are motivated by those of the previous sections, with appropriate adaptations to the Euclidean framework.

\medskip

For the existence of solutions to \eqref{CE with null sigma in sphere intro}, unlike the spherical and hyperbolic cases, where well-known arguments for general data can be adapted to solve the radial conformal equations \eqref{CE with null sigma in sphere intro} in the non-CMC regime, such tools are not available in the Euclidean case. For instance, as we have seen in Theorems \ref{theorem limit equation criterion sphere} and \ref{theorem limit equation criterion hyperbolic}, the limit equation criterion developed in \cite{DahlGicquaudHumbert} and \cite{GicquaudSakovich} guarantees respectively that \eqref{CE with null sigma in sphere intro} on spherical and hyperbolic manifolds admit solutions whenever the mean curvature $\tau$ does not change sign. However, this achievement does not apply in the Euclidean space since $\tau$ must vanish at infinity by definition, which breaks the non-changing sign condition, except for $\tau$ identically zero. For this reason, to obtain an existence result for solutions to the conformal equations \eqref{CE} analogous to those in the previous sections, we should develop a method that goes beyond the general techniques available so far. In this approach, a natural idea is to follow the strategy used in the proof of Theorem \ref{theorem smooth time-derivative sphere}. The main difficulty in adapting this method to the Euclidean setting lies in refining the arguments so as to pass from positive scalar curvature on the sphere to null scalar curvature in the Euclidean space, while simultaneously proving the non-blow-up of solutions at infinity. In this section, we solve these difficulties and establish the corresponding existence result. Moreover,in contrast to the spherical case, the solvability of \eqref{CE with null sigma in sphere intro} here are not constrained by the vanishing condition \eqref{vanishing condition sphere}, therefore, the equations always admit solutions in any radial mean curvature regime, with the set of solutions uniformly bounded independently of $\tau$.

\medskip

For adapting the ideas from hyperbolic spaces in Section \ref{section hyperbolic} to the Euclidean setting, we focus on the sign of the ADM mass in AF manifolds. Similar to Theorem \ref{theorem mass sign AH}, we will discuss in this section the role of the decay rate of the symmetric $(0,2)$-tensor $k$ in determining the sign of the mass. The result we obtain is that when the decay rate of $k$ is critical, the mass may become negative. In particular, this demonstrates that the usual assumption on the decay exponent of $k$ in the Positive Mass Theorem is sharp.

\medskip

The outline of this section is as follows. In Subsection \ref{subsection Lic-type euc}, the proof of Main Theorem \ref{main theorem} on Euclidean spaces will be provided. In Subsection \ref{subsection existence euc}, we establish the existence of solutions to the conformal equations \eqref{CE with null sigma in sphere intro} for all radial mean curvature.  Finally, in Subsection \ref{subsection mass euc}, the role of decay rate of $k$ at infinity in the sign of ADM mass will be discussed. 

\subsection{Lichnerowicz-type equation in the Euclidean space} \label{subsection Lic-type euc} 
Let us first denote by $\RR^n$ the $n-$dimensional Euclidean space, with $n \ge 3$, and fix a point in $\RR^n$ as an origin. Then, in geodesic normal coordinates at this point, the Euclidean metric reads $\delta_{\text{Euc}} = dr^2 + r^2 (r) \sigma_s$, where $\sigma_s$ is the standard round metric on $\mathbb{S}^{n-1}$ and $r$ is the distance from the origin. The functions $h, \phi, u_0$ and the constants $\Ric$ and $\Scal_g$ are given by
\begin{equation} \label{computing in Euclidean}
	\begin{gathered}
		\theta(r) = r^{n - 1}, \quad h(r) = \frac{n-1}{r}, \quad u_0(r) = 1 \\
		\Ric = 0, \quad \Scal=0.
	\end{gathered}
\end{equation}

As in the previous sections, from a radial data set $(V,\tau(r),\psi(r),\rho(r))\in C^0(\RR) \times RC^1(\RR^n) \times RC^1(\RR^n) \times RC^0(\RR^n)$ and a radial function $\varphi(r) \in RC_+^3(\RR^n)$ we define:
\[
I_{V,\psi,\rho,\varphi} := \tfrac{n}{n-1} \bigg( 2 V(\psi) \varphi^{2N} + |\psi^\p|^2 \varphi^{N+2} + \rho^2 \bigg),
\]
and
\[
S_{V,\psi,\rho,\varphi} := N^2 (\varphi^\p)^2 + \frac{2nN}{r} \varphi \varphi^\p + \frac{1}{r^n} \int_0^r \frac{I_{V,\psi,\rho,\varphi}}{\varphi^{2N}} (\varphi^N s^n)^\p  \, ds.
\]
Similarly to Theorems \ref{theorem existence on sphere} and \ref{theorem existence on hyperbolic} in the spherical and hyperbolic spaces, we have the following result.
\begin{theorem}[Solutions with freely specified conformal factor] \label{theorem existence on Euclidean}
	Consider a radial data set $(V,\tau(r),\psi(r),\rho(r))\in C^0(\RR) \times RC^1(\RR^n) \times RC^1(\RR^n) \times RC^0(\RR^n)$ with $\rho \psi^\p \equiv 0$. For $\varphi(r) \in RC_+^3(\RR^n)$ we have:
	\begin{enumerate}[(i)]
		\item $\varphi(r)$ is a solution to the conformal constraint equations \eqref{CE with null sigma} if and only if
			\begin{equation} \label{rce euclic}
			- 2N \varphi^{N + 1} (\varphi^\pp + h \varphi^\p ) + \tau^2  \varphi^{2N} = \bigg( \frac{1}{r^n} \int_0^r \tau^\p \varphi^N  r^n \, ds \bigg)^2 + I_{V,\psi,\rho,\varphi}.
			\end{equation}
		
		\item If $\varphi$ is a solution to \eqref{rce euclic}, then $S_{V,\psi,\rho,\varphi} \ge 0$.
		
		\item If $S_{V,\psi,\rho,\varphi}>0$ on $\RR_+$ and $(\varphi^Nr^n)^\p \neq 0$ a.e on $\RR_+$, then  $\varphi$ is a solution to \eqref{rce euclic} if and only if the function $\tau$ satisfies
		\begin{equation}\label{AEidenity}
			\tau(r)= \pm \frac{S_{V,\psi,\rho,\varphi}(r) + 2N  \big( \varphi^\pp(r) + \tfrac{(n - 1)}{r} \varphi^\p(r) \big) \varphi(r) + I_{V,\psi,\rho,\varphi}(r) \varphi^{-N}(r)}{2\sqrt{\varphi^N(r) S_{V,\psi,\rho,\varphi}(r)}}.
		\end{equation}
	\end{enumerate}
	Moreover, in this case, the $1-$form $W$ can be computed by \eqref{solution to vector equations} with $v=\tfrac{n-1}{n}\varphi^N\tau^\p$.
\end{theorem}
\begin{proof}
	The proof is straightforwardly similar to that of Theorems \ref{theorem existence on sphere} and \ref{theorem existence on hyperbolic}, with the primary distinction being the replacement of Identities \eqref{derivative of S} and \eqref{derivative S hyperbolic} by  
	$$
	\big(r^n S_{V,\psi,\rho,\varphi}\big)^\p
	= \frac{\big( \varphi^N r^n \big)^\p}{\varphi^{N - 1}} \Big( 2N  \big( \varphi^\pp + \tfrac{(n - 1)}{r} \varphi^\p \big) + \frac{I_{V,\psi,\rho,\varphi}}{\varphi^{2N}}\Big).
	$$
\end{proof}

\subsection{Existence of solutions with freely specified mean curvature} \label{subsection existence euc}
Although the Euclidean space is the most familiar setting, existence results for \eqref{CE} are almost restricted to the CMC case. This is because, apart from the limit equation criterion — which, as explained at the beginning of this section, has not yet been established for the conformal equations on AF manifolds — there are essentially only two approaches available for proving existence of solutions in the non-CMC regime. The first relies on the smallness of $\frac{d\tau}{\tau}$ to prevent blow-up of solutions, corresponding to the near-CMC case. The second is based on the smallness of $\sigma$ and $\rho$, leading to small solutions and is commonly called the small TT-tensor method. Unfortunately, despite the results in \cite{DiltsIsenbergMazzeoMeier}, none of these methods is convincingly applicable on $\RR^n$. On the one hand, the weighted Poincaré inequality does not allow us to assume that $\frac{d\tau}{\tau}$ is too small; on the other hand, the fact that solutions converges to $1$ at infinity breaks the smallness assumption on solutions required to implement the small TT-tensor method.

\medskip

In this subsection, as an application of Theorem \ref{theorem existence on Euclidean}, we will continue considering \eqref{CE} with $(V, \psi, \sigma)$ zero everywhere, that is
\begin{subequations}\label{CE with null sigma in euclic}
	\begin{align}
		- \tfrac{4(n-1)}{n-2} \Delta \varphi + \Scal_g \varphi  
		& = - \tfrac{n-1}{n} \tau^2 \varphi^{N-1}  + \frac{|LW|^2 + \rho^2}{\varphi^{N + 1}} \\
		\hspace{1.8cm} \Div(L W) & = \tfrac{n-1}{n} \varphi^N d\tau,
	\end{align}
\end{subequations}
with $\varphi \longrightarrow 1$ at infinity. The close analytical similarity of these equations to the vacuum \eqref{CE} and their importance in providing a global perspective to the solvability of \eqref{CE} have been explained briefly in the previous sections. For that, similarly to what we have done previously, to establish convincing evidence supporting that the conformal method is still an useful tool for constructing solutions to the constraint equations \eqref{constraint} on AF manifolds, we study  \eqref{CE with null sigma in euclic} particularly in the simple model where $(M,g)$ is Euclidean and $(\tau, \rho)$ are assumed to be radial.  In this context, motivated by the proof of Theorem \ref{theorem smooth time-derivative sphere}, the main result we obtain in this section is that as long as $\rho$ is $C^1-$regular, the equations \eqref{CE with null sigma in euclic} are always solvable, with solutions uniformly bounded independently of $\tau$. Before presenting the precis statement and providing the proof, we first introduce the associated Banach spaces on $(\RR^n, \sigma_{\text{Euc}})$ as follows. Given an integer $l \ge 0$, a H\"{o}lder exponent $\alpha \in [ 0 , 1]$, and a decay exponent $ \beta> 0$, we will use the weighted H\"{o}lder spaces $C^{l , \alpha}_{-\beta}$ to capture asymptotic of functions and tensors near infinity. For $\alpha = 0$, we will write $C^{l}_{- \beta}$ instead of $C^{l , 0}_{- \beta}$. The weighted norm convention we are using is that the $C^{ s , \alpha}_{-\beta}$ norm is
given by
\begin{multline*}
	\| f \|_{l , \alpha , -\beta} := \sum_{|s| \le l}\sup_{\RR^n} \big( \zeta^{|s| + \beta} |\del^s f|\big) \\
	+ \sum_{|s| = l} \sup_{\RR^n} \bigg(\zeta^{l + \beta + \alpha} \sup_{0 < |y - x| \le \zeta} \Big( \frac{| \del^s f(y) - \del^s f(x)|}{|y - x|^\alpha}\Big) \bigg),
\end{multline*}
where in this context $\zeta$ is a positive function which equals $|x|$ outside the unit ball and $s$ is a multi-index. It will be clear from the context whether the notation refers to a space of functions on $\RR^n$, or a space of sections of some bundle over $\RR^n$.

\begin{theorem}[Smooth time-derivative solutions and stability] \label{theorem smooth tiem-derivative euclid}
	Given a H\"{o}lder exponent $\alpha \in (0,1)$ and a decay exponent $\beta \in ( 1 , n)$, let $(\rho(r), \tau(r))$ be radial functions on $\RR^n$ with $\tau(r) \in C^{1,\alpha}_{-\beta}(\RR^n)$ and assume that
	\begin{equation} \label{smooth condition rho}
		\rho(r) \in C^{1,\alpha}_{-\beta} (\RR^n).
	\end{equation}
	
	Then, the conformal equations \eqref{CE with null sigma in euclic} admit at least one solution $\varphi(r) > 0$ with $\varphi(r) - 1 \in RC^{3,\alpha}_{- \min\{2\beta - 2, n-2\}}(\RR^n)$. Moreover, if we fix $\rho(r) \in RC^{1,\alpha}_{-\beta} (\RR^n)$ and for any $\tau(r) \in RC^{1,\alpha}_{-\beta} (\RR^n)$ we denote by $\mathcal{S}_{\RR^n}(\tau)$ the set of all radial solutions to \eqref{CE with null sigma in euclic} associated with $(\rho, \tau)$, then there exists a constant $C>0$ depending only on $\rho$ such that 
	$$
	\sup\Big\{\|\varphi(r)\|_{C^0(\RR^n)}\ \Big|\ \text{$\exists \tau \in RC^{1,\alpha}_{-\beta} (\RR^n)$ such that $\varphi \in \mathcal{S}_{\RR^n}(\tau)$}  \Big\} \le C.	
	$$
\end{theorem}
\begin{proof}
	We first define the operator $\Tt: [0,1] \times RC^0(\RR^n) \to RC^0(\RR^n)$ as follows. For any $\phi \in RC^0(\RR^n)$,  there exists a unique $W \in C^{2 , \alpha}_{- \min\{\beta - 1, n -2\}}(\RR^n)$ satisfying
	$$
	\Div(L W_\phi) = \tfrac{n-1}{n} |\phi|^N d\tau.
	$$
	Of course, by Propositions \ref{proposition 1-form W calculations} and \ref{proposition resolve vector equations}, we have $|LW_\phi| \in RC^{1,\alpha}_{- \min\{\beta - 1, n -1\}}(\RR^n)$.  Therefore, since $\tau$ is radial and Laplace's operator $\Delta$ is invariant under rotations, there exists a unique $\varphi>0$ such that  $\varphi(r) - 1 \in RC^{3, \alpha}_{- \min\{2\beta - 2, n - 2\}}$ and
	$$
	- \tfrac{4(n-1)}{n-2} \Delta \varphi  + \tfrac{n-1}{n} t^{2N} \tau^2 \varphi^{N-1} 
	= (|LW_\phi|^2 + \rho^2) \varphi^{- N - 1},
	$$
	which reads
	$$
	- \tfrac{4(n-1)}{n-2} \varphi^{N + 1} (\varphi^\pp + h \varphi^\p ) + \tfrac{n-1}{n} t^{2N}\tau^2 \varphi^{2N}
	= \tfrac{n-1}{n} \bigg( \frac{1}{r^n} \int_0^r \tau^\p \varphit_i^N  s^n \, ds  \bigg)^2 + \rho^2,
	$$
	We define
	$$
	\Tt(t,\phi) := t \varphi.
	$$
	In view of \cite[Section 5]{Diltsthesis}, we obtain that $\Tt$ is continuous compact and it is clear that a fixed point of $\Tt(1,)$ is a solution to \eqref{CE with null sigma in euclic}. Now, we set
	$$
	K = \big\{ (t,\phi(r)) \in [0,1] \times RC^0(\RR^n) \ \big|\  \phi = \Tt(t,\phi) \big\}.
	$$ 
	By Leray--Schauder's fixed point theorem, to show that $\Tt(1,.)$ has a fixed point, it suffices to show that $K$ is bounded. In fact, let  $(t,\phi(r))$ be an arbitrary couple in $K$. If $t=0$, then $\phi = \Tt(0,\phi) = 0$, there is nothing to prove. If $t \ne 0$, then the fact $\phi = \Tt(t,\phi)$ can be rewritten as
	\begin{align*}
		-\tfrac{4(n-1)}{n-2} \Delta \varphi  + \tfrac{n-1}{n} t^{2N} \tau^2 \varphi^{N-1} 
		& = (|LW|^2 + \rho^2) \varphi^{- N - 1} \\ 
		\Div(L W) & = \tfrac{n-1}{n} t^N\varphi^N d\tau,
	\end{align*}
	where $\varphi := \frac{\phi}{t}$ satisfies $\varphi(r) - 1 \in RC^{3, \alpha}_{- 2 \beta + 2}(\RR)$. A function $\varphi$ is a solution of this system means that $\varphi$ is a solution to the conformal equations \eqref{CE with null sigma in euclic} associated with $(t^N \tau, \rho)$. Therefore, since $\tau(r)$, $\rho(r)$ and $\varphi(r)$ are radial, it follows by Theorem \ref{theorem existence on Euclidean} that $S_{V,\psi,\rho,\varphi}\ge0$, and since $V\equiv0$, $\psi\equiv0$ we get
	\begin{equation}\label{condition on varphi}
		N^2 (\varphi^\p)^2 + \frac{2nN}{r} \varphi \varphi^\p + \frac{n}{(n-1)r^n}\int_0^r \frac{\rho^2}{\varphi^{2N}} (\varphi^N s^n)^\p  \, ds \ge 0.
	\end{equation}
	Now, let $r_0 \ge 0$ be the maximum point of $\varphi$. Given a small constant $\delta >0$ which will be chosen later in \eqref{condition delta}, we set
	$$
	\rt_0 = \sup \bigg\{r_\star \ge r_0 ~\bigg|~ \varphi^\p(r) \ge - \frac{1}{(r+1)^{1 + \delta}} \quad \forall r \in [r_0, r_\star] \bigg\}.
	$$
	Of course, if $\rt_0 = +\infty$, then $\varphi^\p(r) + (r + 1)^{-1 - \delta} \ge 0$ for all $r \in [r_0, +\infty)$, and hence
	$$
	\varphi(r_0) - \frac{1}{\delta(r_0 + 1)^\delta} \le \lim_{r \to +\infty}\bigg(\varphi(r) - \frac{1}{\delta (r + 1)^\delta} \bigg) = 1,
	$$
	which implies
	\begin{equation} \label{rt0 = infty}
		\phi \le \|\varphi\|_{C^0} = \varphi(r_0) \le 1 + \frac{1}{\delta}.
	\end{equation} 
	If $\rt_0 \in (r_0,+\infty)$, we have by definition 
	\begin{equation}\label{rt equal}
		\varphi^\p(\rt_0) = - \frac{1}{(\rt_0 + 1)^{1 + \delta}},
	\end{equation}
	and
	\begin{equation}\label{varphit increasing}
		\varphi^\p(r) \ge - \frac{1}{(r + 1)^{1 + \delta}}, \quad \forall r \in [r_0, \rt_0].
	\end{equation}
	On the one hand, it is clear from \eqref{varphit increasing} that $\varphi - \frac{1}{\delta(r + 1)^{\delta}}$ is increasing in $[r_0, \rt_0]$, and so
	\begin{equation}\label{varphi unbounded euclic}
		\varphi(\rt_0) \ge \varphi(r_0) - \frac{1}{\delta} \bigg( \frac{1}{(r_0 + 1)^\delta} - \frac{1}{(\rt_0 + 1)^\delta} \bigg) \ge \varphi(r_0) - \frac{1}{\delta}.
	\end{equation}
	On the other hand, by  \eqref{condition on varphi}, we have particularly at $\rt_0$
	$$
	N^2 (\varphi^\p)^2(\rt_0) + \frac{2nN}{\rt_0} \varphi(\rt_0) \varphi^\p(\rt_0) + \frac{n}{(n - 1)\rt_0^n}\int_0^{\rt_0} \frac{\rho^2}{\varphi^{2N}} (\varphi^N s^n)^\p  \, ds \ge 0.
	$$
	Multiplying by $\frac{\rt_0}{\varphi^\p(\rt_0)}$ and taking \eqref{rt equal} into account, we get that
	$$
	2nN \varphi(\rt_0) \le \tfrac{n}{n-1} \frac{(\rt_0 + 1)^{1 +\delta}}{\rt_0^{n-1}}\int_0^{\rt_0} \frac{\rho^2}{\varphi^{2N}} (\varphi^N s^n)^\p  \, ds +  \frac{N^2}{(\rt_0 + 1)^{\delta}}.
	$$
	Since 
	$$
	\frac{\rho^2(\varphi^N r^n)^\p}{\varphi^{2N}} = - (\rho^2 r^{2n}) \bigg( \frac{1}{\varphi^N r^n} \bigg)^\p,
	$$
	integrating by parts, it follows
	\begin{equation}\label{eq condition euc}
		2nN \varphi(\rt_0) \le \tfrac{n}{n-1} \bigg(\frac{(\rt_0 + 1)^{1 + \delta}}{\rt_0^{n-1}}\int_0^{\rt_0} \frac{(\rho^2 s^{2n})^\p}{\varphi^N s^n} \, ds  - \frac{(\rt_0 + 1)^2\rho^2 (\rt_0)}{\varphi^N(\rt_0)} \bigg) +  \frac{N^2}{(\rt_0 + 1)^{\delta}}.
	\end{equation}
	For estimating the integral term, with
	$$
	\beta_0 := \min \bigg\{\beta, \, \frac{n}{2} - \frac{1}{4} \bigg\} > 1,
	$$
	we observe that since $\rho(r) \in C^{1,\alpha}_{-\beta}(\RR^n) \subset C^{1,\alpha}_{-\beta_0}(\RR^n)$,
	\begin{align*}
		\int_0^{\rt_0} \frac{(\rho^2 s^{2n})^\p}{\varphi^N s^n} \, ds & = \int_0^{\rt_0} \frac{2(\rho\rho^\p s^n + n\rho^2s^{n-1})}{\varphi^N} \,ds \\
		& \le 4n \|\rho\|^2_{C^{1,\alpha}_{-\beta_0}} \int_0^{\rt_0} \frac{\min\{s^{n - 1 - 2\beta_0}, s^{n-1}\}}{\varphi^N} \, ds.
	\end{align*}
	Then, letting $\varphi_\star > 0$ be the unique solution to the Lichnerowicz equation
	$$
	- \tfrac{4(n-1)}{n-2} \Delta \varphi_\star  + \tfrac{n-1}{n} \tau^2 \varphi_\star^{N-1} = 0,
	$$
	since $\varphi \ge \varphi_\star$, we deduce from the inequality that
	\begin{align*}
		\int_0^{\rt_0} \frac{(\rho^2 s^{2n})^\p}{\varphi^N s^n} \, ds 
		& \le 4n \frac{\|\rho\|^2_{C^{1,\alpha}_{-\beta_0}}}{(\min\varphi_\star)^N} \int_0^{\rt_0} \min\{s^{n - 1 - 2\beta_0}, s^{n-1}\} \,ds. \\
		& \le 4n \frac{\|\rho\|^2_{C^{1,\alpha}_{-\beta_0}}}{(\min\varphi_\star)^N} \min \bigg\{ \frac{\rt_0^{n-2\beta_0}}{n-2\beta_0}, \, \frac{\rt_0^n}{n} \bigg\}
	\end{align*}
	Taking into \eqref{eq condition euc}, we get
	\begin{multline*}
		2nN \varphi(\rt_0) \le \tfrac{4n^2}{n-1} \frac{\|\rho\|^2_{C^{1,\alpha}_{-\beta_0}}}{(\min\varphi_\star)^N} \bigg( \min \bigg\{ \frac{(\rt_0 + 1)^{\delta-2(\beta_0 - 1)}}{n-2\beta_0}, \, \frac{(\rt_0 + 1)^{2 +\delta}}{n} \bigg\}\bigg) \\
		- \tfrac{n}{n-1}\frac{(\rt_0 + 1)^2\rho^2 (\rt_0)}{\varphi^N(\rt_0)} + \frac{N^2}{(\rt_0 + 1)^{\delta}}.
	\end{multline*}
	Therefore, since $\beta_0 > 1$, if we choose $\delta$ such that
	\begin{equation} \label{condition delta}
		\delta < 2(\beta_0 - 1),
	\end{equation}
	combined with \eqref{varphi unbounded euclic}, it implies that there exists a constant $C>0$ depending only on $\rho$ such that
	\begin{equation} \label{rt0 < infty}
		\phi \le \|\varphi\|_{C^0} \le \varphi(\rt_0) + \frac{1}{\delta} \le C.
	\end{equation}
	Hence, we conclude that \eqref{rt0 = infty}, \eqref{condition delta} and \eqref{rt0 < infty} complete the proof.
\end{proof}

\subsection{The sharpness of the rate of decay in the positive mass theorem} \label{subsection mass euc}
We next pass to another application of Theorem \ref{theorem existence on Euclidean} concerning the sign of the ADM mass in AF spaces when the  rate of decay of $k$ is critical. This result is already proven in \cite{NguyenRadialAF}, however for the sake of completeness, we would like to present it again for the convenience of interested readers. Similarly to the discussion on AH mass case in the previous section, we will restrict our attention to $\RR^n$. In this context, let $(\RR^n , g)$ be an AF manifold with 
\begin{equation} \label{mass well defined}
	g - \delta_{\text{Euc}} \in  C^{2, \alpha}_{-\frac{n-2}{2}-\eps}(\RR^n)
\end{equation}
for some $\eps >0$. The ADM mass of $(\RR^n , g)$ is defined by
$$
\madm (g) := \frac{1}{2(n-2)\omega_{n - 1} } \lim_{r \to +\infty} \int_{|x| = r} \sum_{i,j = 1}^n (g_{ij,i} - g_{ii,j})\frac{x_j}{r} d \Mcal_0^{n - 1},
$$
where $\Mcal_0^{n - 1}$ is the $(n - 1)-$dimensional Euclidean Hausdorff measure and $\omega_{n-1}$ is the volume of the standard unit sphere in $\RR^n$. In particular, when $g = \varphi^{N - 2} \delta_{\text{Euc}}$ with $\varphi$ radial, the formula becomes
\begin{equation} \label{adm mass radial}
	\madm (g) = - \tfrac{n-1}{2(n-2)} \lim_{r \to +\infty} (r^{n - 1} \varphi^\prime).
\end{equation}
The main result in Bartnik \cite{Bartnik} showed that under the decay condition \eqref{mass well defined}, the mass is a geometric invariant. A long-standing conjecture in general relativity states that the ADM mass of an AF initial data set satisfying the dominant condition is positive unless it is a Cauchy hypersurface of Minkowski space. This conjecture was proven to be true under suitable decay assumptions of $(g , k)$ at infinity, which are expected to be that
\begin{subequations}
	\begin{align}
		g - \delta_{\text{Euc}} & \in  C^{2, \alpha}_{ - p}(\RR^n),  \label{mass well defined 2a} \\
		k & \in  C^{2, \alpha}_{ - q}(\RR^n)  \label{mass well defined 2b}
	\end{align}
\end{subequations}
with $p > - (n-2)/2$ and $q > n/2$. Of course, the condition $p > - (n-2)/2$, which coincides with \eqref{mass well defined}, is necessary since it guarantees that the mass is a geometric invariant. Moreover, P. Chrusciel showed in \cite{Chruscielletter} that there exists an AF manifold with $p = - (n-2)/2$ having negative mass. Therefore, the decay exponent \eqref{mass well defined 2a} on $g$ is proven to be sharp. A natural question to ask is whether the remaining condition \eqref{mass well defined 2b} on $k$ is also sharp. In the following result, we will give an answer to this question by showing that the assumption \eqref{mass well defined 2b} is necessary and plays a crucial role in the positivity of the mass.
\begin{theorem}[Negative sign of AF mass]\label{theorem mass Euclidean}
	Let $\tau(r)$ be an arbitrary radial function in $C^{1, \alpha}_{ - \beta} (\RR^n)$ with $\alpha \in (0 , 1)$ and $ \beta \in \big( 1 , \frac{n}{2} \big)$. There exists a solution $(g , k)$ to the vacuum constraint equations \eqref{constraint} such that $(g - \delta_{\text{Euc}} , k) \in C^{3, \alpha}_{- 2 \beta + 2}(\RR^n) \times C^{1, \alpha}_{ - \beta}(\RR^n)$ and $\tr_g k = \tau(r)$. Moreover, assume that $| \tau | \sim c r^{- q}$ at infinity for some constant $c > 0$ and decay exponent $q \in ( \frac{n+2}{4}, n)$. 
	
	Then we have
	\begin{equation} \label{new decay}
		(g - \delta_{\text{Euc}} , k) \in C^{3, \alpha}_{- 2 q + 2}(\RR^n) \times C^{1, \alpha}_{ - q}(\RR^n)
	\end{equation}
	and furthermore
	\begin{itemize}
		\item if $q < \frac{n}{2}$, then $\madm (g) = - \infty$,
		\item if $q = \frac{n}{2}$ , then $- \infty < \madm (g) < 0$,
		\item if $q > \frac{n}{2}$, then $\madm (g) = 0$.
	\end{itemize}
\end{theorem}
\begin{proof}
	Since $\tau(r) \in RC^{1,\alpha}_{-\beta}(\RR^n)$, it follows by Theorem \ref{theorem smooth tiem-derivative euclid} that the vacuum conformal equations \eqref{CE with null sigma in euclic} associated with $(\rho \equiv 0, \tau(r))$ admit a radial solution $\varphi(r) > 0$ with $\varphi(r) - 1 \in C^{3,\alpha}_{-2\beta + 2}(\RR^n)$. Therefore, in view of \eqref{parametre} and Theorem \ref{theorem general}, letting
	$$
	W := \tfrac{n}{n-1} F^\p_v(r) dr
	$$
	where $v = \tfrac{n-1}{n} \varphi^N\tau^\p$ and $F_v$ is defined in \eqref{def F nonpositive Ric}, we obtain that
	\begin{subequations} \label{parameter new}
		\begin{align}
			g & = \varphi^{N - 2} \delta_{\text{Euc}} \\
			k  &= \frac{\tau}{n} \varphi^{N-2}\delta_{\text{Euc}}  + \varphi^{-2} LW
		\end{align}
	\end{subequations}
	is a solution the vacuum constraint equations \eqref{constraint} and
	$$
	(g - \delta_{\text{Euc}} , k) \in C^{3, \alpha}_{- 2 \beta + 2}(\RR^n) \times C^{1, \alpha}_{ - \beta}(\RR^n), \qquad \tr_g k = \tau(r).
	$$ 
	Hence, the first part of theorem is proven. Now, since $\psi \equiv 0$, $\rho \equiv 0$ and $V \equiv 0$, we have $I_{V,\psi,\rho,\varphi}=0$ and in regard to Theorem \ref{theorem existence on Euclidean}, the fact that the couple of radial functions $(\tau(r),\varphi(r))$ solves the vacuum equation \eqref{CE with null sigma in euclic} gives
	\begin{align*}
		|\tau| &=  \bigg| \frac{S_{V,\psi,\rho,\varphi} + 2N  \big( \varphi^\pp + (n - 1) r^{-1} \varphi^\p \big) \varphi}{2\sqrt{\varphi^N S_{V,\psi,\rho,\varphi}}} \bigg| \\
		& =  \frac{\big| 2 N (r \varphi) \big( r^{n - 1} \varphi^\prime \big)^\prime + n N r^{n - 1} ( \varphi^2 )^\prime + N^2 r^n (\varphi^\prime)^2 \big|}{2 \sqrt{\big( r^n \varphi^N \big) \big( n N r^{n - 1} ( \varphi^2 )^\prime + N^2 r^n (\varphi^\prime)^2  \big)} }.
	\end{align*}
	Simplifying this identity, we obtain
	$$
	 |\tau| =	\frac{ 2 \Big| \Big( r^{2n - 1} \big( \varphi^{ (N + 2)/2 } \big)^\prime \Big)^\prime \Big|}{(N+2) r^{3n / 2} \varphi^{N - 1}\sqrt{\varphi^\prime (r^{n - 2} \varphi)^\prime} }.
	$$
	Therefore, as long as $|\tau| \sim c r^{-q}$ at infinity, we deduce
	\begin{equation} \label{limit}
		\lim_{r \to +\infty} \Bigg( \frac{2\Big| \Big( r^{2n - 1} \big( \varphi^{ (N + 2)/2 } \big)^\prime \Big)^\prime \Big|}{(N+2) r^{\frac{3n}{2} - q} \varphi^{N - 1}\sqrt{\varphi^\prime (r^{n - 2} \varphi)^\prime} }\Bigg)= c.
	\end{equation}
	Observing that by straightforward calculations
	\begin{multline*}
		\frac{2\Big| \Big( r^{2n - 1} \big( \varphi^{ (N + 2)/2 } \big)^\prime \Big)^\prime \Big|}{(N+2)r^{\frac{3n}{2} - q} \varphi^{N - 1}\sqrt{\varphi^\prime (r^{n - 2} \varphi)^\prime} } \\
		=  \frac{(2n - 1) r^{-1} \varphi^{N/2} \varphi^\prime + N \varphi^{(N - 2)/2} (\varphi^\prime)^2 + \varphi^{N/2} \varphi^{\prime\prime}}{r^{-q} \varphi^{N-1} \sqrt{(\varphi^\prime)^2 + (n-2) r^{-1} \varphi \varphi^\prime}}.
	\end{multline*}
	Then, since $\varphi(r) - 1 \in C^{3 , \alpha}_{-\min\{2q - 2, n-2\}}(\RR^n)$, this gives
	\begin{align}\label{limit2}
		\lim_{r \to +\infty} \Bigg( \frac{ 2 \Big| \Big( r^{2n - 1} \big( \varphi^{ (N + 2)/2 } \big)^\prime \Big)^\prime \Big|}{(N+2) r^{\frac{3n}{2} - q} \varphi^{N - 1}\sqrt{\varphi^\prime (r^{n - 2} \varphi)^\prime} }\Bigg) 
		& = \tfrac{1}{\sqrt{n-2}}\lim_{r \to +\infty} \frac{\big| (r^{2n - 1} \varphi^\prime )^\prime \big|}{r^{n-1-q}\sqrt{r^{2n-1} \varphi^\prime}} \notag \\
		& = \tfrac{2(n-q)}{\sqrt{n-2}}\lim_{r \to +\infty} \Bigg| \frac{\big( \sqrt{ r^{2n - 1} \varphi^\prime} \big)^\prime}{(r^{n-q} )^\prime }\Bigg|.
	\end{align}
	Combined with \eqref{limit}, we get
	$$
	\lim_{r \to +\infty} \Bigg| \frac{\big( \sqrt{ r^{2n - 1} \varphi^\prime} \big)^\prime}{(r^{n-q})^\prime }\Bigg| = \frac{c\sqrt{n-2}}{2(n-q)}.
	$$
	Therefore, it follows from L'H\^{o}pital's rule that
	\begin{align} \label{calculate mass}
		\lim_{r \to +\infty} \Big( \frac{\varphi^\prime}{r^{- 2q + 1}} \Big) & = \Bigg(\lim_{r \to +\infty} \frac{\sqrt{r^{2n - 1}\varphi^\prime }}{r^{n - q}}\Bigg)^2 \notag \\
		& = \Bigg( \lim_{r \to +\infty} \Bigg| \frac{\big(\sqrt{r^{2n-1}\varphi^\prime} \big)^\prime}{(r^{n-q} )^\prime }\Bigg| \Bigg)^2 \notag \\
		& = \frac{c\sqrt{n-2}}{2(n-q)}.
	\end{align}
	In particular, this gives
	$$
	\varphi(r) - 1 = - \int_r^{+\infty} \varphi^\prime \, ds \in C^0_{-2 q + 2}(\RR^n)
	$$
	and hence thanks to the Lichnerowicz equation and the weighted elliptic regularity for the Laplacian, we deduce from \eqref{parameter new} that $(g - \delta_{\text{Euc}} , k) \in C^{3 , \alpha}_{- 2 q + 2}(\RR^n) \times C^{1, \alpha}_{ - q}(\RR^n)$.
	\\
	Finally, taking \eqref{calculate mass} into account in the formula \eqref{adm mass radial} of the ADM mass yields
	$$
	\madm (g) = - \tfrac{c\sqrt{n-2}}{4(n-q)} \lim_{r \to +\infty} r^{n-2q},
	$$
	which completes the proof.
\end{proof}

\begin{remark}
	As we have mentioned in Remark \ref{remark Birkhoff theorem}, since the vacuum solution $(g,k)$ constructed in the proof of Theorem \ref{theorem mass Euclidean} is spherically symmetric and regular in $\RR^+$, Birkhoff's theorem tells us that such a $(\RR^n, g, k)$ is indeed Cauchy initial data for Minkowski space.
\end{remark}

\appendix

\section{Explicit solutions on the sphere} \label{appendix construct explicit solutions}
Thanks to Theorem \ref{theorem existence on sphere}, to construct an explicit initial setting to the constraint on the sphere, it suffices to seek a set of $(V, \psi, \rho, \varphi)$ such that $S(V, \psi, \rho, \varphi) > 0$, that is
$$
N^2 (\varphi^\p)^2 - n^2 \varphi^2 + 2 n N \cot (r) \varphi \varphi^\p + \frac{1}{\sin^n(r)} \int_0^r \frac{I(V, \psi, \rho, \varphi)}{\varphi^{2N}} (\varphi^N \sin^n(s))^\p  \, ds > 0.
$$
There are many ways to find such a set as we see in the sequence of solutions in Theorem \ref{theorem instability sphere}. Another example of this construction can be given as follows. Let $\phi \in C^1(\Sp)$ be a radial function on the sphere. If we would like to assign
\begin{equation}\label{integral = phi}
	\frac{1}{\sin^n(r)} \int_0^r \frac{I(V, \psi, \rho, \varphi)}{\varphi^{2N}}(\varphi^N \sin^n(s))^\p  \, ds = \phi,
\end{equation}
then it follows
$$
I(V, \psi, \rho, \varphi)  (\varphi^N \sin^n(r))^\p = \big(\phi \sin^n(r)\big)^\p \varphi^{2N}.
$$
For simplicity, assuming that $(V, \psi)$ are zero everywhere, the equation becomes
$$
(\varphi^N \sin^n(r))^\p = \big(\tfrac{n-1}{n}\big) \frac{(\phi \sin^n(r))^\p (\varphi^N \sin^n(r))^2}{\rho^2 \sin^{2n}(r)}.
$$
Remark that this is a Bernoulli equation with respect to $y = \varphi^N \sin^n(r)$. So, we can resolve that
$$
\varphi^N = - \big(\tfrac{n-1}{n}\big) \frac{1}{\sin^{n}(r) \big(\int_{\pi/2}^r \frac{(\phi \sin^n(s))^\p}{\rho^2 \sin^{2n}(s)} \, ds + c_1 \big)},
$$
for some constant $c_1$ depending on $\rho$. In particular, if we take
\begin{equation}\label{rho = phi}
	\rho = \phi,
\end{equation}
we can easily compute
$$
\varphi^N =  \big(\tfrac{n}{n-1}\big) \frac{\rho}{1 - c \rho \sin^n(r)},
$$
with 
$$
c := c_1 + \frac{1}{\rho(\pi/2)}.
$$
By a straightforward calculation, we get
$$
\varphi^\p = \frac{1}{N}\bigg(\frac{n}{n -1 }\bigg)^{1/N} \frac{\rho^\p + nc \rho^2 \cos(r)\sin^{n-1}(r)}{\rho^{1 - 1/N}\big(1 - c\rho \sin^n(r)\big)^{1 + 1/N} }.
$$
Therefore, as long as we choose $\rho$ and $c_1$ such that
\begin{itemize}
	\item[-] $\min \rho > 0$ is large,
	\item[-] $c (\max\rho)$ and $|\rho^\p|$ are small,
	\item[-] $\cot(r) \rho^\p$ is small or positive,
\end{itemize}
it follows that $\varphi \sim \big(\tfrac{n}{n - 1}\big)^{1/N} \rho^{1/N}$ and $\cot(r)\varphi^\p\varphi$ is either positive or negative and small. Hence, by definition, we get from \eqref{integral = phi} and \eqref{rho = phi} that
$$
S(V,\psi, \rho, \varphi) \ge - 2 n^2 \rho^{2/N} + \rho > 0,
$$
which is our desire.

\small
\bibliographystyle{plain}

\bigskip

\begin{flushleft}
	Philippe Castillon \\
	\textsc{imag} (\textsc{u.m.r. c.n.r.s.} 5149) \\
	Univ. de Montpellier \\
	34095 \textsc{Montpellier} Cedex 5, France \\
	\texttt{philippe.castillon\symbol{64}umontpellier.fr}
\end{flushleft}

\medskip

\begin{flushleft}
	Cang Nguyen-The \\
	Faculty of Mathematics and Statistics,\\
	Ton Duc Thang University, Ho Chi Minh City, Vietnam \\
	\texttt{alpthecang\symbol{64}gmail.com}
\end{flushleft}

\end{document}